\documentclass[a4paper, USenglish, cleveref, autoref]{lipics-v2019}
\nolinenumbers
\hideLIPIcs

\RequirePackage{xcolor}
\hypersetup{
  colorlinks,
  linkcolor={blue!50!black},
  citecolor={blue!50!black},
  urlcolor={black}
}

\renewcommand\path[1]{{\normalfont\small\detokenize{#1}}}

\usepackage{tikz}
\usetikzlibrary{decorations.pathreplacing,angles,quotes,shapes,decorations.pathmorphing}

\usepackage{booktabs}
\usepackage{stmaryrd}

\allowdisplaybreaks[1]

\theoremstyle{plain}
\newtheorem{fact}[theorem]{Fact}
\newtheorem{observation}[theorem]{Observation}

\newcommand{\cc}[1]{\ensuremath{\mathrm{#1}}}

\newcommand{\op}[1]{\ensuremath{\operatorname{#1}}}

\newcommand{\W}{\cc{W[1]}}
\newcommand{\Wtwo}{\cc{W[2]}}
\newcommand{\Atwo}{\cc{A[2]}}

\newcommand{\poly}{\op{poly}}

\newcommand{\N}{\mathbb{N}}

\usepackage{xparse}
\DeclareDocumentCommand{\restrict}{O{}}{\mathord{\restriction}_{#1}}

\newcommand{\homs}[2]{\ensuremath{\mathsf{Hom}(#1 \to #2)}}
\newcommand{\surj}[2]{\ensuremath{\mathsf{Surj}(#1 \to #2)}}
\newcommand{\cfhoms}[2]{\ensuremath{\mathsf{cf}\textsf{-}\mathsf{Hom}(#1 \to #2)}}
\newcommand{\cphoms}[2]{\ensuremath{\mathsf{cp}\textsf{-}\mathsf{Hom}(#1 \to #2)}}
\newcommand{\partinj}[2]{\ensuremath{\mathsf{PartInj}(#1 \to #2)}}

\newcommand{\homsprob}{\ensuremath{\mathsf{Hom}}}
\newcommand{\cfhomsprob}{\ensuremath{\mathsf{cf}\textsf{-}\mathsf{Hom}}}
\newcommand{\cphomsprob}{\ensuremath{\mathsf{cp}\textsf{-}\mathsf{Hom}}}
\newcommand{\partinjprob}{\ensuremath{\mathsf{PartInj}}}
\newcommand{\exclass}{\ensuremath{\Delta}}

\newcommand{\bs}[1]{ \ensuremath{ \left\lbrace #1 \right\rbrace } }

\newcommand{\qgraphof}[1]{Q[#1]}

\newcommand{\labeled}{\text{query }}

\newcommand{\arity}{\mathsf{a}}

\newcommand{\mcA}{\mathcal{A}}
\newcommand{\mcB}{\mathcal{B}}

\newcommand{\mcF}{\mathcal{F}}
\newcommand{\mcG}{\mathcal{G}}
\newcommand{\mcH}{\mathcal{H}}
\newcommand{\aug}{\mathsf{aug}}
\newcommand{\raug}{E_\aug}
\newcommand{\acore}{\mathsf{core}_\mathsf{aug}}
\newcommand{\lmn}{\mathsf{lmn}} 
\newcommand{\supp}{\mathsf{supp}}

\newcommand{\pmc}{\mathsf{p}\text{-}\mathsf{MC}}
\newcommand{\kite}{\text{grate}}
\newcommand{\kkite}[1]{#1\text{-}\kite}

\usepackage{mathtools}

\DeclarePairedDelimiter\abs{\lvert}{\rvert}
\DeclarePairedDelimiter\set{\{}{\}}

\DeclarePairedDelimiterX\setc[2]{\{}{\}}{\,#1 \;|\; #2\,}
\DeclarePairedDelimiterX\parenc[2]{\lparen}{\rparen}{\,#1 \;\delimsize\vert\; #2\,}

\title{Counting Answers to Existential Questions}

\author{Holger Dell}{Cluster of Excellence (MMCI), Saarland Informatics Campus (SIC), Saarbrücken, Germany}{hdell@mmci.uni-saarland.de}{https://orcid.org/0000-0001-8955-0786}{}

\author{Marc Roth}{Cluster of Excellence (MMCI), Saarland Informatics Campus (SIC), Saarbrücken, Germany}{mroth@mmci.uni-saarland.de}{https://orcid.org/0000-0003-3159-9418}{}

\author{Philip Wellnitz}{Max Planck Institute for Informatics, Saarland Informatics Campus (SIC),
Saarbrücken, Germany}{wellnitz@mpi-inf.mpg.de}{https://orcid.org/0000-0002-6482-8478}{Partially funded by the Saarbrücken Graduate School of Computer Science.}

\authorrunning{H. Dell, M. Roth, and P. Wellnitz}

\Copyright{Holger Dell, Marc Roth, and Philip Wellnitz}



\acknowledgements{We thank Cornelius Brand, Karl Bringmann, Radu Curticapean, Reinhard Diestel, Joshua Erde, Stephan Kreutzer, Stefan Mengel, Daniel Weißauer, and an anonymous reviewer for discussions and advice. In particular, we are very grateful to Joshua and Daniel for pointing out the reference~\cite{routedminors}.}

\EventEditors{John Q. Open and Joan R. Access}
\EventNoEds{2}
\EventLongTitle{42nd Conference on Very Important Topics (CVIT 2016)}
\EventShortTitle{CVIT 2016}
\EventAcronym{CVIT}
\EventYear{2016}
\EventDate{December 24--27, 2016}
\EventLocation{Little Whinging, United Kingdom}
\EventLogo{}
\SeriesVolume{42}
\ArticleNo{23}

\begin{document}

\maketitle

\begin{abstract}
	Conjunctive queries select and are expected to return certain tuples from a relational database.
	We study the potentially easier problem of \emph{counting} all selected tuples, rather than enumerating them.
	In particular, we are interested in the problem's parameterized and data complexity, where the query is considered to be small or even fixed, and the database is considered to be large. We identify two structural parameters for conjunctive queries that capture their inherent complexity: The dominating star size and the linked matching number.
	If the \emph{dominating star size} of a conjunctive query is large, then we show that counting solution tuples to the query is at least as hard as counting dominating sets, which yields a fine-grained complexity lower bound under the Strong Exponential Time Hypothesis (SETH) as well as a $\#\Wtwo$-hardness result in parameterized complexity.
  Moreover, if the \emph{linked matching number} of a conjunctive query is large, then we show that the structure of the query is so rich that arbitrary queries up to a certain size can be encoded into it; in the language of parameterized complexity, this essentially establishes a $\#\Atwo$-completeness result.

  Using ideas stemming from Lovász (1967)\nocite{lovasz}, we lift complexity results from the class of conjunctive queries to arbitrary existential or universal formulas that might contain inequalities and negations on constraints over the free variables.
  As a consequence, we obtain a complexity classification that refines and generalizes previous results of Chen, Durand, and Mengel (ToCS 2015; ICDT 2015; PODS 2016)\nocite{DurandM15,ChenM15,ChenM16} for conjunctive queries and of Curticapean and Marx (FOCS 2014)\nocite{CurticapeanM14} for the subgraph counting problem.
  Our proof also relies on graph minors, and we show a strengthening of the Excluded-Grid-Theorem which might be of independent interest: If the linked matching number (and thus the treewidth) is large, then not only can we find a large grid somewhere in the graph, but we can find a large grid whose diagonal has disjoint paths leading into an assumed node-well-linked set.
\end{abstract}

\section{Introduction}
Conjunctive query evaluation is a core problem in database theory.
Using first-order logic, conjunctive queries can be expressed by formulas of the form
\begin{equation}\label{eq:conjunctivequeries}
	x_1\dots x_k
	\exists y_1\dots \exists y_\ell
	( a_1\wedge\dots\wedge a_m )\,,
\end{equation}
where
the $x_i$ are the \emph{free variables},
the $y_i$ are the (existentially) \emph{quantified variables}, and
the $a_i$ are \emph{atomic formulas} (such as edge $E(x_1,y_4)$ or relational $R(x_7,y_3,y_6)$
constraints on the variables).
Conjunctive queries exactly correspond to select-project-join queries; a detailed introduction can be found in the textbook of Abiteboul, Hull, and Vianu~\cite{AbiteboulHV95}.
The \emph{conjunctive query evaluation problem} is given a conjunctive query and a relational database, and is tasked to compute the set of all assignments to the free variables such that the formula is satisfied.
Since enumerating all solution tuples $s_1\dots s_k$ can be costly for reasons not inherent to the problem's complexity, it is more meaningful to consider the decision problem (Does there exist a solution tuple?) or the more general counting problem (How many solution tuples exist?).
The decision problem is equivalent to setting~$k=0$ and also called the constraint satisfaction problem (CSP).
In this paper, we study the problem of counting the number of all solution tuples for conjunctive and more general queries.

Perhaps the most na\"ive way to study the complexity of this problem is via its \emph{combined complexity}, in which both the query and the database are considered to be worst-case inputs.
Since conjunctive queries generalize the clique problem on graphs, the problem is clearly $\cc{NP}$-hard in this setting~\cite{ChandraM77}.
In the real world, however, the database is much larger than the query, and thus the combined complexity may fixate on instances that we do not care about.
Instead, we consider two other models in this paper: the {data complexity} and the {parameterized complexity} of conjunctive query evaluation.

The \emph{data complexity} considers the query to be completely fixed and only the database to be worst-case input.
If the query is fixed, the number of variables $k+\ell$ is a constant, and so the problem is polynomial-time solvable: even the exhaustive search algorithm just needs to try out and check all $n^{k+\ell}$ possible assignments to the variables, where~$n$ is the size of the universe.
Unsurprisingly, exhaustive search is not the best strategy for every query.
For example, Chekuri and Rajaraman~\cite{ChekuriR00} showed that the decision and counting problems can be solved in time $O(n^{t+1})$ where $t$ is the treewidth of the query's Gaifman graph, that is, the graph containing a vertex for every variable and an edge between two vertices whenever the corresponding variables are contained in a common constraint.
Since $t+1$ is typically much smaller than~$k+\ell$, this algorithm is better than exhaustive search.
For each fixed query~$Q$, the guiding question for a fine-grained understanding of data complexity is this:
What is the smallest constant~$c_Q$ such that the query evaluation problem can be solved in time $O(n^{c_Q})$?

\emph{Parameterized complexity} offers a third vantage point from which conjunctive query evaluation can be studied.
Here the query isn't completely fixed, but it's also not completely free either.
Instead, it is assumed that only certain types of queries will be used, meaning that the class of queries that are allowed as input is restricted.
As a hybrid between data complexity and combined complexity, the parameterized complexity of query evaluation is more difficult to study than the combined complexity, but easier than the data complexity, while still offering some insight.
For example, Grohe, Schwentick, and Segoufin~\cite{GroheSS01} used the parameterized complexity approach to show that the treewidth of the Gaifman graph is a necessary structural parameter for determining the complexity of the problem in the sense that there are queries of unbounded treewidth whose query evaluation problem cannot be solved in polynomial time unless the assumption $\cc{FPT}=\cc{W[1]}$ from parameterized complexity fails.

\subsection{Context and Previous Work}\label{sec:prevwork}

When only one constraint type~$E$ of arity two is allowed, the conjunctive query evaluation problem specializes to the graph homomorphism problem: The decision problem (where $k=0$) is given two graphs~$H,G$ to decide whether there is a homomorphism from~$H$ to~$G$.
Dalmau et al.~\cite{DBLP:conf/cp/DalmauKV02} prove that this problem can be solved in polynomial time if the homomorphic core of~$H$ has bounded treewidth, and conversely, Grohe~\cite{homsdicho2} shows that the graph homomorphism problem is~\W-complete even if~$H$ is restricted to be from an arbitrary class of graphs whose homomorphic cores have unbounded treewidth.
Taken together, these two results yield a \emph{dichotomy theorem} for the complexity of detecting graph homomorphisms: Depending on the class of allowed graphs~$H$, the problem is either polynomial-time computable or \W-complete, and in particular there are not infinitely many cases of intermediate complexity.
For the counting problem without quantified variables (where~$\ell=0$), such a dichotomy is also known: Dalmau and Jonsson~\cite{homsdicho1} show that the number of homomorphisms from~$H$ to~$G$ is polynomial-time computable if~$H$ itself has bounded treewidth, and it is~$\#\W$-complete if~$H$ comes from any class of unbounded treewidth.
In the mixed situation when both free and quantified variables may exist (and thus~$k,\ell>0$), then the resulting counting problem actually counts \emph{partial homomorphisms}, that is, homomorphisms from~$k$ vertices of~$H$ that can be extended to a homomorphism on all~$k+\ell$ vertices of~$H$.
A line of work~\cite{PichlerS13,Mengel13}, culminating in Durand and Mengel~\cite{DurandM15} and Chen and Mengel~\cite{ChenM15}, studies the parameterized complexity of this mixed problem, and depending on the class of graphs~$H$ that are allowed, they classify the complexity either as polynomial-time, $\W$-equivalent, or~$\#\W$-hard.
A corollary to the present work is a finer classification that splits up the $\#\W$-hard cases into three classes.

One way to go beyond homomorphisms is to consider \emph{injective} homomorphisms, which leads to the corresponding decision problem that is given~$H,G$ to decide whether~$H$ is a subgraph of~$G$ -- this problem can be solved in time~$f(H) n^{O(t)}$ if~$t$ is the treewidth of~$H$ (e.g.,~\cite{DBLP:journals/jcss/FominLRSR12}), that is, it is \emph{fixed-parameter tractable} when parameterized by~$\abs{H}$ and if the treewidth is bounded.
However, it is an important open problem~\cite{marx2007can} whether the subgraph detection problem is \W-hard when~$H$ is restricted to be from an arbitrary class of unbounded treewidth.
The counting problem is better understood: Vassilevska Williams and Williams~\cite{williams2013finding} (also cf.\ \cite{DBLP:journals/siamdm/KowalukLL13,CurticapeanM14}) show that the number of times~$H$ occurs as a subgraph in~$G$ can be computed
in time $f(H)n^{\mathsf{vc}(H)+O(1)}$ where $\mathsf{vc}(H)$ is the size of the smallest vertex cover, but Curticapean and Marx~\cite{CurticapeanM14} (also cf.~\cite{hombasis2017}) show that the problem is $\#\W$-complete if $H$ is from any class of graphs whose minimum vertex cover is not bounded.
Now, what do injective homomorphisms have to do with conjunctive queries?
As it turns out, what we are doing is to add \emph{inequalities} as an additional, but very restricted constraint type: Injective homomorphisms correspond to queries without quantified variables that have edge constraints and are augmented with inequalities $(x_i\ne x_j)$ for all distinct~$i,j$.
If some, but not all, inequality constraints are present, we obtain partially injective homomorphisms, the complexity of which has a known dichotomy theorem for the counting version~\cite{Roth17}, and has been studied to some extent for the decision version~\cite{Klug88}.
As part of the present work, we are able to classify the mixed situation with free and quantified variables ($k,\ell>0$) \emph{as well as} some inequalities on the free variables.

The mentioned complexity classification for counting partial homomorphisms into three cases~\cite{DurandM15,ChenM15} was actually proved in the more general setting of conjunctive queries.
Chen and Mengel~\cite{ChenM16} extended their classification to queries that are monotone, but not necessarily conjunctive. That is, the corresponding formula is supposed to be an existential positive formula, which may contain existential quantifiers~$\exists$, logical ands $\wedge$, and ors $\vee$.\linebreak
In the present work, we are able to further extend (our finer version of) the classification to existential formulas that may have negations on constraints involving only free variables;
we truly study the complexity of \emph{counting answers to existential questions}.

\subsection{Our contributions}
As already indicated in Section~\ref{sec:prevwork}, we make simultaneous progress on two fronts: Our complexity classifications are finer than previous work, and we can prove the classification for more general classes of queries.
An important feature of our work is that the proofs are modular and largely self-contained: We first prove the complexity results for counting partial homomorphisms, then lift them to conjunctive queries, and then further to a more general class of queries.
So what is the most general class of queries that we study?
We allow queries~$\varphi$ of the form
\begin{equation}\label{eq:generalqueries}
	x_1\dots x_k\exists y_1\dots \exists y_\ell:\,\psi\,,
\end{equation}
where $\psi$ is a quantifier-free formula in first-order logic and all negations in~$\psi$ must be directly applied to constraints that only involve free variables (e.g.~$E(x_1,x_7) \vee (R(x_7,y_7,y_9)\wedge\neg R(x_1,x_4,x_9))$). Constraints of the form $\neg R(x_1,x_4,x_9)$ are referred to as \emph{non-monotone constraints} in the remainder of the paper. Furthermore $\varphi$ may be equipped with a set of inequalities over the free variables (eg.~$x_3\ne x_5$).

All of our theorems also apply to the corresponding \emph{universal} queries, where each~$\exists$ in~\eqref{eq:generalqueries} is replaced with~$\forall$, but for the sake of readability we will often omit this fact.
We are able to generalize from conjunctive queries to queries of the form~\eqref{eq:generalqueries} by using ideas that go back to Lovász's work from~1967~\cite{lovasz} (also cf.~\cite{lovaszbook}):
We prove that queries~$\varphi$ of the form \eqref{eq:generalqueries} can be expressed in a meaningful way as an abstract linear combination of conjunctive queries (which are of the form~\eqref{eq:conjunctivequeries}); positive results (algorithms) as well as negative results (hardness) for each ``summand'' translate to the abstract linear combination and thus to~$\varphi$.

\paragraph*{Data Complexity}
To study the data complexity of the problem, we employ the Strong Exponential Time Hypothesis (SETH) by Impagliazzo and Paturi~\cite{IP01}, which was developed in the context of fine-grained complexity.
The $k$-dominating set problem can be easily expressed as a (universal) conjunctive query, and Williams and Pătraşcu~\cite{DBLP:conf/soda/PatrascuW10} show that this problem cannot be solved in time~$O(n^{k-\varepsilon})$ unless SETH is false.
We are able to lift this hardness result to all queries~$\varphi$ that have the $k$-dominating set query as a \emph{query minor}, a notion that we translate from graphs and formalize later.
The \emph{dominating star size} $\mathsf{dss}(\varphi)$ of a conjunctive query~$\varphi$ is the maximum number~$k$ such that the~$k$-dominating set query is a query minor.
Equivalently, this means that some connected component in the quantified variables of~$\varphi$ has~$k$ neighbors in the free variables.\footnote{The dominating star size coincides with the \emph{strict star size} from \cite{ChenM15}.}
We obtain the following result:

\begin{theorem}
	\label{thm:data complexity}
	Let~$\varphi$ be a fixed query of the form~\eqref{eq:generalqueries}.
	Given a database~$B$ with a domain of size $n$, we wish to compute the number of solutions of~$\varphi$
	in~$B$.
	If SETH holds, this problem
	cannot be solved in time $O(n^{\mathsf{dss}(\varphi)-\varepsilon})$ for any
	$\varepsilon>0$.
\end{theorem}
\pagebreak

\noindent In Remark~\ref{rem:running_time_graphs}, we also obtain an algorithm for the problem in~Theorem~\ref{thm:data complexity}, with a running time of~$O(n^{\mathsf{dss}(\varphi)+t+1}+n^{t'+1})$, where~$t$ and~$t'$ are treewidths related to the query~$\varphi$.
Neglecting many technical details, the proof of Theorem~\ref{thm:data complexity} reduces the $k$-dominating set problem to the model counting problem for~$\varphi$ by following operations of the query minor.
If $\varphi$ is a query of the form~\eqref{eq:generalqueries}, then it can be represented by an abstract linear combination of conjunctive queries~$\varphi'$; in this case, we define $\mathsf{dss}(\varphi)$ as the maximum $\mathsf{dss}(\varphi')$ over all constituents~$\varphi'$ that occur in this abstract linear combination.
We formalize the notion in Section~\ref{sec:4}.

Theorem~\ref{thm:data complexity} is similar in spirit to other known conditional lower bounds for first-order model checking, such as the one of Williams~\cite{Williams14} and Gao et al.~\cite{GaoIKW17}.
One of their results is that first-order sentences with~$k+1$ variables cannot be decided in time $O(m^{k-\varepsilon})$, where~$m$ is the size of the structure, unless SETH fails.
However, these results are incomparable to Theorem~\ref{thm:data complexity} for several reasons:
The results in \cite{Williams14,GaoIKW17} allow negations and consider the decision problem, while we allow only limited negations and consider the counting problem.
More fundamentally, however, Theorem~\ref{thm:data complexity} gives a hardness result for every fixed query~$\varphi$, while the results in~\cite{Williams14,GaoIKW17} show that there exists a query~$\varphi$ that is hard. Moreover, the lower bounds in~\cite{Williams14,GaoIKW17} are in terms of the size~$m$ of the structure, not merely the size~$n$ of the domain.

\paragraph*{Parameterized Complexity}
We refine the classification of Chen and Mengel~\cite{ChenM15} for counting answers to conjunctive queries. For every class of allowed queries they show the problem to be either fixed-parameter tractable, $\W$-equivalent or $\#\W$-hard.
Here, $\W$-equivalent means that there are parameterized Turing reductions from and to the decision version of the $k$-Clique problem.
Understanding the parameterized complexity of problems even beyond the usual classes~$\W$ and~$\#\W$ is interesting from a structural complexity point of view, and it also provides meaningful information about the studied problem. Indeed we show that the dominating star size, i.e., the parameter considered in Theorem~\ref{thm:data complexity}, is a structural parameter for conjunctive queries that, if unbounded, makes the problem $\#\Wtwo$-hard and that, if bounded, keeps the problem $\#\W$-easy.

This extension to $\#\Wtwo$-hard cases only partially resolves the parameterized complexity of the problem of counting answers to conjunctive queries. It is known that the general problem of counting answers to formulas of the form
\begin{equation}\label{eq:queryA2}
	x_1\dots x_k\exists y_1\dots \exists y_\ell:\,\psi\,, \qquad\text{ where } \psi \text{ is a quantifier-free first-order formula,}
\end{equation}
is $\#\Atwo$-equivalent.\footnote{Due to a technicality in the original definition of~$\#\Atwo$, we cannot establish $\#\Atwo$-completeness and will instead only talk about \emph{equivalence} to a $\#\Atwo$-complete problem under parameterized Turing reductions (see Section~\ref{sec:a2} for details).}
For which families of \emph{conjunctive} queries is the counting problem as hard as for unrestricted queries as in~\eqref{eq:queryA2}? Such families have the hardest counting problems, even harder than the $\#\Wtwo$-hard cases unless $\#\Atwo = \#\Wtwo$ holds, which seems unlikely.\footnote{See Chapt.~8 and~14 in~\cite{flumgrohe} for a discussion.}
We prove that families of conjunctive queries are $\#\Atwo$-hard if their \emph{linked matching number} is unbounded. Intuitively a conjunctive query $\varphi$ with free variables $X$ and quantified variables $Y$ has a large linked matching if there is a large well-linked set in $Y$ that cannot be separated from $X$ by removing a small number of variables. It is formally defined in Section~\ref{sec:4}.
We obtain the following refined complexity classification.

\begin{theorem}\label{thm:extension}
	Let $\Phi$ be a family of conjunctive queries. Given a formula $\varphi$ from $\Phi$ and a database~$B$, we wish to compute the number of solutions of $\varphi$ in $B$.
	When parameterized by $|\varphi|$ this problem is
	\begin{enumerate}
		\item $\#\W$-easy if the dominating star size of $\Phi$ is bounded,
		\item $\#\Wtwo$-hard if the dominating star size of $\Phi$ is unbounded, and
		\item $\#\Atwo$-equivalent if the linked matching number of $\Phi$ is unbounded.
	\end{enumerate}
\end{theorem}

It is instructive to provide examples for the application of the above theorem. First consider the problem of, given a graph $G$ without self-loops and a natural number $k$, computing the number of cliques of size $k$ that are not maximal. While the problem of counting cliques of size $k$ is $\#\W$-complete, adding the non-maximality constraint makes the problem hard for $\#\Wtwo$. To see this, we will express the problem as a conjunctive query
\begin{equation}
	\varphi_k := x_1\dots x_k \exists y: \bigwedge_{1\leq i<j \leq k} E(x_i,x_j) \wedge \bigwedge_{1\leq i\leq k} E(x_i,y) \,.
\end{equation}
Note that the number of solutions to $\varphi_k$ in $G$ is precisely $k!$ times the number of non-maximal cliques of size $k$ in $G$. Furthermore, it holds that $\varphi_k$ has dominating star size $k$ and hence that $\Phi=\{\varphi_k~|~k\in \mathbb{N}\}$ has unbounded dominating star size. By Theorem~\ref{thm:extension} the problem of counting answers to queries in $\Phi$ is $\#\Wtwo$-hard. Furthermore, invoking Theorem~\ref{thm:data complexity}, we obtain that counting non-maximal cliques of size $k$ cannot be done in time $O(n^{k-\varepsilon})$ for any $\varepsilon > 0$. Note that this is also in sharp contrast to the problem of counting (not necessarily non-maximal) cliques of size $k$ which can be done in time $O(n^{\omega k/3})$ \cite{nevsetvril1985complexity}. Furthermore deciding the existence of a non-maximal clique of size $k$ is equivalent to deciding the existence of a clique of size $k+1$ and hence the lower bound under SETH crucially depends on the fact that we count the solutions.

On the other hand, counting non-maximal cliques of size $k$ is most likely not $\#\Atwo$-hard as it is $\#\Wtwo$-easy\footnote{If there is a constant bound on the number of quantified variables then the problem of counting answers to conjunctive queries is reducible to a $\#\Wtwo$-complete problem w.r.t. parameterized Turing reductions. We omit a proof of this statement but point out that it can be done by lifting the results of Chapt.~7.4 in~\cite{flumgrohe} to the realm of counting problems.}. An example for a $\#\Atwo$-hard problem would be the following. Assume a graph $G$ and a natural number $k$ are given. Then the goal is to compute the number of $k$-vertex sets that can be (perfectly) matched to a $k$-clique. Let us express the problem as a conjunctive query
\begin{equation}
	\psi_k := x_1\dots x_k \exists y_1\dots \exists y_k: \bigwedge_{1\leq i<j \leq k} E(y_i,y_j) \wedge \bigwedge_{1\leq i\leq k} E(x_i,y_i) \,.
\end{equation}
We point out that $\psi_k$ does not correspond directly to the vertex sets we would like to count as $x_i$ and $x_j$ could be the same vertex in $G$. However, it can be shown along the lines of Section~\ref{sec:hard_w_2} that an oracle for counting answers to $\psi_k$ allows us to compute the desired number efficiently and vice versa. Finally, as the linked matching number of $\psi_k$ is not bounded for $k\rightarrow \infty$, $\#\Atwo$-hardness follows from Theorem~\ref{thm:extension}.

Building up on Theorem~\ref{thm:extension} and using the framework of linear combinations, we obtain the following, extensive classification result.
\pagebreak

\begin{theorem}
	\label{thm:main2_int}
	Let $\Phi$ be a family of existential or universal positive formulas with
	inequalities and non-monotone constraints, both over the free variables.
	Given a formula~$\varphi$ from~$\Phi$ and a database~$B$, we wish to compute the number of solutions of $\varphi$ in $B$.

	When parameterized by $|\varphi|$, this problem is either fixed parameter
	tractable, $\W$-equivalent, $\#\W$-equivalent, $\#\Wtwo$-hard or $\#\Atwo$-equivalent.
\end{theorem}

Note that allowing the inequalities and non-monotone constraints over all variables, not just the free ones, would in particular include the subgraph decision problem.
However, the parameterized complexity of finding a subgraph in~$G$ that is
isomorphic to a small pattern graph~$P$ is a long-standing open question in
parameterized complexity.

\subsection{Techniques and Overview}
Our paper brings together questions and techniques from a wide variety of areas, such as parameterized and fine-grained complexity, logics, database theory, matroid theory, lattice theory, graph minor theory, and the theory of graph limits.
The interested reader should not be alarmed, however, as we put considerable effort into making the presentation as self-contained and smooth as possible, introducing the required background material carefully and only once needed:
After reviewing some basic preliminaries (Section~\ref{sec:prelims}), we establish our refined complexity classification in Sections~\ref{sec:minors_and_colors}--\ref{sec:lin_combs} for the special case of partial graph homomorphisms, rather than the full query evaluation problem.
Only in Section~\ref{sec:structures} do we introduce the notation necessary to deal with arbitrary logical structures, and we lift or generalize the results from the previous sections to this case.
Section~\ref{sec:a2} contains some technical material, which transfers a logical normalization theorem from decision to counting.

\paragraph*{Colors and Query minors}
We will mainly work with a color-prescribed variant of the problem of counting answers to conjunctive queries. Here we assume that the elements of a given database $B$ are colored according to the variables of the given conjunctive query $\varphi$ and the goal is to compute the number of solutions that are additionally color-preserving. For this variant we will show and heavily exploit that the problem of counting answers to a conjunctive query $\varphi$ is at least as hard as counting answers to any query that is a minor of $\varphi$. Minors of a query are defined via the (graph theoretic) minors of its Gaifman graph.
It is then required to show that the color-prescribed variant and the uncolored variant are interreducible for all minimal queries. Intuitively, a query is minimal if it does not contain a proper subquery that produces the same set of solutions for each database. The proof of the interreducibility relies on the theory of homomorphic equivalence.

For Theorem~\ref{thm:data complexity} and the second case of Theorem~\ref{thm:extension} we construct a tight reduction from the problem of counting dominating sets of size $k$ which cannot be solved in time $O(n^{k-\varepsilon})$ for any $\varepsilon > 0$ unless SETH fails~\cite{DBLP:conf/soda/PatrascuW10} and which is hard for $\#\Wtwo$~\cite{flumgrohe_counting}.

\paragraph*{Minor theory}
For $\#\Atwo$-hardness in Theorem~\ref{thm:extension} we take a detour to graph minor theory. In particular we strengthen the Excluded-Grid-Theorem, which might be of independent interest:
\begin{theorem}[Intuitive version]\label{thm:ex_grids_intro}
  There exists an unbounded function $f$ such that every graph containing a node-well-linked set~$S$ of~$k$ vertices has an $(f(k)\times f(k))$-grid minor with the property that every vertex in the first column of the grid is the preimage of at least one element of $S$ with respect to the minor mapping.
\end{theorem}\pagebreak

\noindent Using Theorem~\ref{thm:ex_grids_intro} we show that every conjunctive query with a large linked matching number contains a minor in which the free variables can be matched to a set of quantified variables that are connected in a grid-like manner. Those queries are then shown to be $\#\Atwo$-hard building up on a $\#\Atwo$-normalization theorem which we will provide at the end of the paper.

\paragraph*{Abstract linear combinations}
To prove Theorem~\ref{thm:main2_int}, we use abstract linear combinations that are called \emph{quantum graphs} (or rather, quantum queries in our setting) and were developed in the theory of graph limits~\cite{lovaszbook}.
For our computational questions, the \emph{complexity monotonicity property}~\cite{hombasis2017} is the useful phenomenon that the quantum graph and its constituents (i.e., its abstract summands) often lead to computational problems that have precisely the same complexity.
Using elementary linear-algebraic and polynomial interpolation arguments, we prove that this property holds (Lemma~\ref{lem:homs linearly independent}), and we use Rota's NBC Theorem from lattice theory~\cite{rota1964foundations} to determine which graphs are constituents of the relevant quantum graphs.
The complexity monotonicity property has been used (implicitly) by Chen and Mengel~\cite{ChenM16} for their extension from conjunctive queries to monotone queries; and the extension from homomorphisms to partially injective homomorphism~\cite{Roth17} used Rota's Theorem in a similar fashion as we do in the present work.

\section{Preliminaries}\label{sec:prelims}
We use the notation $[n]=\{1,\dots,n\}$ and $[m,n]=\set{m,\dots,n}$ for natural numbers with $m<n$.
We write $\#M$ for the cardinality of a finite set~$M$.
We write $f|_M$ for the restriction of a function~$f$ to elements of $M$.
For a function $f\colon A \times B \to C$ and $a \in A$, we write $f(a,\star)$ for the function $b\mapsto f(a,b)$.

\paragraph*{Graphs, homomorphisms and formulas}
Graphs in this paper are unlabeled, undirected, simple and without self-loops, unless stated otherwise.
Let $V(G)$ denote the set of vertices and $E(G)$ denote the set of edges of $G$.
We define the \emph{size} of a graph $G$ to be the number of vertices. Given a subset $Y$ of $V(G)$, we write $G[Y]$ for the subgraph induced by the vertices of $Y$. The \emph{complement graph $\overline{G}$} has the same vertices as $G$ and contains an edge $uv$ if and only if $u \neq v$ and $uv\notin E(G)$. A \emph{homomorphism} $h$ from a graph $F$ to a graph~$G$ is a mapping from $V(F)$ to $V(G)$ that is edge-preserving, that is, all $uv \in E(F)$ satisfy $h(u)h(v)\in E(G)$.
We write $\homs{F}{G}$ for the set of all homomorphisms from~$F$ to~$G$.
A bijective homomorphism whose inverse is also a homomorphism is called an \emph{isomorphism}, and a homomorphism from $F$ to $F$ itself is called \emph{endomorphism}. An endomorphism that is also an isomorphism is called an \emph{automorphism}. We write $\mathsf{Aut}(F)$ for the set of all automorphisms of $F$.

\paragraph*{Parameterized counting complexity}
A \emph{counting problem} is a function~$P\colon\{0,1\}^\ast\to\N$, and a \emph{parameterized counting problem} is a pair $(P,\pi)$ where $\pi\colon\{0,1\}^\ast\to\N$ is computable and called a \emph{parameterization}. Parameterized \emph{decision} problems are defined likewise for decision problems~$P\colon\{0,1\}^\ast\to\bs{0,1}$. A parameterized (decision or counting) problem is \emph{fixed-parameter tractable} if there is a computable function $t\colon\N\to\N$ such that, for every input $x\in\{0,1\}^\ast$, the function~$P$ can be computed in time $t(\pi(x)) \cdot \mathsf{poly}(|x|)$. We denote the class of all fixed-parameter tractable problems as~$\cc{FPT}$.

\noindent A \emph{parameterized Turing-reduction} from $(P,\pi)$ to $(P',\pi')$ is an algorithm $\mathbb{A}$ with oracle access to $P'$ that solves $P$, such that $\mathbb{A}$ runs in fixed-parameter tractable time when parameterized by~$\pi$ and there exists a computable function $r$ such that, for every input $x$, the parameter $\pi'(y)$ of every query~$y$ is bounded by $r(\pi(x))$.
A \emph{parameterized parsimonious reduction} is a parameterized Turing-reduction with the additional requirement that $\mathbb{A}$ is only allowed to query the oracle a single time at the very end of the computation and then outputs the result of the query without further modification.

$\mathsf{Clique}$ is the parameterized (decision) problem to decide whether a given graph~$G$ contains a $k$-clique.
Similarly, $\mathsf{DomSet}$ is to decide whether~$G$ has a dominating set of size $k$.
The parameterized counting problems $\#\mathsf{Clique}$ and $\#\mathsf{DomSet}$ count the number of the respective objects.
We define the parameterized complexity classes that appear in this paper by their well-known complete problems:
$\W$ contains all parameterized problems that are reducible to $\mathsf{Clique}$ with respect to parameterized parsimonious reductions.
Similarly, $\#\W$, $\Wtwo$, and $\#\Wtwo$ contain all problems reducible to $\#\mathsf{Clique}$, $\mathsf{DomSet}$, and $\#\mathsf{DomSet}$, respectively. Furthermore $\#\Atwo$ is the class of all parameterized counting problems that are expressible as model counting problem with one quantifier alternation --- a formal introduction is given in Section~\ref{sec:a2}. It is known that
\begin{equation*}
  \cc{FPT} \leq^T \W \leq^T \#\W \subseteq \#\Wtwo \subseteq \#\Atwo\,,
\end{equation*}
where $\cc{C} \leq^T \cc{D}$ denotes that every problem in $\cc{C}$ can be reduced to a problem in $\cc{D}$ with respect to parameterized Turing-reductions. For further background on parameterized counting complexity, see \cite[Chapter~14]{flumgrohe}. While the parameterized complexity classes are defined via parsimonious reductions, we will rely on Turing reductions. Hence we cannot speak of completeness but instead of equivalence.
\begin{definition}
	Let $\cc{C}$ be a parameterized complexity class. A parameterized counting problem $(P,\pi)$ is $\cc{C}$-\emph{easy} if it can be reduced to a problem in $\cc{C}$ and it is $\cc{C}$-\emph{hard} if every problem in $\cc{C}$ reduces to~$(P,\pi)$, both with respect to parameterized Turing-reductions. A problem is $\cc{C}$-\emph{equivalent} if it is $\cc{C}$-\emph{easy} and $\cc{C}$-\emph{hard}.
\end{definition}

\paragraph*{Exponential-time hypotheses}
The strong exponential time hypothesis (SETH) asserts that for all $\delta>0$ there is some~$k\in\N$ such that $k$-SAT cannot be computed in time $O(2^{(1-\delta)n})$, where $n$ is the number of variables of the input formula.
A dominating set of size~$k$ in an $n$-vertex graph cannot be computed in time~$O(n^{k-\varepsilon})$ for any $\varepsilon>0$ unless SETH is false~\cite{DBLP:conf/soda/PatrascuW10}.
The exponential time hypothesis (ETH) asserts that $3$-SAT cannot be computed in time $\exp(o(m))$, where $m$ is the number of clauses of the input formula.

\section{Graphical conjunctive queries and colored variants}\label{sec:minors_and_colors}

It is instructive to first focus on conjunctive queries with one relation symbol~$E$ of arity two.
An example of such a query is the following formula:
\begin{equation}\label{eq: example query}
	x_1 \dots x_k \exists y : Ex_1y\land\dots\land Ex_ky\,.
\end{equation}
The relation~$E$ corresponds to a graph~$G$ and the free and quantified variables will be assigned vertices of~$G$.
In this example, an assignment~$a_1,\dots,a_k\in V(G)$ to the free variables satisfies the formula if and only if the vertices~$a_1,\dots,a_k$ have a common neighbor in~$G$.
It will be convenient for us to view the formula as a graph $H$ as depicted in Figure~\ref{fig: graphical conjunctive query}.
The vertices of~$H$ are partitioned into a set $X=\set{x_1,\dots,x_k}$ of free variables and a set $Y=\set{y}$ of quantified variables.
\begin{figure}
	\centering
	\begin{tikzpicture}
		\node[circle,inner sep=1.5pt,fill, label=right:{$y$}] (y) at (1,0) {};
		\node[circle,inner sep=1.5pt,fill, label=left:{$x_1$}] (x1) at (0,.5) {};
		\node[circle,inner sep=1.5pt,fill, label=left:{$x_k$}] (xk) at (0,-.5) {};
		\draw (xk) -- (y) -- (x1);
		\draw[dotted,shorten >=0.2cm,shorten <=0.2cm] (xk) -- (x1);
	\end{tikzpicture}
	\qquad
	\begin{tikzpicture}
		\node[circle,inner sep=1.5pt,fill, label=right:{$y$}] (y) at (1,.25) {};
		\node[circle,inner sep=1.5pt,fill, label=right:{$\tilde y$}] (y2) at (1,-.25) {};
		\node[circle,inner sep=1.5pt,fill, label=left:{$x_1$}] (x1) at (0,.5) {};
		\node[circle,inner sep=1.5pt,fill, label=left:{$x_k$}] (xk) at (0,-.5) {};
		\draw (xk) -- (y) -- (x1);
		\draw (xk) -- (y2);
		\draw[dotted,shorten >=0.2cm,shorten <=0.2cm] (xk) -- (x1);
	\end{tikzpicture}
	\caption{\label{fig: graphical conjunctive query}%
		\emph{Left:} Graphical representation of the conjunctive query in~\eqref{eq: example query}.
		\emph{Right:} A graphical conjunctive query that is ``equivalent'' to the example on the left in the sense that an assignment $a\colon\set{x_1,\dots,x_k}\to V(G)$ is a partial homomorphism from the left graph to~$G$ if and only if it is a partial homomorphism from the right graph to~$G$.
	}
\end{figure}
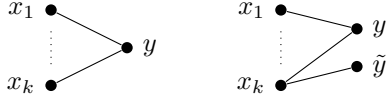
An assignment to the free variables corresponds to a function~$a:X\to V(G)$, and such an assignment satisfies the formula if it can be consistently extended to a homomorphism from~$H$ to~$G$.
This motivates the following definition, where we only consider simple graphs without loops, so we do not allow atomic subformulas of the form $Ezz$.
\begin{definition}
	A \emph{graphical conjunctive query} $(H,X)$ consists of a graph $H$ and a set $X$ of vertices of~$H$.
	We let $\homs{H,X}{G}$ be the set of all mappings from $X$ to $V(G)$ that can be extended to a homomorphism from $H$ to $G$, and we call these mappings partial homomorphisms. Formally, the set of partial homomorphisms is defined via
	\begin{equation}
		\homs{H,X}{G}= \set[\Big]{ a: X \rightarrow V(G) \;\Big|\; \exists h \in \homs{H}{G} : h|_X = a } \,.
  \end{equation}
\end{definition}

  Given two different graphical conjunctive queries $(H,X)$ and $(\hat{H},\hat{X})$ it might be the case that $\#\homs{H,X}{\star}$ and $\#\homs{\hat{H},\hat{X}}{\star}$ are the same functions. An example for this is given in Figure~\ref{fig: graphical conjunctive query}. In this case, we say that $(H,X)$ and $(\hat{H},\hat{X})$ are \emph{equivalent}, denoted as $(H,X)\sim(\hat{H},\hat{X})$, and the subgraph-minimal elements of the induced equivalence classes are called \emph{minimal}. An explicit notion of equivalence is given in Section~\ref{sec:lin_combs}.
  In our proofs, we make use of the following property of minimal queries, whose elementary proof we defer to Section~\ref{sec:counting_minimality}, where we generalize it to arbitrary structures.
	\begin{lemma}~\label{lem:counting_minimality_graphs}
		Let $(H,X)$ be a minimal conjunctive query and let $h$ be an endomorphism of~$H$. If $h$ maps $X$ bijectively to itself then $h$ is an automorphism.
	\end{lemma}

	\paragraph*{Color-prescribed Homomorphisms}
	While we are ultimately interested in the complexity of computing the number of partial homomorphisms, our hardness proofs become much more pleasant if we consider vertex-colored graphs (see also \cite[Chapter~5.4.2]{lovaszbook}).
  A graph $G$ is \emph{$H$-colored} if there is a homomorphism $c$ from $G$ to $H$. The value $c(v)$ for a vertex $v\in V(G)$ is called the \emph{color} of~$v$ and given a vertex $u\in V(H)$ we write $c^{-1}(u)$ for the set of all vertices in $G$ that are mapped to $u$. We consider \emph{color-prescribed homomorphisms}, which are homomorphisms $h\in\homs{H}{G}$ with the additional property that every vertex~$u\in V(H)$ maps to a vertex $h(u)$ whose color is~$u$. We write $\cphoms{H}{G}$ for the set of all color-prescribed homomorphisms.
    Given $X\subseteq V(H)$, we write $\cphoms{H,X}{G}$ for the set of \emph{partial} color-prescribed homomorphisms,
    that is, functions $a:X\to V(G)$ that can be extended to color-prescribed homomorphisms. Note that the set $\cphoms{H}{G}$ is empty if the $H$-coloring of $G$ is not surjective. Therefore it will be convenient to assume that the input graphs of the functions $\#\cphoms{H}{\star}$ and $\#\cphoms{H,X}{\star}$ are surjectively $H$-colored.

\subsection{Color-prescribed Homomorphisms Under Taking Minors}\label{subsec:min}

\def\rmedge{\ensuremath -}
\def\rmvtx{\ensuremath -}
\def\ctedge{\ensuremath/}

Recall that a \emph{minor} of a graph $H$ is any graph that can be obtained
from $H$ by deleting vertices and edges and by contracting edges.
In this section, we extend the minor relation to graphical conjunctive queries.
As it turns out, if $M$ is a minor of a graphical conjunctive query~$(H,X)$, then~$M$ can be reduced to~$(H,X)$.

Formally, given a graph $H$ and an edge $e \in E(H)$, we write $H \rmedge e$ to denote the
graph obtained from $H$ by deleting $e$ and $H \ctedge e$ for the graph $H$ where $e$ is
contracted. Note that any multiple edges and self-loops are deleted. Similarly, for an isolated vertex $v \in V(H)$ we write $H \rmvtx v$ for the graph resulting from $H$ by deleting~$v$. A minor of $H$ is any graph that can be
obtained from $H$ by iteratively applying these operations.

The deletion and contraction operations extend to graphical conjunctive queries~$(H,X)$ in the natural way, but we must decide each time how to modify the set~$X$.
For an isolated vertex $v \in V(H)$, we set $(H, X) \rmvtx v := (H \rmvtx v, X \setminus \set{v})$, and for an edge $e \in E(H)$, we set $(H,X) \rmedge e := (H\rmedge e,X)$.
For the contraction operation, let $e\in E(H)$ and let $w_e$ be the vertex that~$e$ is contracted to in~$H/e$.
The contraction of two quantified variables yields a quantified variable, but as soon as one endpoint of~$e$ is a free variable, the contracted variable is a free variable. Formally, we define $(H,X)/e = (H/e, X')$ where
\begin{equation*}
	X'=
	\begin{cases}
		X\,,                             & \text{if $e$ is disjoint from $X$,} \\
		(X \setminus e) \cup \{w_e \}\,, & \text{otherwise.}
	\end{cases}
\end{equation*}
Here $w_e$ denotes the vertex that~$e$ got contracted to.
As before, we say that $(\hat H,\hat X)$ is a \emph{minor} of $(H,X)$ if $(\hat H,\hat X)$ can be obtained from $(H,X)$ by iteratively applying these deletion and contraction operations.
We now that color-prescribed homomorphisms are ``minor-closed'', that is, if we know the color-prescribed homomorphisms from~$(H,X)$ then we know them from any minor as well.
\begin{lemma}\label{lem:minor_closed}
	Let $(H,X)$ be a graphical conjunctive query and let $(\hat H,\hat X)$ be a minor of $(H,X)$.
	Given an~$\hat H$-colored graph~$G$, we can in polynomial time compute an~$H$-colored graph~$G'$ with $\abs{V(G')}\le\abs{V(H)}\cdot\abs{V(G)}$ and
	$
		\#\cphoms{\hat H,\hat X}{G}
		=
		\#\cphoms{H,X}{G'}
	$.
\end{lemma}%
\begin{figure}[ptb]
	\centering
	\begin{tikzpicture}[yscale=0.4,xscale=0.34,-,thick]
		\node[circle,inner sep=1pt,fill] (1) at (0,0) {};
		\node[circle,inner sep=1pt,fill] (2) at (0,1) {};
		\node[circle,inner sep=1pt,fill] (3) at (0,2) {};

		\node[draw,inner sep=1pt] (4) at (2,0) {};
		\node[draw,inner sep=1pt] (5) at (2,1) {};
		\node[draw,inner sep=1pt] (6) at (2,2) {};

		\node[circle,inner sep=1pt,fill] (7) at (5,0) {};
		\node[circle,inner sep=1pt,fill] (8) at (5,1) {};
		\node[circle,inner sep=1pt,fill] (9) at (5,2) {};

		\node[draw,inner sep=1pt] (10) at (7,0) {};
		\node[draw,inner sep=1pt] (11) at (7,1) {};
		\node[draw,inner sep=1pt] (12) at (7,2) {};
		\draw[dashed] (7) -- (12); \draw[dashed] (7) -- (10); \draw[dashed] (7) -- (11);
		\draw[dashed] (8) -- (12); \draw[dashed] (8) -- (10); \draw[dashed] (8) -- (11);
		\draw[dashed] (9) -- (12); \draw[dashed] (9) -- (10); \draw[dashed] (9) -- (11);

		\node[circle,inner sep=1pt,fill,label={90:{\small $u$}}] (23) at (0,3.5) {};
		\node[draw,inner sep=1pt,label={90:{\small $v$}}] (24) at (2,3.5) {};
		\node[circle,inner sep=1pt,fill,label={90:{\small $u$}}] (25) at (5,3.5) {};
		\node[draw,inner sep=1pt,label={90:{\small $v$}}] (26) at (7,3.5) {};
		\draw (25) -- (26);

		\node (a) at (1,5.5) {\small $H\rmedge uv$};
		\draw[thin] (-.5,3) rectangle (2.5,4.75);
		\node (a) at (6,5.5) {\small $H$};
		\draw[thin] (4.5,3) rectangle (7.5,4.75);

		\node (a) at (1,-1.25) {\small $G$};
		\draw[thin] (-.5,-.5) rectangle (2.5,2.5);
		\node (a) at (6,-1.25) {\small $G'$};
		\draw[thin] (4.5,-.5) rectangle (7.5,2.5);

		\node (1000) at (3.5,-3.5) {\textbf{Edge deletion}};


		\node[circle,inner sep=1pt,fill] (101) at (11,0) {};
		\node[circle,inner sep=1pt,fill] (102) at (11,1) {};
		\node[circle,inner sep=1pt,fill] (103) at (11,2) {};

		\node[diamond,fill,inner sep=1pt] (104) at (12.9,0) {};
		\node[diamond,fill,inner sep=1pt] (105) at (12.9,1) {};
		\node[diamond,fill,inner sep=1pt] (106) at (12.9,2) {};

		\node[draw,inner sep=1pt] (104a) at (13.1,0) {};
		\node[draw,inner sep=1pt] (105a) at (13.1,1) {};
		\node[draw,inner sep=1pt] (106a) at (13.1,2) {};

		\node[circle,inner sep=1pt,draw] (107) at (15,0) {};
		\node[circle,inner sep=1pt,draw] (108) at (15,1) {};
		\node[circle,inner sep=1pt,draw] (109) at (15,2) {};

		\draw[thin] (101) -- (105);
		\draw[thin] (102) -- (106);
		\draw[thin] (103) -- (105);
		\draw[thin] (104a) -- (108);
		\draw[thin] (105a) -- (109);

		\node[circle,inner sep=1pt,fill] (110) at (18,0) {};
		\node[circle,inner sep=1pt,fill] (111) at (18,1) {};
		\node[circle,inner sep=1pt,fill] (112) at (18,2) {};

		\node[diamond,inner sep=1pt,fill] (113) at (20,0) {};
		\node[diamond,inner sep=1pt,fill] (114) at (20,1) {};
		\node[diamond,inner sep=1pt,fill] (115) at (20,2) {};

		\node[draw,inner sep=1pt] (116) at (22,0) {};
		\node[draw,inner sep=1pt] (117) at (22,1) {};
		\node[draw,inner sep=1pt] (118) at (22,2) {};

		\node[circle,inner sep=1pt,draw] (119) at (24,0) {};
		\node[circle,inner sep=1pt,draw] (120) at (24,1) {};
		\node[circle,inner sep=1pt,draw] (121) at (24,2) {};

		\draw[thin] (110) -- (114);
		\draw[thin] (111) -- (115);
		\draw[thin] (112) -- (114);
		\draw[thin,dashed] (110) -- (117);
		\draw[thin,dashed] (111) -- (118);
		\draw[thin,dashed] (112) -- (117);

		\draw[thin,dashed] (116) -- (120);
		\draw[thin,dashed] (117) -- (121);
		\draw[thin] (113) -- (120);
		\draw[thin] (114) -- (121);

		\draw[dashed] (113) -- (116);
		\draw[dashed] (114) -- (117);
		\draw[dashed] (115) -- (118);

		\node[circle,inner sep=1pt,fill] (120) at (11,3.5) {};
		\node[diamond,inner sep=1pt,fill,label={90:{\small $\,uv$}}] (121) at (12.9,3.5) {};
		\node[draw,inner sep=1pt] (121a) at (13.1,3.5) {};
		\node[circle,inner sep=1pt,draw] (122) at (15,3.5) {};
		\draw[thin] (120) -- (121); \draw[thin] (122) -- (121a);

		\node[circle,inner sep=1pt,fill] (123) at (18,3.5) {};
		\node[diamond,inner sep=1pt,fill,label={90:{\small $u$}}] (124) at (20,3.5) {};
		\node[inner sep=1pt,draw,label={90:{\small $v$}}] (125) at (22,3.5) {};
		\node[circle,inner sep=1pt,draw] (126) at (24,3.5) {};
		\draw[thin] (123) -- (124); \draw (124) -- (125); \draw[thin] (126) -- (125);

		\node (a) at (13,5.5) {\small $H\ctedge uv$};
		\draw[thin] (10.5,3) rectangle (15.5,4.75);
		\node (a) at (21,5.5) {\small $H$};
		\draw[thin] (17.5,3) rectangle (24.5,4.75);

		\node (a) at (13,-1.25) {\small $G$};
		\draw[thin] (10.5,-.5) rectangle (15.5,2.5);
		\node (a) at (21,-1.25) {\small $G'$};
		\draw[thin] (17.5,-.5) rectangle (24.5,2.5);
		\node (2000) at (17.5,-3.5) {\textbf{Edge contraction}};


		\node[circle,inner sep=1pt,fill] (201) at (28,0) {};
		\node[circle,inner sep=1pt,fill] (202) at (28,1) {};
		\node[circle,inner sep=1pt,fill] (203) at (28,2) {};

		\node[diamond,inner sep=1pt,fill] (204) at (30,0) {};
		\node[diamond,inner sep=1pt,fill] (205) at (30,1) {};
		\node[diamond,inner sep=1pt,fill] (206) at (30,2) {};

		\node[circle,inner sep=1pt,draw] (207) at (32,0) {};
		\node[circle,inner sep=1pt,draw] (208) at (32,1) {};
		\node[circle,inner sep=1pt,draw] (209) at (32,2) {};
		\draw[thin] (201) -- (204); \draw[thin] (202) -- (204);
		\draw[thin] (202) -- (205); \draw[thin] (202) -- (206);
		\draw[thin] (203) -- (205);
		\draw[thin] (206) -- (207); \draw[thin] (206) -- (209);
		\draw[thin] (205) -- (208);

		\node[circle,inner sep=1pt,fill] (210) at (35,0) {};
		\node[circle,inner sep=1pt,fill] (211) at (35,.7) {};
		\node[circle,inner sep=1pt,fill] (212) at (35,1.4) {};

		\node[diamond,inner sep=1pt,fill] (213) at (37,-.2) {};
		\node[diamond,inner sep=1pt,fill] (214) at (37,.5) {};
		\node[diamond,inner sep=1pt,fill] (215) at (37,1.2) {};

		\node[circle,inner sep=1pt,draw] (216) at (39,0) {};
		\node[circle,inner sep=1pt,draw] (217) at (39,.7) {};
		\node[circle,inner sep=1pt,draw] (218) at (39,1.4) {};
		\draw[thin] (210) -- (213); \draw[thin] (211) -- (213);
		\draw[thin] (211) -- (214); \draw[thin] (211) -- (215);
		\draw[thin] (212) -- (214);
		\draw[thin] (215) -- (216); \draw[thin] (215) -- (218);
		\draw[thin] (214) -- (217);

		\node[circle,inner sep=1pt,fill] (220) at (28,3.5) {};
		\node[diamond,inner sep=1pt,fill] (221) at (30,3.5) {};
		\node[circle,inner sep=1pt,draw] (222) at (32,3.5) {};
		\draw[thin] (220) -- (221); \draw[thin] (222) -- (221);

		\node[circle,inner sep=1pt,fill] (224) at (35,3.5) {};
		\node[diamond,inner sep=1pt,fill] (225) at (37,3.25) {};
		\node[circle,inner sep=1pt,draw] (226) at (39,3.5) {};
		\draw[thin] (224) -- (225); \draw[thin] (225) -- (226);

		\node[inner sep=1.5pt,draw,label={90:{\small $z$}}] (223) at (37,3.75) {};

		\node[inner sep=1.5pt,draw] (219) at (37,2.2) {};

		\node (a) at (30,5.5) {\small $H\rmvtx z$};
		\draw[thin] (27.5,3) rectangle (32.5,4.75);
		\node (a) at (37,5.5) {\small $H$};
		\draw[thin] (34.5,3) rectangle (39.5,4.75);

		\node (a) at (30,-1.25) {\small $G$};
		\draw[thin] (27.5,-.5) rectangle (32.5,2.5);
		\node (a) at (37,-1.25) {\small $G'$};
		\draw[thin] (34.5,-.5) rectangle (39.5,2.5);
		\node (3000) at (33.5,-3.5) {\textbf{Vertex deletion}};
	\end{tikzpicture}
	\caption{\label{fig:minor}%
		Illustration of the reduction for each operation as demonstrated in Lemma~\ref{lem:minor_closed}. Edges that are added in the reduction are dashed.}
\end{figure}
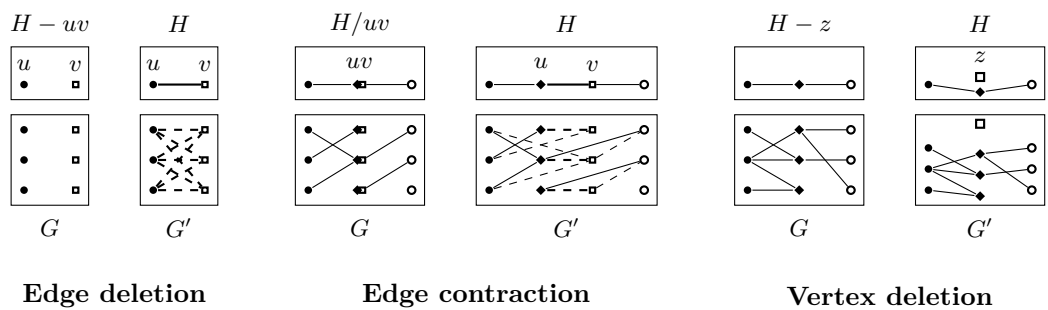%
\begin{proof}
	The claim is trivial if~$(\hat H,\hat X)$ and~$(H,X)$ are equal.
	We prove the claim in case~$(\hat H,\hat X)$ is obtained from~$(H,X)$ by a single deletion or contraction operation, the full result then follows by induction.
	Figure~\ref{fig:minor} illustrates the proof for each of the three operations.

	\emph{Edge deletions.}
	Let $e\in E(H)$ be an edge with~$e=uv$ and suppose that $\hat H=H \rmedge e$ and $\hat X=X$.
	Let $G$ be a~$\hat H$-colored graph given as input, together with the coloring~${c\colon V(G)\to V(\hat H)}$.
	To construct~$G'$, we start from~$G$ and simply add all possible edges between the color classes $c^{-1}(u)$ and $c^{-1}(v)$; clearly this construction takes polynomial time, $G'$ has the same number of vertices as~$G$, and $c$ is a homomorphism from~$G'$ to~$H$.
	To verify the correctness, we show that $\cphoms{\hat H,X}{G}=\cphoms{H,X}{G'}=$ holds.
	Indeed, let~$h:V(H)\to V(G)$ a color-prescribed mapping.
	Since~$h(e)\in E(G')$ holds by construction and~$e$ is the only constraint where~$H$ and~$\hat H$ differ, the addition of the edge~$e$ does not matter.
	Hence~$h$ is an element of~$\cphoms{\hat H,X}{G}$ if and only if it is an element of~$\cphoms{H,X}{G'}$.
	Moreover, the set of partial color-prescribed homomorphisms stays the same.

	\emph{Vertex deletions.}
	Let $z\in V(H)$ be an isolated vertex and suppose~$\hat H=H \rmvtx z$ and $\hat X=X\setminus\set{z}$.
	Let $G$ be a~$\hat H$-colored graph given as input, together with the coloring~${c\colon V(G)\to V(\hat H)}$.
	To construct~$G'$, we start from~$G$ and simply add an isolated vertex~$z'$ to it, whose color~$c(z')$ we define as~$z$; clearly $c$ is now a homomorphism from~$G'$ to~$H$.
	To verify the correctness, observe that $\#\cphoms{\hat H,\hat X}{G}=\#\cphoms{H,X}{G'}$ holds:
	Any color-prescribed homomorphism~$h$ from~$H$ to~$G'$ remains a color-prescribed homomorphism from~$\hat H$ to~$G$ by restricting~$h$ to~$V(\hat H)$.
	Conversely, any~$h$ from~$\hat H$ to~$G$ can be extended in exactly one color-prescribed way by setting~${h(z)=z'}$.
	Thus the number of partial color-prescribed homomorphisms stays the same.

	\emph{Edge contractions.}
	Let $e\in E(H)$ be an edge with $e=uv$, and suppose~$(\hat H,\hat X)=(H,X)\ctedge e$.
	Contracting the edge $e$ in $H$ identifies the vertices $u$ and $v$; let us call the new vertex~$w\in V(\hat H)$.
	Let $G$ be a~$\hat H$-colored graph given as input, together with the coloring~${c\colon V(G)\to V(\hat H)}$.
	We want to use~$G'$ ensure that any color-prescribed homomorphism~$h$ from~$H$ to~$G'$ assigns~$u$ and~$v$ to the same value, that is, satisfies the equality constraint~$h(u)=h(v)$.
	To do this, we simply put an induced perfect matching in~$G'$ between the color class of~$u$ and the color class of~$v$.
	More formally, we start from~$G$ and split every vertex~$x\in c^{-1}(w)$ into an edge~$x_u x_v$ in~$G'$, but we leave their neighborhoods intact, that is, we have $N_{G'}(x_u)\cap V(G)=N_{G'}(x_v)\cap V(G)=N_{G}(x)$.
	Clearly $G'$ is now~$H$-colored, and it has exactly $\abs{c^{-1}(w)}$ vertices more than~$G$.
	To verify correctness, again observe that
	$\#\cphoms{\hat H,\hat X}{G}=\#\cphoms{H,X}{G'}$ holds:
	Our construction forces any color-prescribed homomorphism~$h$ from~$H$ to~$G'$ to satisfy~$h(u)=h(v)$ and thus gives rise to a color-prescribed homomorphism~$\hat h$ from~$\hat H$ to~$G$ by setting $\hat h(w)=h(u)$; this mapping $h\mapsto\hat h$ is a bijection.
	If~$e$ is disjoint from~$X$, then~$X=\hat X$ holds and the set of partial homomorphisms is the same because~$h|_X=\hat h|_{\hat X}$ holds.
	If~$e$ is not disjoint from~$X$, then $\hat X\subsetneq X$ holds, but still the mapping $h|_X\mapsto\hat h|_{\hat X}$ is bijective.
	In any case, the number of partial homomorphisms is the same.
\end{proof}

\subsection{Reducing Color-prescribed to Uncolored Homomorphisms}\label{sec:phom_to_hom}

We show that the number of color-prescribed homomorphism for graphical conjunctive queries can be expressed by using the number of uncolored homomorphisms.
For the reduction, we need yet another type of homomorphisms as an intermediate step, namely \emph{colorful homomorphisms}.
In contrast to color-prescribed homomorphisms, \emph{colorful homomorphisms} are homomorphisms $h\in\homs{H}{G}$
with the less prescriptive property that the image of $h$ contains a vertex for each color, that is, we have $c(h(V(H))) = V(H)$.
We write $\cfhoms{H}{G}$ for the set of all colorful homomorphisms. Given $X\subseteq V(H)$,
we write $\cfhoms{H,X}{G}$ for the set of \emph{partial} colorful homomorphism, i.e., functions~$a$ from~$X$ to~$c^{-1}(X)$ that can be extended to a colorful homomorphism from~$H$ to~$G$.
We emphasize that we only consider functions $a$ satisfying $c(a(X))=X$.\newline
In the reduction, we need the following simple observation about the relationship between
$H$-colored graphs and homomorphisms into them.
\begin{fact}\label{fact:canonical_endo}
	Let $c$ be the $H$-coloring of a graph $G$ and let $h\in\homs{H}{G}$.
	Then the function $\pi: v\mapsto c(h(v))$ is an endomorphism of $H$.
	If $h$ is colorful and satisfies $c(h(X))=X$ for a set $X \subseteq H$, then $\pi$ is an automorphism that maps $X$ to $X$ in such a way that the function composition~$h\circ \pi^{-1}$ is a color-prescribed homomorphism.
\end{fact}
\begin{proof}
The first statement holds as $\pi$ is the composition of two homomorphisms $h: V(H) \rightarrow V(G)$ and $c: V(G)\rightarrow V(H)$. For the second statement observe that colorfulness of $h$ implies that $\pi(V(H))=V(H)$ and hence, together with the assumption that $c(h(X))=X$, the endomorphism $\pi$ is an automorphism that maps $X$ to $X$. Finally we have that
\begin{equation*}
c(h\circ\pi^{-1}(v)) = c(h\circ (c\circ h)^{-1}(v)) = c(h(h^{-1}(c^{-1}(v)))) = v \,,
\end{equation*}
and hence that $h\circ \pi^{-1}$ is color-prescribed.
\end{proof}

Using Fact~\ref{fact:canonical_endo} and by defining a suitable equivalence relation,
we obtain the first part of the reduction, namely from $\#\cphomsprob(\Delta)$ to
$\#\cfhomsprob(\Delta)$.

\begin{lemma}[Reduction from Color-prescribed to Colorful Homomorphisms]~\\
	Let $(H,X)$ be a graphical conjunctive query and let $G$ be an $H$-colored graph.
  Then
  \begin{equation*}
    \#\cfhoms{H,X}{G} = \#\mathsf{Aut}(H,X) \cdot \#\cphoms{H,X}{G},
  \end{equation*}
  where $\mathsf{Aut}(H,X):=\{b:X\rightarrow X\mid\exists h\in\mathsf{Aut}(H):h|_X=b\}$.
\end{lemma}
\begin{proof}
	We define an equivalence relation $\sim$ on the set
	\mbox{$\cfhoms{H,X}{G}$} as follows: Two mappings $a, a'\in \cfhoms{H,X}{G}$ are equivalent,
	written $a \sim a'$, if and only if their image is equal, that is, $a(X)=a'(X)$ holds.
	We denote the equivalence class of $a$ with $\llbracket a\rrbracket$.

	To show~``$\ge$'', let $a\in\cphoms{H,X}{G}$ be a partial color-prescribed homomorphism.
	We show that~$\llbracket a\rrbracket$ contains at least~$\#\mathsf{Aut}(H,X)$ elements, exactly one of which is color-prescribed.
	Indeed, composing~$a$ with a bijection~$b\in\mathsf{Aut}(H,X)$ yields distinct functions~$a\circ b$, each of which has the same image as~$a$ and thus is an element of~$\llbracket a\rrbracket$.
	Moreover, each~$a\circ b$ can be extended to a colorful homomorphism by composing the assumed automorphism extension of~$b$ with the assumed colorful homomorphism extension for~$a$.
	Finally,~$a\circ b$ is color-prescribed only if~$b$ is the identity function.
	Thus each color-prescribed~$a$ leads to at least~$\#\mathsf{Aut}(H,X)$ distinct colorful~$a\circ b$, which proves~``$\ge$''.

	To show~``$\le$'', it suffices to prove that every colorful~$a\in\cfhoms{H,X}{G}$ has some partial automorphism~$b\in\mathsf{Aut}(H,X)$ such that~$a\circ b\in\cphoms{H,X}{G}$ holds.
	Let $h$ denote the assumed colorful extension of~$a$.
	Using Fact~\ref{fact:canonical_endo}, $h$ induces a canonical automorphism
	$\pi \in \mathsf{Aut}(H)$ which maps~$X$ to~$X$. Thus the function~$b$ with~$b=\pi^{-1}|_X$ is a member of~$\mathsf{Aut}(H,X)$.
	Moreover, by definition of~$\pi$, the mapping~$a\circ b$ is a partial color-prescribed homomorphism.
\end{proof}

It remains to reduce colorful homomorphisms to uncolored homomorphisms.
We first observe that for minimal queries $(H,X)$, the property $c(a(X))=X$
already implies the existence of a colorful extension of $a$.
\begin{observation}\label{lem:colhom_counting_minimal}
	Let $(H,X)$ be a minimal graphical conjunctive query
	and let $G$ be an $H$-colored graph with coloring~$c$.
	Furthermore let $a$ be a function from $X$ to $V(G)$ satisfying $c(a(X))=X$.
	If $h \in\homs{H}{G}$ is a homomorphism that extends $a$, then~$h$ is colorful.
\end{observation}
\begin{proof}
	By Fact~\ref{fact:canonical_endo}, $h$ induces the canonical endomorphism $\pi: v\mapsto c(h(v))$.
	As $h$ extends $a$ and $c(a(X))=X$ holds, the endomorphism $\pi$ bijectively maps $X$ to $X$.
	Therefore, by Lemma~\ref{lem:counting_minimality_graphs}, $\pi$ is an automorphism, which implies that $h$ is colorful.
\end{proof}

We proceed with the reduction to uncolored homomorphisms.
\begin{lemma}[Reduction from Colorful to Uncolored Homomorphisms]
	\label{lem:colorful to uncolored}
	Let $(H,X)$ be a minimal graphical conjunctive query.
	Then there exists a deterministic algorithm $\mathbb{A}$ with oracle access to
	$\#\homs{H,X}{\star}$ that computes $\#\cfhoms{H,X}{\star}$.
	Furthermore $\mathbb{A}$ runs in time $O(f(|H,X|)\cdot n^c)$ for some computable
	function $f$ and some constant $c$ independent of $(H,X)$.
\end{lemma}
\begin{proof}
	Our reduction will use multivariate polynomial interpolation.
	We start by providing an intuition.
	Let $k=|V(H)|$, $\ell = |X|$ and assume that the vertices of $H$ are the integers $1,\dots,k$
	from which the first $\ell$ are in $X$.
	Further, let an $H$-colored graph $G$ with coloring $c$ be given.
	For every color $i \in [k]$, we clone (including incident edges) all vertices with color~$i$ precisely $z_i-1$ times for some positive integer $z_i$.
	We denote the resulting graph as $G^{\vec{z}}$, which is still $H$-colored.

	Next, the numbers $z_i$ for $i \in [k]$ are interpreted as formal variables
	and it will turn out that $\#\homs{H,X}{G^{\vec{z}}}$ is a polynomial
	in $\mathbb{Q}[z_1,\dots,z_k]$.
	Additionally, the coefficient of $\Pi_{i=1}^\ell z_i$ is the number of assignments
	$a$ from~$X$ to $V(G)$ such that $c(a(X))=X$ and that $a$ can be extended to a homomorphism.

	Applying Lemma~\ref{lem:colhom_counting_minimal} we obtain that those homomorphisms
	are indeed colorful.
	We will be able to compute the coefficient by standard multivariate interpolation.
	Note that the evaluation of the polynomial in $\vec{z}$ can be done by querying the oracle
	for $G^{\vec{z}}$.

	Formally, we define an equivalence relation on $\homs{H,X}{G^{\vec{z}}}$ as follows.
	Two assignments $a_{\vec{z}}$ and $a_{\vec{z}}'$ are equivalent if and only if for every
	$x\in X$ it holds that $a_{\vec{z}}(x)$ and $a_{\vec{z}}'(x)$ are clones of the same vertex.
	Note that every equivalence class corresponds to precisely one mapping
	$a \in \homs{H,X}{G}$ and we write $\llbracket a \rrbracket_{\vec{z}}$ for that class.
	Next observe that every $a \in \homs{H,X}{G}$ induces a color-vector
	$c_a =(c(a(1),\dots,c(a(\ell)))\in [k]^{\ell}$, which allows us to express the size of
	$\llbracket a \rrbracket_{\vec{z}}$ as $\Pi_{i=1}^\ell z_{c_a(i)}$.
	This yields the following polynomial for $\#\homs{H,X}{G^{\vec{z}}}$:
	\begin{align*}
		\#\homs{H,X}{G^{\vec{z}}} & = \sum_{a\in\homs{H,X}{G}}
		\#\llbracket a \rrbracket_{\vec{z}}                    \\
		                          & = \sum_{v \in [k]^{\ell}}~
		\sum_{\substack{a\in\homs{H,X}{G}                      \\c_a=v}}~
		\prod_{i=1}^\ell z_{c_a(i)}
	\end{align*}

	Finally it can be verified easily that the coefficient of $\Pi_{i=1}^\ell z_i$ is indeed the
	number of assignments $a$ from $X$ to $V(G)$ such that $c(a(X))=X$ and that $a$ can be
	extended to a homomorphism $h$. Note that $h$ is colorful by minimality of $H$ and
	Lemma~\ref{lem:colhom_counting_minimal}.
	In other words, the coefficient of $\Pi_{i=1}^\ell z_i$ is precisely $\#\cfhoms{H,X}{G}$.

	As we can evaluate the polynomial for every vector $\vec{z}\in \mathbb{N}_{>0}^k$
	the coefficient can be computed using standard multivariate polynomial interpolation
	(see e.g.\ \cite[Theorem~1.38]{radu_dis}).
\end{proof}\pagebreak

\section{The complexity of graphical conjunctive queries}\label{sec:4}

In this section, we classify the complexity of counting homomorphisms for classes of graphical conjunctive queries.
Formally, we consider the parameterized counting problem $\#\homsprob(\Delta)$ for each fixed class~$\Delta$ of graphical conjunctive queries.
This problem is given as input a query $(H,X)\in \Delta$ and a graph~$G$ and the task is to compute the number $\#\homs{H,X}{G}$. The problem is parameterized by the size of~$H$. From Subsection~\ref{sec:hard_w_2} on we will also consider the color-prescribed variant which yields significantly more pleasant proofs as Lemma~\ref{lem:minor_closed} allows us to reduce from minors of conjunctive queries in this case. Using the observations in Section~\ref{sec:phom_to_hom}, the hardness results we discuss for color-prescribed homomorphisms carry over to the uncolored situation and yield our refined complexity classification for the case of graphs (Theorem~\ref{thm:extension}).

The first five subsections correspond to the five cases in the Complexity Pentachotomy.
In each case, we define the precise parameters that we need in order to classify the complexity of~$\#\homsprob(\Delta)$, and we also give an example class of queries that exhibits that complexity.
All five example classes along with their structural properties are depicted in Figure~\ref{fig:ex_table}.
In the sixth subsection, we are then in position to restate the Complexity Pentachotomy in an explicit fashion, relying on the structural parameters defined in subsections 1--5.
The first three subsections should be considered a review of previous work~\cite{DurandM15,ChenM15}, which is necessary to formally state our techniques and results.

\begin{figure}[t!]
	\centering

	\begin{center}
		\begin{tabular}{ l  c  c  c  c  c }
			\toprule
			{\large Query Classes}           & {\large {\LARGE $\exclass$}$_{ \mathsf{poly}}$} & {\large {\LARGE $\exclass$}$_{\W}$}          & {\large {\LARGE $\exclass$}$_{\#\W}$}          & {\large {\LARGE $\exclass$}$_{\#\Wtwo}$}         & {\large {\LARGE $\exclass$}$_{\#\Atwo}$}             \\ \midrule%
			\begin{tikzpicture}[-,xscale=0.575,yscale = 0.75,thick]
				\node (42) at (0.5,3) {$~$};
				\node (42) at (0.5,0) {$~$};
				\node[align=left] (0) at (0,1.5) {Query\\for $k=4$};
			\end{tikzpicture}       &
			\begin{tikzpicture}[-,xscale=0.575,yscale = 0.75,thick]
				\node (42) at (0.5,3) {$~$};
				\node (42) at (0.5,0) {$~$};
				\node[circle,inner sep=1.5pt,fill] (1) at (1,0) {};
				\node[circle,inner sep=1.5pt,fill] (2) at (1,1) {};
				\node[circle,inner sep=1.5pt,fill] (3) at (1,2) {};
				\node[circle,inner sep=1.5pt,fill] (4) at (1,3) {};

				\node[rectangle,inner sep=1.5pt,draw] (10) at (0,0.5) {};
				\node[rectangle,inner sep=1.5pt,draw] (20) at (0,1.5) {};
				\node[rectangle,inner sep=1.5pt,draw] (30) at (0,2.5) {};

				\draw (1) --(10);
				\draw (2) --(10);
				\draw (2) --(20);
				\draw (3) --(20);
				\draw (3) --(30);
				\draw (4) --(30);
			\end{tikzpicture}       &
			\begin{tikzpicture}[-,xscale=0.575,yscale = 0.75,thick]
				\node (42) at (0.5,3) {$~$};
				\node (42) at (0.5,0) {$~$};
				\node[rectangle,inner sep=1.5pt,draw] (1) at (0,0.5) {};
				\node[rectangle,inner sep=1.5pt,draw] (2) at (1,0.5) {};
				\node[rectangle,inner sep=1.5pt,draw] (3) at (0,2.5) {};
				\node[rectangle,inner sep=1.5pt,draw] (4) at (1,2.5) {};

				\draw (1) --(2);
				\draw (1) --(3);
				\draw (1) --(4);
				\draw (2) --(3);
				\draw (2) --(4);
				\draw (3) --(4);
			\end{tikzpicture}       &
			\begin{tikzpicture}[-,xscale=0.575,yscale = 0.75,thick]
				\node (42) at (0.5,3) {$~$};
				\node (42) at (0.5,0) {$~$};
				\node[circle,inner sep=1.5pt,fill] (1) at (1,0) {};
				\node[circle,inner sep=1.5pt,fill] (2) at (1,1) {};
				\node[circle,inner sep=1.5pt,fill] (3) at (1,2) {};
				\node[circle,inner sep=1.5pt,fill] (4) at (1,3) {};

				\node[rectangle,inner sep=1.5pt,draw] (10) at (0,0) {};
				\node[rectangle,inner sep=1.5pt,draw] (20) at (0,0.5) {};
				\node[rectangle,inner sep=1.5pt,draw] (30) at (0,1.25) {};
				\node[rectangle,inner sep=1.5pt,draw] (40) at (0,1.75) {};
				\node[rectangle,inner sep=1.5pt,draw] (50) at (0,2.5) {};
				\node[rectangle,inner sep=1.5pt,draw] (60) at (0,3) {};
				\draw (1) --(10);
				\draw (2) --(10);
				\draw (1) --(20);
				\draw (3) --(20);
				\draw (1) --(30);
				\draw (4) --(30);
				\draw (2) --(40);
				\draw (3) --(40);
				\draw (2) --(50);
				\draw (4) --(50);
				\draw (3) --(60);
				\draw (4) --(60);
			\end{tikzpicture}       &
			\begin{tikzpicture}[-,xscale=0.575,yscale = 0.75,thick]
				\node (42) at (0.5,3) {$~$};
				\node (42) at (0.5,0) {$~$};
				\node[circle,inner sep=1.5pt,fill] (1) at (1,0) {};
				\node[circle,inner sep=1.5pt,fill] (2) at (1,1) {};
				\node[circle,inner sep=1.5pt,fill] (3) at (1,2) {};
				\node[circle,inner sep=1.5pt,fill] (4) at (1,3) {};

				\node[rectangle,inner sep=1.5pt,draw] (10) at (0,1.5) {};
				\draw (1) --(10);
				\draw (2) --(10);
				\draw (3) --(10);
				\draw (4) --(10);
			\end{tikzpicture}       &
			\begin{tikzpicture}[-,xscale=0.575,yscale = 0.75,thick]
				\node (42) at (0,3) {$~$};
				\node (42) at (0,0) {$~$};
				\node[circle,inner sep=1.5pt,fill] (1) at (1.5,0) {};
				\node[circle,inner sep=1.5pt,fill] (2) at (1.5,1) {};
				\node[circle,inner sep=1.5pt,fill] (3) at (1.5,2) {};
				\node[circle,inner sep=1.5pt,fill] (4) at (1.5,3) {};

				\node[rectangle,inner sep=1.5pt,draw] (10) at (0,0) {};
				\node[rectangle,inner sep=1.5pt,draw] (20) at (0,1) {};
				\node[rectangle,inner sep=1.5pt,draw] (30) at (0,2) {};
				\node[rectangle,inner sep=1.5pt,draw] (40) at (0,3) {};

				\node[rectangle,inner sep=1.5pt,draw] (100) at (-1,0.5) {};
				\node[rectangle,inner sep=1.5pt,draw] (200) at (-1,1.5) {};
				\node[rectangle,inner sep=1.5pt,draw] (300) at (-1,2.5) {};

				\node[rectangle,inner sep=1.5pt,draw] (1000) at (-2,1) {};
				\node[rectangle,inner sep=1.5pt,draw] (2000) at (-2,2) {};

				\node[rectangle,inner sep=1.5pt,draw] (10000) at (-3,1.5) {};
				\draw (1) --(10);
				\draw (2) --(20);
				\draw (3) --(30);
				\draw (4) --(40);

				\draw (100) --(10);
				\draw (100) --(20);
				\draw (200) --(20);
				\draw (200) --(30);
				\draw (300) --(30);
				\draw (300) --(40);

				\draw (1000) --(100);
				\draw (1000) --(200);
				\draw (2000) --(200);
				\draw (2000) --(300);

				\draw (10000) --(1000);
				\draw (10000) --(2000);
			\end{tikzpicture}                                                                                                                                               \\ \midrule	
			\begin{tikzpicture}[-,xscale=0.575,yscale = 0.75,thick]
				\node (42) at (0,3) {$~$};
				\node (42) at (0,0) {$~$};
				\node[align=left] (0) at (0,1.5) {$\mathsf{contract}$\\for $k=4$};
			\end{tikzpicture}       &
			\begin{tikzpicture}[-,xscale=0.575,yscale = 0.75,thick]
				\node (42) at (1,3) {$~$};
				\node (42) at (1,0) {$~$};
				\node[circle,inner sep=1.5pt,fill] (1) at (1,0) {};
				\node[circle,inner sep=1.5pt,fill] (2) at (1,1) {};
				\node[circle,inner sep=1.5pt,fill] (3) at (1,2) {};
				\node[circle,inner sep=1.5pt,fill] (4) at (1,3) {};

				\draw (1) --(2);
				\draw (2) --(3);
				\draw (3) --(4);
			\end{tikzpicture}       &
			\begin{tikzpicture}[-,xscale=0.575,yscale = 0.75,thick]
				\node (42) at (0.5,3) {$~$};
				\node (42) at (0.5,0) {$~$};
				\node (0) at (0.5,1.5) {\LARGE $\emptyset$};
			\end{tikzpicture}       &
			\begin{tikzpicture}[-,xscale=0.575,yscale = 0.75,thick]
				\node (42) at (0.5,3) {$~$};
				\node (42) at (0.5,0) {$~$};
				\node[circle,inner sep=1.5pt,fill] (1) at (0,0.5) {};
				\node[circle,inner sep=1.5pt,fill] (2) at (1,0.5) {};
				\node[circle,inner sep=1.5pt,fill] (3) at (0,2.5) {};
				\node[circle,inner sep=1.5pt,fill] (4) at (1,2.5) {};

				\draw (1) --(2);
				\draw (1) --(3);
				\draw (1) --(4);
				\draw (2) --(3);
				\draw (2) --(4);
				\draw (3) --(4);
			\end{tikzpicture}       &
			\begin{tikzpicture}[-,xscale=0.575,yscale = 0.75,thick]
				\node (42) at (0.5,3) {$~$};
				\node (42) at (0.5,0) {$~$};
				\node[circle,inner sep=1.5pt,fill] (1) at (0,0.5) {};
				\node[circle,inner sep=1.5pt,fill] (2) at (1,0.5) {};
				\node[circle,inner sep=1.5pt,fill] (3) at (0,2.5) {};
				\node[circle,inner sep=1.5pt,fill] (4) at (1,2.5) {};

				\draw (1) --(2);
				\draw (1) --(3);
				\draw (1) --(4);
				\draw (2) --(3);
				\draw (2) --(4);
				\draw (3) --(4);
			\end{tikzpicture}       &
			\begin{tikzpicture}[-,xscale=0.575,yscale = 0.75,thick]
				\node (42) at (0.5,3) {$~$};
				\node (42) at (0.5,0) {$~$};
				\node[circle,inner sep=1.5pt,fill] (1) at (0,0.5) {};
				\node[circle,inner sep=1.5pt,fill] (2) at (1,0.5) {};
				\node[circle,inner sep=1.5pt,fill] (3) at (0,2.5) {};
				\node[circle,inner sep=1.5pt,fill] (4) at (1,2.5) {};

				\draw (1) --(2);
				\draw (1) --(3);
				\draw (1) --(4);
				\draw (2) --(3);
				\draw (2) --(4);
				\draw (3) --(4);
			\end{tikzpicture}                                                                                                                                               \\ \midrule
			$\mathsf{tw}$                    & $O(1)$              & $\infty$                 & $\infty$                   & $O(1)$                  & $\infty$                      \\
			$\mathsf{tw}(\mathsf{contract})$ & $O(1)$              & $O(1)$                   & $\infty$                   & $\infty$                & $\infty$                      \\
			$\mathsf{dss}$                   & $O(1)$              & $O(1)$                   & $O(1)$                     & $\infty$                & $\infty$                      \\
			$\mathsf{lmn}$                   & $O(1)$              & $O(1)$                   & $O(1)$                     & $O(1)$                  & $\infty$                      \\
			Complexity                       & {\small $\cc{P}$}   & {\small $\W$-eq.} & {\small $\#\W$-eq.} & {\small $\#\Wtwo$-hard} & {\small $\#\Atwo$-eq.} \\
			\bottomrule
		\end{tabular}
	\end{center}
	\caption{\label{fig:ex_table}%
		Illustration of the five typical classes of conjunctive queries that we discuss in Section~\ref{sec:4}.
		Depicted is the query~$(H,X)$ for~$k=4$, where free variables (i.e., vertices in~$X$) are drawn as solid discs and quantified variables (i.e., vertices in~$V(G)\setminus X$) are drawn as hollow squares.
		We also display the contract (see Definition~\ref{def:contract}) of each query.
		We write $O(1)$ whenever a parameter is bounded by a constant in the entire query class, and $\infty$ whenever it is unbounded.%
	}
\end{figure}

\subsection{Query Classes That Are Polynomial-time}
\label{sec:parameters}\label{sec:queries polytime}

Which classes~$\Delta$ of graphical conjunctive queries make the problem~$\#\homsprob(\Delta)$ polynomial-time computable?
Chen, Mengel and Durand~\cite{DurandM15,ChenM15} proved that the problem is polynomial-time computable if all graphs in~$\Delta$ as well as their \emph{contracts} have at most a constant treewidth.
While we do not need to define treewidth in this paper (see, e.g., \cite[Chap.~11]{flumgrohe}), we do define contracts.

\begin{definition}[Contract]
	\label{def:contract}
	The \emph{contract} of a conjunctive query $(H,X)$ is a graph on the vertex set~$X$, obtained by adding an edge between two vertices $u$ and $v$ in $X$ if~$uv$ is an edge of~$H$ or if there exists a connected component $C$ in $H\setminus X$ that is adjacent to both~$u$ and~$v$.
	Given a class $\Delta$ of conjunctive queries, we write $\mathsf{contract}(\Delta)$ for the set of all of its contracts.
\end{definition}
The following example class~$\exclass_{\poly}$ of queries is satisfied by the $k$-tuples~$v_1,v_3,\dots,v_{2k-1}$ of vertices in the input graph~$G$ for which there exists an extension~$v_2,v_4,\dots,v_{2k-2}$ such that~$v_1,\dots,v_{2k-1}$ is a walk in~$G$:
\begin{equation}\label{eq:example poly}
	\exclass_{\poly} = \setc{\psi_k}{k \in \mathbb{N}}\,,\quad\text{where } \psi_k
	:= x_1\dots x_k ~\exists y_1 \dots \exists y_{k-1}:
	\bigwedge_{1 \leq i < k} Ex_iy_i \wedge Ey_ix_{i+1}
	\,.
\end{equation}
Since these queries and their contracts are just paths (cf.\ Figure~\ref{fig:ex_table}), their treewidth is bounded by a constant.
Thus the corresponding problem~$\#\homsprob(\exclass_{\poly})$ is polynomial-time computable by the complexity trichotomy of Chen, Mengel and Durand~\cite{DurandM15,ChenM15}.
This can be seen more directly using dynamic programming, or by considering the square~$A^2$ of the adjacency matrix of~$G$ and replacing each positive entry by a~$1$ to obtain~$B$ -- then the sum of all entries in $B^{k}$ is the desired number.

Formally, the trichotomy theorem of~\cite{DurandM15,ChenM15} is as follows:
\begin{theorem}[\cite{DurandM15,ChenM15}]\label{thm:old_classification}
	Let $\exclass$ be a set of minimal conjunctive queries.
	If the treewidth of the formulas in $\exclass$ and their contracts,
	is bounded then $\#\homsprob(\exclass)$ is solvable in polynomial time.
	If the treewidth of the formulas is unbounded but the treewidth of the contracts is bounded,
	then $\#\homsprob(\exclass)$ is $\W$-equivalent.
	If the treewidth of the contracts is unbounded, then $\#\homsprob(\exclass)$ is $\#\W$-hard.
\end{theorem}

\subsection{Query Classes That Are W[1]-equivalent}
As it turns out, the situation in which the treewidth of the queries and their cores is bounded appears to be the only one that is polynomial-time computable:
If the treewidth of~$\Delta$ or~$\mathsf{contract}(\Delta)$ is unbounded, then Theorem~\ref{thm:old_classification} implies that $\#\homsprob(\Delta)$ is not polynomial-time computable, unless~$\mathsf{FPT}=\W$ holds.
More precisely, when the treewidth of~$\mathsf{contract}(\Delta)$ is bounded but the treewidth of~$\Delta$ is unbounded, then the problem is~$\W$-equivalent.
To exemplify this latter situation further, note that the $\W$-complete~$k$-clique problem is a special case: The following query~$\psi_k$ in the class~$\exclass_{\W}$ is satisfiable (i.e., has the empty tuple as a satisfying assignment) if and only if the input graph has a clique of size~$k$.
Formally,
\begin{equation}\label{eq:obstructions_3}
	\exclass_{\W} = \setc{\psi_k}{k \in \mathbb{N}},\quad\text{where } \psi_k
	:= \exists y_1\dots \exists y_k: \bigwedge_{1\leq i<j\leq k} Ey_iy_j
	\,.
\end{equation}
Indeed, the contract of each query in~$\exclass_{\W}$ is the empty graph, but the treewidth of the~$k$-th query is equal to~$k-1$, so $\homsprob(\exclass_{\W})$ is $\W$-equivalent by Theorem~\ref{thm:old_classification}.

\subsection{Query Classes That Are \#W[1]-equivalent}\label{sec:in_W_1}

If the treewidth of~$\mathsf{contract}(\Delta)$ is unbounded, then $\#\homsprob(\Delta)$ is $\#\W$-hard by Theorem~\ref{thm:old_classification}.
We now define the \emph{dominating star size}, a structural parameter with the property that, if all elements of~$\Delta$ have bounded dominating star size, then~$\#\homsprob(\Delta)$ is $\#\W$-easy.
\begin{definition}[Dominating star size]
	Let $(H,X)$ be a conjunctive query and let $Y_1,\dots, Y_\ell$ be
	the connected components of the subgraph~$H[V(H)\setminus X]$ induced by the quantified variables.
	Further, let~$k_i$ be the number of vertices~$x\in X$ for which there exists a vertex
	$y \in Y_i$ that is adjacent to $x$.
	The \emph{dominating star size} of $(H,X)$ is defined via
  \begin{equation*}
    \mathsf{dss}(H,X) = \max \setc[\big]{k_i}{i \in \ell}\,.
  \end{equation*}
\end{definition}
This notion is identical to the notion of \emph{strict star size}, which was used by Chen and Mengel~\cite{ChenM15} in an intermediate step of their $\#\W$-hardness proof.

Before we prove that bounded~$\mathsf{dss}$ implies~$\#\W$-easiness, we first give an example query class~$\exclass_{\#\W}$ that fits into this situation.
The query~$\psi_k$ contains as satisfying assignments exactly those tuples~$v_1,\dots,v_k$ of vertices in~$G$ such that there is a length-$2$ walk $v_i w_{ij} v_j$ in~$G$ for any distinct~$i,j$. Formally,
\begin{equation}\label{eq:obstructions_4}
	\exclass_{\#\W} = \setc{\psi_k}{k \in \mathbb{N}},\quad\text{where } \psi_k
	:= x_1\dots x_k:
	\bigwedge_{1\leq i<j\leq k}
	\exists y_{ij}: Ex_iy_{ij}\wedge Ey_{ij}x_j
	\,.
\end{equation}
Note that the graphical representation of~$\psi_k$ corresponds to a subdivided $k$-clique (cf.\ Figure~\ref{fig:ex_table}), and its contract is a $k$-clique.
Thus the queries of~$\exclass_{\#\W}$ and their contracts have unbounded treewidth.
However, the dominating star size is equal to~$2$ because each connected component of~$H[V(G)\setminus X]$ consists of a variable~$y_{ij}$ which has two neighbors.

The $\#\W$-hardness of $\#\homsprob(\exclass_{\#\W})$ as claimed by Theorem~\ref{thm:old_classification} can be proved using a straightforward reduction from counting multicolored cliques of size~$k$, where each edge is subdivided once.
Conversely, we establish that $\#\homsprob(\exclass_{\#\W})$ is $\#\W$-easy by reducing it \emph{to} counting cliques. We prove the special case of~$\exclass_{\#\W}$ here for illustration and then sketch the proof of the general result when the dominating star size is bounded.
\begin{lemma}
	$\#\homsprob(\exclass_{\#\W})$ is $\#\W$-easy.
\end{lemma}
\begin{proof}
	Given $\psi_k$ for some $k\in \mathbb{N}$ and a graph $G$,
	we wish to compute $\#\homs{\psi_k}{G}$.
	We reduce to the problem of counting cliques.
	First, we construct a graph $G'$ from $G$ as follows.
	The vertex set of $G'$ consists of $k$ copies of the vertex set of $G$.
	We add an edge between two vertices $u$ and $v$ in $G'$ if and only if they are contained in different copies and if there exists a vertex $z$ such that $\{u,z\}$ and $\{z,v\}$ are
	edges in $G$.
	Now it can easily be observed that $\#\homs{\psi_k}{G}$ equals $k!$ times the number of cliques of size $k$ in $G'$.
	As counting cliques of size $k$ is the canonical $\#\W$-complete
	problem~\cite{flumgrohe_counting}, this concludes the proof.
\end{proof}

Important in this proof is the preprocessing phase, where for every pair of vertices $u$ and $v$ we check if they have a common neighbor.
After that, we expressed the problem as a homomorphism counting problem without quantified variables.
Indeed, whenever the dominating star size is bounded, the preprocessing works and allows us to get rid of the quantified variables.
The remainder of the reduction to a problem in $\#\W$ follows from the fact that counting answers to model-checking problems without quantified variables is complete for $\#\W$.\footnote{The complexity class $\#\cc{A[1]}$, which is known to be equal to $\#\W$, is defined by the problem of counting answers to first-order formulas without quantified variables. See~\cite[Chapter~14]{flumgrohe} for an overview.}
We remark that, in general, we are not able to show containment in $\#\W$, since we use the oracle for $\#\W$ already in the preprocessing phase.

\begin{theorem}\label{thm:algo_graphs}
	Let $\Delta$ be a class of graphical conjunctive queries with bounded dominating star size.
	Then $\#\homsprob(\Delta)$ is ${\#\W}$-easy.
\end{theorem}
We give a sketch here, the formal proof is deferred to Section~\ref{sec:algo_hyper}.
\begin{proof}[Sketch]
	Let $c\in\N$ be the maximum dominating star size among all queries in~$\Delta$.
	Let $\delta \in \Delta$ be a conjunctive query with free variables $X$ and
	quantified variables $Y$.
	Furthermore let $(H,X)$ be the associated \labeled graph of $\delta$.
	Recall that $H[Y]$ and $H[X]$ are the induced subgraphs of $H$ that only contain vertices
	in $Y$ and in $X$, respectively.
	Furthermore, we let $Y_1,\dots,Y_\ell$ be the connected components of $H[Y]$.
	Given a graph $G$, we wish to compute $\#\homs{\delta}{G}$.
	Since $\mathsf{dss}(H,X)\le c$, the number of vertices in $X$ that are adjacent to a vertex in $Y_i$ in $H$ is bounded by $c$.

	This allows us to perform the following preprocessing:
	For every tuple $\vec{v}=(v_1,\dots,v_c)$ of vertices in $G$ and for every $i \in [\ell]$,
	we check whether $\vec{v}$ is a candidate for the image of the neighbors of $Y_i$ in an
	answer to $\delta$.
	Note that these checks can be done using an oracle for $(\#)\W$ as they can equivalently
	be expressed as a (decision version of a) model checking problem where all variables
	are existentially quantified, which is known to be in $\W$
	(see e.g. Theorem~7.22 in \cite{flumgrohe}).

	After performing all of those checks---at most $\ell \cdot n^c$ many---we need to count
	the number of homomorphisms from $H[X]$ to $G$ that additionally are
	consistent with the checks.
	This final step can be expressed as a counting model checking problem such that every
	variable is free, which is known to be in $\#\W$ (see Chapter~14 in~\cite{flumgrohe}).
\end{proof}

\begin{remark}\label{rem:running_time_graphs}
For each fixed conjunctive query $(H,X)$, we can use Theorem~\ref{thm:algo_graphs} to obtain a deterministic algorithm for computing $\#\homs{H,X}{\star}$: Each oracle query is answered by a subroutine that uses standard dynamic programming over the tree decompositions of $H[Y]$ and the contract of $(H,X)$. The overall running time of the algorithm is bounded by
\begin{equation*}
O\left(n^{\mathsf{dss}(H,X) + \mathsf{tw}(H[Y])+1} + n^{\mathsf{tw}(\mathsf{contract}(H,X))+1}\right) \,.
\end{equation*}
\end{remark}

\subsection{Query Classes That Are \#W[2]-hard}\label{sec:hard_w_2}

We have seen that $\#\homsprob(\Delta)$ is $\#\W$-easy if the dominating star size of~$\Delta$ is bounded.
We now show that the dominating star size is the right parameter for this complexity demarcation, since if it is unbounded for~$\Delta$, then we show the problem to be $\#\Wtwo$-hard. To this end, we will from now on consider its color-prescribed variant. Formally, the problem $\#\cphomsprob(\Delta)$ is given $\delta \in \Delta$ and a $\delta$-colored graph $G$ and the task is to compute $\#\cphoms{\delta}{G}$. It is parameterized by the size of~$\delta$.

As an example, consider the queries~$\psi_k$ in~$\exclass_{\#\Wtwo}$, which have as satisfying assignments exactly the tuples~$v_1,\dots,v_k$ of vertices in~$G$ whose neighborhood contains at least one common vertex. Formally,
\begin{equation}\label{eq:example psi}
	\exclass_{\#\Wtwo} = \setc{\psi_k}{k \in \mathbb{N}}\,,\quad\text{where } \psi_k
	:= x_1 \dots x_k \exists y: \bigwedge_{1 \leq i \leq k} Ex_iy
	\,.
\end{equation}
Note that~$\mathsf{dss}(\psi_k)=k$ holds because the only quantified variable~$y$ has~$k$ neighbors. Thus $\Delta_{\#\Wtwo}$ has unbounded dominating star size.
Moreover, the negated formula~$\neg\psi_k$ on the complement graph~$\overline{G}$ has exactly the dominating sets (or rather, tuples) of size at most~$k$ as its satisfying assignments.
Since the counting problem $\#\cphomsprob(\Delta)$ allows for this negation by subtracting the number of color-prescribed solutions of $\psi_k$ from the number of all possible color-prescribed tuples, it is clear that $\#\cphomsprob(\Delta)$ is indeed $\#\Wtwo$-hard.\linebreak
 Using the same observation, it is also clear that counting solutions of~$\psi_k$ cannot be done in time~$O(n^{k-\varepsilon})$ for any~$\varepsilon>0$ unless SETH is false. This implies the following Lemma.
\begin{lemma}\label{lem:ex_lem_3}
$\#\cphomsprob(\exclass_{\#\Wtwo})$ is $\#\Wtwo$-hard. Furthermore, for every $k \geq 3$, counting answers to $\psi_k$ cannot be done in time $O(n^{k-\varepsilon})$ for any $\varepsilon > 0$ unless SETH fails.
\end{lemma}
\begin{proof}
	We construct a reduction from the problem of counting dominating sets of size $k$, which is known to be $\#\Wtwo$-hard when parameterized by $k$~\cite{flumgrohe_counting} and which cannot be solved in time $O(n^{k-\varepsilon})$ for any $\varepsilon > 0$ assuming SETH holds~\cite{DBLP:conf/soda/PatrascuW10}. Intuitively, the proof exploits that the set of solutions to~$\psi_k$ is in some sense the complement of the set of all $k$-dominating sets and that the ability to count solutions allows us to compute the cardinality of the complementary set.
	Let $k \in \mathbb{N}$ and let $G$ be a graph. It will be convenient to relabel the quantified variable in $\psi_k$ with $0$ and the $k$ free variables with $1,\dots,k$. Recall that a subset $S$ of vertices dominates a graph $G$ if every vertex in $V(G)\setminus S$ is adjacent a vertex in $S$. We first show how to compute the cardinality of the following set using an oracle for $\#\cphoms{\psi_k}{\star}$:
\begin{equation*}
	\mathsf{Dom}_k(G) := \{ a: [k] \rightarrow V(G) ~|~ \mathsf{im}(a) \text{ dominates } G \}\,.
\end{equation*}
	We assume a given graph $G$ to be not complete as otherwise $\mathsf{Dom}_k(G)$ can be computed trivially. Now a $\psi_k$-colored graph $G'$ is constructed from $G$ as follows. First, we take $k+1$ copies $V^0,\dots,V^{k}$ of the vertex set of $G$ and color $V^0$ with the quantified variable $0$ and $V^i$ with the free variable $i$ for $i \in [k]$. Finally, for every $i \in [k]$, we add an edge between a pair of vertices $u \in V^0$ and $v \in V^i$ if and only if the primal vertices of $u$ and $v$ are \emph{not} adjacent in $G$. Observe that $G'$ is indeed $\psi_k$-colored as $G$ is not a complete graph.
	Now let $F$ be the set of all assignments $a$ from $[k]$ to~$V(G')$ such that for all $i \in [k]$ the vertex $a(i)$ is colored with $i$, i.e., contained in $V^i$. Then we have that $\cphoms{\psi_k}{G'} \subseteq F$ and, in particular,
	\begin{align*}
		~&~F \setminus \cphoms{\psi_k}{G'} \\
		=& ~F ~\setminus  ~\{a : [k] \rightarrow V(G') ~|~ a(i) \in V^i \wedge \exists y \in V^0 : \{a(i),y\} \in E(G') \}    \\
		=                                              & ~\{a : [k] \rightarrow V(G') ~|~ a(i) \in V^i \wedge \forall y \in V^0 : \{a(i),y\} \notin E(G') \}
	\end{align*}
	Observe that by construction of $G'$, the cardinality of the latter set is equal to $\#\mathsf{Dom}_k(G)$. As $\#F = \#V(G)^k$ we hence obtain
\begin{equation*}
V(G)^k - \#\cphoms{\psi_k}{G'} = \#\mathsf{Dom}_k(G) \,.
\end{equation*}
Now, given a graph $G$ and $j \in \mathbb{N}$, we define $G^j$ to be the graph obtained from $G$ by adding $j$ isolated vertices. Furthermore we let $\mathsf{Surj}(i,j)$ be the number of surjections from $[i]$ to $[j]$. Then we claim that
	\begin{equation}\label{eq:triangular}
		\#\mathsf{Dom}_k(G^j) = \sum_{i=1}^k \binom{k}{i} \cdot \mathsf{Surj}(i,j) \cdot \#\mathsf{Dom}_{k-i}(G) \,.
	\end{equation}
	To see this we observe that every isolated vertex has to be in the image of every $a \in \mathsf{Dom}_k(G^j)$. Hence we can partition the elements in $\mathsf{Dom}_k(G^j)$ by the number of elements in $[k]$ that are mapped to the isolated vertices. Let $i$ be this number. Then there are $\binom{k}{i}$ possibilities to choose these elements and $\mathsf{Surj}(i,j)$ to map them to the isolated vertices. Finally, the remaining $k-i$ elements have to be mapped to $V(G)$ such that their image dominates $G$.
	We observe that \ref{eq:triangular} yields a system of linear equations such that the corresponding matrix is triangular if proper values for $j$ are chosen. Hence we can compute $\#\mathsf{Dom}_\ell(G)$ for all $\ell \leq k$.\linebreak
	Finally we show how to use these numbers to compute the numbers $D_1,\dots,D_k$ of dominating sets of size $i=1,\dots,k$ in $G$. We proceed inductively. If $i=1$ we have that $D_1=\#\mathsf{Dom}_1(G)$. Otherwise let $\ell\leq k$ and assume that $D_1,\dots,D_{\ell-1}$ have be computed so far. It can easily be seen that
	\begin{equation}
		\#\mathsf{Dom}_\ell(G) = \sum_{i=1}^\ell D_i \cdot \mathsf{Surj}(\ell,i)\,.
	\end{equation}
	Hence $D_\ell = \mathsf{Surj}(\ell,\ell)^{-1} \cdot \left(\#\mathsf{Dom}_k(G)-\sum_{i=1}^{\ell-1} D_i \cdot \mathsf{Surj}(\ell,i) \right)$. The above steps constitute a tight reduction from counting dominating sets of size $k$ to counting solutions to $\#\cphoms{\psi_k}{\star}$ which implies both, the lower bound under SETH and $\#\Wtwo$-hardness of $\#\cphomsprob(\exclass_{\#\Wtwo})$.
\end{proof}

The class~$\exclass_{\#\Wtwo}$ is not only an example of a class for which $\homsprob(\exclass_{\#\Wtwo})$ is $\#\Wtwo$-hard, but it is the minimal one.
Indeed, every class $\Delta$ of unbounded dominating star size contains arbitrarily large elements of~$\exclass_{\#\Wtwo}$ as a minor, and we have already seen that this implies that $\homsprob(\exclass_{\#\Wtwo})$ reduces to~$\homsprob(\Delta)$. Using the fact that counting color-prescribed answers to a conjunctive query is at least as hard as counting color-prescribed answers for any minor of the query (Lemma~\ref{lem:minor_closed}), we are now able to prove the following theorem.
\begin{theorem}\label{thm:wtwo_hardness_graphs}
	Let $\Delta$ be a recursively enumerable class of conjunctive queries with unbounded dominating star size. Then $\#\cphomsprob(\Delta)$ is $\#\Wtwo$-hard. In particular, given a formula $\delta$ with $\mathsf{dss}(\delta) \geq 3$, computing $\#\cphoms{\delta}{\star}$ cannot be done in time $O(n^{\mathsf{dss}(\delta) - \varepsilon})$ for any $\varepsilon > 0$ unless SETH fails.
\end{theorem}
\begin{proof}
	Assume that we are given such a class $\Delta$. As the dominating star size of $\Delta$ is unbounded, we have that for every $k$, there exists $\delta_k \in \Delta$ with dominating star size $\geq k$ and hence $\psi_{k}$ is a minor of $\delta_k$. Therefore we have that the set $\exclass_{\#\Wtwo}$ is a set of minors of $\Delta$. By Lemma~\ref{lem:ex_lem_3} and Lemma~\ref{lem:minor_closed} the claim of the theorem follows.
\end{proof}

\subsection{Query Classes That Are \#A[2]-equivalent and the Excluded-Grate-Theorem}\label{sec:hard_a_2}

Recall that the parameterized complexity class~$\#\Atwo$ is defined via the model checking problem of universally quantified first-order formulas, and it is not known to be equal to~$\#\Wtwo$.
By the same observation as in the preceding subsection, it follows that~$\#\homsprob(\Delta)$ and $\#\cphomsprob(\Delta)$ are $\#\Atwo$-easy, which is made formal in Section~\ref{sec:a2}.
We now introduce the structural parameter \emph{linked matching number} for conjunctive queries that, if unbounded for~$\Delta$, leads to $\#\Atwo$-equivalence.
To define the parameter, we use the notion of a node-well-linked set.

\begin{definition}[Node-well-linked]
	Let $G$ be a graph. A set $S \subseteq V(G)$ is called \emph{node-well-linked} if, for
	every two disjoint and equal-sized subsets~$A,B$ of $S$, there are $\abs{A}$ vertex
	disjoint paths in~$G$ that connect the vertices in $A$ with the vertices in $B$.
\end{definition}

Node-well-linked sets play a central role in the theory of graph minors, particularly in the proof of Chekuri and Chuzhoy~\cite{ChekuriC16,Chuzhoy15} for the Excluded-Grid-Theorem.
Indeed, if large node-well-linked sets exist in a graph, then its treewidth is large and it contains a large grid as a minor.
We now introduce a structural parameter for conjunctive queries that measures the size of the largest set that is node-well-linked in the quantified variables and has a saturating matching to the free variables.

\begin{definition}[Linked matching number]
	Let $(H,X)$ be a conjunctive query, let $Y$ be the set $V(H)\setminus X$ of quantified variables, and let $M$ be a matching from $X$ to $Y$.
	The matching $M$ is called \emph{linked} if the set $V(M) \cap Y$ is node-well-linked in the graph $H[Y]$.
	The \emph{linked matching number} $\lmn$ of $(H,X)$ is defined to be the size of the
	largest linked matching of $H$.
\end{definition}

We prove later that queries with a large linked matching number not only have large treewidth, but they also contain a large \emph{grate} as a minor. Informally, a grate is just half of a grid that lives in the quantified variables and is cut along its diagonal, and the diagonal has a saturating matching to the free variables
(the upper right corner of Figure~\ref{fig:ex_table} depicts the $\kkite{4}$).
\begin{definition}\label{def:grates}
	For a positive integer~$k$, the \emph{\kkite{k}} is the graphical conjunctive query
	whose $k$ free variables are $x^i_j$ for $i,j\ge 0$ with $i+j=k-1$, and whose
	quantified variables are $y^i_j$ for $i,j\ge 0$ with $i+j \leq k-1$.
	The edges between free and quantified variables are $x^i_jy^i_j$ for $i+j=k-1$.
	The edges on the quantified variables are~$y^i_jy^i_{j+1}$ and $y^i_jy^{i+1}_j$.
	Let $\mathsf{Grates}$ be the set of all grates.
\end{definition}

We now sketch a proof that $\#\cphomsprob(\mathsf{Grates})$ is $\#\Atwo$-hard.
The full proof (Sections~\ref{sec:gamma_hard} and~\ref{sec:grates_hard}), requires lifting a rather technical normalization theorem for~$\Atwo$ to~$\#\Atwo$. In particular, this normalization implies that the general problem of counting answers to conjunctive queries in graphs is $\#\Atwo$-equivalent.
\begin{lemma}\label{lem:grates_are_hard}
	The problem $\#\cphomsprob(\mathsf{Grates})$ is $\#\Atwo$-hard.
\end{lemma}
\begin{proof}[Sketch]
	The construction is quite similar to the $\#\W$-hardness proof of the $\mathsf{GridTiling}$ problem in~\cite[Chapter~1.2.2]{radu_dis}. We will reduce from $\#\cphomsprob(\Gamma)$ where $\Gamma$ contains the queries
\begin{equation}
	\gamma_k := x_1 \dots x_k \exists y_1 \dots \exists y_k : \bigwedge_{i=1}^k Ex_iy_i \wedge \bigwedge_{1 \leq i < j \leq k} Ey_iy_j
\end{equation}
for all $k \in \mathbb{N}$. $\#\cphomsprob(\Gamma)$ is shown to be $\#\Atwo$-hard in Section~\ref{sec:gamma_hard}. Intuitively, answers to $\gamma_k$ are vertex sets of size at most $k$ that can be perfectly matched to a clique.

Now let $\omega_k$ be the $k$-grate. Roughly speaking, given a $\gamma_k$-colored graph $G$ for which we want to compute $\#\cphoms{\gamma_k}{G}$, we just need to modify the part of $G$ that is colored with quantified variables. To this end recall that in case of $\gamma_k$, the free variables have to be connected to a clique of quantified variables by a matching and in case of $\omega_k$, the free variables have to be connected to vertices on the diagonal of a half-grid by a matching. Now, given $G$, we obtain a new graph $G'$ by first deleting all edges between vertices that are colored with quantified variables, and then adding blocks of vertices that correspond to the former edges in a half-grid like manner. Then, given two vertices $a$ and $a'$ corresponding to the former edges $\{u,v\}$ and $\{u',v'\}$ we add an edge between $a$ and $a'$ if and only if either $u = u'$ and the blocks of $a$ and $a'$ are adjacent horizontally or $v=v'$ and the blocks of $a$ and $a'$ are adjacent vertically. We encourage the reader to verify the correctness of the construction in the case $k=3$ using Figure~\ref{fig:grate_reduction}.
\end{proof}

\begin{figure}[ptb]
	\centering
	\begin{tikzpicture}[scale=0.5,-,thick]
		\node[circle,inner sep=1pt,fill,label={[label distance = -1mm]315:{\small $1$}}] (1) at (0,8) {};
		\node[circle,inner sep=1pt] (11) at (-0.75,8) {\small $\hdots$};
		\node[circle,inner sep=1pt,fill,label={[label distance = -1mm]315:{\small $2$}}] (2) at (0,4) {};
		\node[circle,inner sep=1pt] (22) at (-0.75,4) {\small $\hdots$};
		\node[circle,inner sep=1pt,fill,label={[label distance = -1mm]315:{\small $3$}}] (3) at (0,0) {};
		\node[circle,inner sep=1pt] (33) at (-0.75,0) {\small $\hdots$};
		\draw(1)--(2);\draw(3)--(2);\draw[bend right=40](1) to (3);

		\node[circle,inner sep=1pt,fill,label={[label distance = -1mm]315:{\small $a$}}] (4) at (4,8) {};
		\node[circle,inner sep=1pt] (44) at (4.75,8) {\small $\hdots$};
		\node[circle,inner sep=1pt,fill,label={[label distance = -1mm]315:{\small $b$}}] (5) at (4,4) {};
		\node[circle,inner sep=1pt] (55) at (4.75,4) {\small $\hdots$};
		\node[circle,inner sep=1pt,fill,label={[label distance = -1mm]315:{\small $c$}}] (6) at (4,0) {};
		\node[circle,inner sep=1pt] (66) at (4.75,0) {\small $\hdots$};
		\draw(1)--(4);\draw(5)--(2);\draw(6) to (3);

		\draw (0,0) circle (35pt);
		\draw (0,4) circle (35pt);
		\draw (0,8) circle (35pt);

		\draw (4,0) circle (35pt);
		\draw (4,4) circle (35pt);
		\draw (4,8) circle (35pt);

		\node[circle,inner sep=1pt,fill,label={[label distance = -2mm]315:{\tiny $(1,1)$}}] (7) at (18,8) {};
		\node[circle,inner sep=1pt] (77) at (17.25,8) {\small $\hdots$};
		\node[circle,inner sep=1pt,fill,label={[label distance = -2mm]315:{\tiny $(2,2)$}}] (8) at (18,4) {};
		\node[circle,inner sep=1pt] (88) at (17.25,4) {\small $\hdots$};
		\node[circle,inner sep=1pt,fill,label={[label distance = -2mm]315:{\tiny $(3,3)$}}] (9) at (18,0) {};
		\node[circle,inner sep=1pt] (99) at (17.25,0) {\small $\hdots$};

		\node[circle,inner sep=1pt,fill,label={[label distance = -2mm]135:{\tiny $(1,2)$}}] (13) at (14,6) {};
		\node[circle,inner sep=1pt] (1313) at (14.75,6) {\small $\hdots$};
		\node[circle,inner sep=1pt,fill,label={[label distance = -2mm]225:{\tiny $(2,3)$}}] (14) at (14,2) {};
		\node[circle,inner sep=1pt] (1414) at (14.75,2) {\small $\hdots$};
		\node[circle,inner sep=1pt,fill,label={[label distance = -2mm]135:{\tiny $(1,3)$}}] (15) at (10,4) {};
		\node[circle,inner sep=1pt] (1515) at (10.75,4) {\small $\hdots$};

		\node[circle,inner sep=1pt,fill,label={[label distance = -1mm]315:{\small $a$}}] (10) at (22,8) {};
		\node[circle,inner sep=1pt] (1010) at (22.75,8) {\small $\hdots$};
		\node[circle,inner sep=1pt,fill,label={[label distance = -1mm]315:{\small $b$}}] (11) at (22,4) {};
		\node[circle,inner sep=1pt] (1111) at (22.75,4) {\small $\hdots$};
		\node[circle,inner sep=1pt,fill,label={[label distance = -1mm]315:{\small $c$}}] (12) at (22,0) {};
		\node[circle,inner sep=1pt] (1212) at (22.75,0) {\small $\hdots$};

		\draw(7)--(13);\draw(7)--(10);\draw(15)--(13);\draw(15)--(14);\draw(8)--(11);
		\draw(8)--(13);\draw(8)--(14);\draw(9)--(14);\draw(9)--(12);

		\draw (10,4) circle (35pt);
		\draw (14,2) circle (35pt);
		\draw (14,6) circle (35pt);

		\draw (22,0) circle (35pt);
		\draw (22,4) circle (35pt);
		\draw (22,8) circle (35pt);

		\draw (18,0) circle (35pt);
		\draw (18,4) circle (35pt);
		\draw (18,8) circle (35pt);

		\node (100) at (1.4,-1.4) { $y_0(G)$};
		\node (100) at (1.4,2.6) { $y_1(G)$};
		\node (100) at (1.4,6.6) { $y_2(G)$};
		\node (100) at (5.4,-1.4) { $x_0(G)$};
		\node (100) at (5.4,2.6) { $x_1(G)$};
		\node (100) at (5.4,6.6) { $x_2(G)$};

		\node (100) at (19.4,-1.4) { $y^2_0(G')$};
		\node (100) at (19.4,2.6) { $y^1_1(G')$};
		\node (100) at (19.4,6.6) { $y^0_2(G')$};
		\node (100) at (23.4,-1.4) { $x^2_0(G')$};
		\node (100) at (23.4,2.6) { $x^1_1(G')$};
		\node (100) at (23.4,6.6) { $x^0_2(G')$};

		\node (100) at (14,7.75) { $y^0_1(G')$};
		\node (100) at (14,0.25) { $y^1_0(G')$};
		\node (100) at (10,5.75) { $y^0_0(G')$};

	\end{tikzpicture}
	\caption{\label{fig:grate_reduction} Illustration of the construction of $G'$ for $k=3$. The graph $G$ (\emph{left}) is $\gamma_3$-colored and the mapping $m=\{x_2 \mapsto a, x_1 \mapsto b, x_0 \mapsto c\}$ is contained in $\cphoms{\gamma_3}{G}$ as $a,b,c$ are connected to $1,2,3$ by a matching and $1,2,3$ form a clique. The graph $G'$ (\emph{right}) is $\omega_3$-colored and $m$ corresponds to the mapping $m'=\{x^0_2 \mapsto a, x^1_1 \mapsto b, x^2_0 \mapsto c\}$. Now $m'$ is contained in $\cphoms{\omega_3 }{G'}$ as $a,b,c$ are connected to $(1,1),(2,2),(3,3)$ by a matching and $(1,1),(2,2),(3,3)$ are the diagonal of a half-grid. Note that the latter is only the case as there are vertices $(1,2),(2,3)$ and $(1,3)$ that correspond to the edges of the clique in $G$.}
\end{figure}
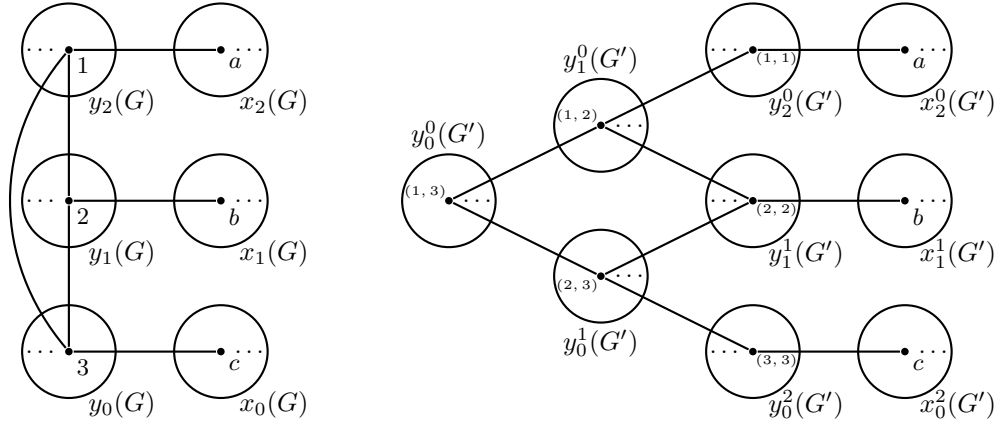

Next we show that every class $\Delta$ with unbounded linked matching number contains arbitrarily large grate minors.
Due to the hardness of $\#\cphomsprob(\mathsf{Grates})$ (Lemma~\ref{lem:grates_are_hard}) and using the fact that the homomorphism counting problem is ``minor-closed'' (Lemma~\ref{lem:minor_closed}), this yields the $\#\Atwo$-hardness of $\#\cphomsprob(\Delta)$.\pagebreak

\noindent We use the work of Marx, Seymour and Wollan~\cite{routedminors} as well as of Diestel et al.~\cite{diestel} for an easy proof of the following version of the Excluded-Grid-Theorem.
The $(g\times g)$-grid $\mathcal{G}_g$ has
vertices $v_{i,j}$ for $i,j\in\set{1,\dots,g}$ and edges $v_{i,j}v_{i+1,j}$ for $i<g$ and $v_{i,j}v_{i,j+1}$ for $j<g$.
A \emph{minor mapping}\footnote{We use the definition from~\cite[Chapter~13.2]{flumgrohe}.} $\eta$ from a graph $H$ to a graph $G$ is a function mapping vertices~$v\in V(H)$ to sets $\eta(v)\subseteq V(G)$ such that the following constraints are satisfied:
\begin{itemize}
	\item For every $v \in V(H)$ the graph $G[\eta(v)]$ is nonempty and connected.
	\item For all $u,v \in V(H)$ with $u \neq v$ the sets $\eta(u)$ and $\eta(v)$ are disjoint.
	\item For all edges $\{u,v\} \in E(H)$ there exist $u' \in \eta(u)$ and $v' \in \eta(v)$ such that $\{u',v'\} \in E(G)$.
\end{itemize}
A set of vertices $S$ of a graph $G$ is called $\ell$\emph{-connected} if, for every pair $A,B$ of disjoint size-$\ell$ subsets of $S$, there are $\ell$ vertex-disjoint paths in $G$ connecting $A$ and $B$. A \emph{separation} of a graph is an ordered pair $(A,B)$ of vertex subsets of $G$ such that $V(A)\cup V(B) = V(G)$ and there are no edges between $A\setminus B$ and $B\setminus A$.

\begin{theorem}[Theorem~\ref{thm:ex_grids_intro} restated] \label{thm:ex_grid_easy}
	For all integers $g>0$ there exists $\kappa \geq 1$ such that the following is true. Let $G$ be a graph and $X\subseteq V(G)$ be a node-well-linked set of size at least $\kappa$. Then there exists a minor mapping $\eta$ from $\mathcal{G}_g$ to $G$ satisfying that for all $j \in [g]$ there exists $x \in X$ such that $x \in \eta(v_{1,j})$.
\end{theorem}
\begin{proof}
	Given a graph $G$, we write $\mathsf{Sep}(G)$ for the set of all separations of $G$. We apply Theorem~1.2 of~\cite{routedminors} with $k=g$. By the theorem, there exists $K\in \mathbb{N}$ such that for any tangle\footnote{For the purpose of this proof we do not need the definition of a tangle. The interested reader is referred to e.g. Chapter~4 in~\cite{diestel}.} $\mathcal{T}$ of order at least $K$ in $G$ and any set $Z\subseteq V(G)$ with $|Z|=g$ the following is true:
	\begin{fact}\label{fac:MSW_theorem}
		If there is no separation $(A,B) \in \mathcal{T}$ with $|V(A \cap B)|< g$ and $Z \subseteq V(A)$ then there is a minor mapping $\eta$ from $\mathcal{G}_g$ to $G$ satisfying that for all $j \in [g]$ there exists $z \in Z$ such that $z\in \eta(v_{1,j})$.
	\end{fact}
	Now let $\ell := \max\{g,K\}$ and $\kappa := 3\ell$. Hence $X$ is an node-well-linked set of size $3\ell$. In particular, $X$ is an $\ell$-connected set of size $3\ell$. Diestel et al. (see Chapter~4 in~\cite{diestel}) have shown that the following is a tangle of order $\ell \geq K$ in $G$:
	\begin{equation}
		\mathcal{T}[X] := \{ (A,B)~|~ (A,B)\in\mathsf{Sep}(G) \wedge |V(A\cap B)| < \ell \wedge |V(A)\cap X| \leq |V(B) \cap X| \}
	\end{equation}
	Next let $Z$ be any subset of $X$ of size $g \leq \ell$ and assume for contradiction that there exists a separation $(A,B) \in \mathcal{T}[X]$ such that $|V(A\cap B)|< g$ and $Z \subseteq V(A)$. By the definition of $\mathcal{T}[X]$ we have that
	\begin{equation}
		|V(B) \cap X| \geq |V(A)\cap X| \geq |Z| = g \,.
	\end{equation}
	Consequently there exists $Z' \in V(B)\cap X$ with $|Z'|=g$ and $Z \cap Z' = \emptyset$. As $X$ is $\ell \geq g$-connected, there are $g$ vertex-disjoint paths from $Z$ to $Z'$. Therefore, by Menger's Theorem, $|V(A\cap B)|< g$ is false and we obtain a contradiction. We hence conclude the proof with the application of Fact~\ref{fac:MSW_theorem} and the observation that $Z$ was chosen to be a subset of $X$.
\end{proof}

\begin{theorem}[Excluded-Grate-Theorem]\label{thm:grate_minor}
  Let $\Delta$ be a class of graphical conjunctive queries.
  If the linked matching number of~$\Delta$ is unbounded, then $\Delta$ contains arbitrarily large grates as minors.
\end{theorem}
\begin{proof}
	Let $g \in \mathbb{N}$ and let $\omega_g$ be the $g$-grate. We show that $\omega_g$ is a minor of some query in $\Delta$. To this end, invoke Theorem~\ref{thm:ex_grid_easy} with $g$ to obtain $\kappa$ for which its statement is true. Now let $\delta \in \Delta$ be a query with linked matching number at least $\kappa$ and let $(H,X)$ be the query graph of $\delta$. By assumption there exists a set $S$ of at least $\kappa$ many vertices in the quantified variables that is node-well-linked in the $H\setminus X$ and that is connected to $X$ by a matching. By Theorem~\ref{thm:ex_grid_easy} there exists a minor mapping $\eta$ from the $g\times g$-grid, such that for every vertex $v$ in the first column of the grid, we have that $\eta(v)$ contains an element of $S$. The grid minor can now further be contracted to obtain a half-grid. Finally, we obtain the $g$-grate $\omega_g$ as a minor by deleting all vertices and edges in $X$ and then all edges between $X$ and $H\setminus Y$ except for the matching connecting $S$ to $X$.
\end{proof}

\begin{theorem}\label{thm:atwo_hardness_graphs}
	Let $\Delta$ be a class of conjunctive queries with unbounded linked matching number. Then $\#\cphomsprob(\Delta)$ is $\#\Atwo$-equivalent.
\end{theorem}
\begin{proof}
	Follows from the $\#\Atwo$-equivalence of $\#\cphomsprob(\mathsf{Grates})$ (Lemma~\ref{lem:grates_are_hard}), the Excluded-Grate-Theorem (Theorem~\ref{thm:grate_minor}) as well as the minor reduction (Lemma~\ref{lem:minor_closed}).
\end{proof}

\subsection{Full Statement of our Complexity Classification for Graphical Conjunctive Queries}

We are now in position to state our main result, the full classification for counting answers to conjunctive queries. Note that Theorem~\ref{thm:extension} is subsumed by the full classification in the case of graphs. The general version, that is, the case of arbitrary logical signatures with bounded arity, is proved in Section~\ref{sec:main_result_structs}.\pagebreak

\begin{theorem}\label{thm:main_thm_cq_graphs}
	Let $\Delta$ be a recursively enumerable class of minimal conjunctive queries.
	\begin{enumerate}
		\item If the treewidth of $\Delta$ and $\mathsf{contract}(\Delta)$ is bounded, then $\#\homsprob(\Delta)$ can be computed in polynomial time.
		\item If the treewidth of $\Delta$ is unbounded and the treewidth of $\mathsf{contract}(\Delta)$ is bounded, then $\#\homsprob(\Delta)$ is $\W$-equivalent.
		\item If the treewidth of $\mathsf{contract}(\Delta)$ is unbounded and the dominating star size of $\Delta$ is bounded, then $\#\homsprob(\Delta)$ is $\#\W$-equivalent.
		\item If the dominating star size of $\Delta$ is unbounded, then $\#\homsprob(\Delta)$ is $\#\Wtwo$-hard.
		Moreover, for any fixed query~$\delta$ with $\mathsf{dss}(\delta) \geq 3$, the problem $\#\homs{\delta}{\star}$ cannot be computed in time $O(n^{\mathsf{dss}(\delta) - \varepsilon})$ for any $\varepsilon > 0$ unless SETH fails.
		\item If the linked matching number of $\Delta$ is unbounded, then $\#\homsprob(\Delta)$ is $\#\Atwo$-equivalent.
	\end{enumerate}
\end{theorem}
\begin{proof}
	The first two claims and the $\#\W$-hardness in the third follow from Theorem~\ref{thm:old_classification}. The $\#\W$-easiness in the third claim follows from Theorem~\ref{thm:algo_graphs} and Lemma~\ref{lem:hom_to_phom}. The fourth and fifth claim follow from Theorem~\ref{thm:wtwo_hardness_graphs} and Theorem~\ref{thm:atwo_hardness_graphs}, respectively, as well as from the fact that $\#\cphomsprob(\Delta)$ reduces to $\#\homsprob(\Delta)$ as shown in Section~\ref{sec:minors_and_colors}.
\end{proof}
Our classification leaves open the question whether every class $\Delta$ that has a bounded linked matching number is in fact $\#\Wtwo$-easy; this question is related to some exotic parameterized complexity classes between $\#\Wtwo$ and $\#\Atwo$.\footnote{The interested reader is encouraged to make themself familiar with the class $\cc{W}^{\mathsf{func}}[2]$ (see Chapt.~8 in~\cite{flumgrohe}) and to observe that strengthening the classification as suggested would imply $\#\Wtwo = \#\cc{W}^{\mathsf{func}}[2]$ or $\#\Atwo = \#\cc{W}^{\mathsf{func}}[2]$.}

\section{Linear combinations of conjunctive queries}\label{sec:lin_combs}

In this section, we extend our results to disjunctions of conjunctive queries, and to conjunctive queries with inequality constraints and negations on the free variables.
We show in this section that both of these extensions are captured by considering abstract linear combinations of conjunctive queries.
To this end, we first adapt the notion of quantum graphs as used by Lovász~\cite[Chapter~6]{lovaszbook} to the setting of graphical conjunctive queries.

\begin{definition}\label{def:q_graphs}
	A \emph{quantum query} $Q$ is a formal linear combination of a finite number of graphical conjunctive queries. We write
	\begin{equation}\label{eq:quantumquery}
		Q = \sum_{i=1}^\ell \lambda_i \cdot (H_i,X_i) \,,
	\end{equation}
	where all $\lambda_i$ are non-zero rational numbers. The \emph{support} of $Q$ is the set $\mathsf{supp}(Q)=\setc{(H_i,X_i)}{i \in \set{1,\dots,\ell}}$. The number of homomorphisms extends to quantum queries linearly, i.e., if $Q$ is a quantum query as above and $G$ is a simple graph, then we define
	\begin{equation}\label{eq: Quantum Hom}
		\#\homs{Q}{G} = \sum_{i=1}^\ell \lambda_i \cdot \#\homs{H_i,X_i}{G} \,.
	\end{equation}
\end{definition}
In the subsequent sections we are going to collect for equivalent queries in a quantum query and hence consider the support to be a set of pairwise non-equivalent and minimal conjunctive queries. The structural parameters discussed in Section~\ref{sec:4} then extend to quantum queries by taking the maximum over all queries in the support.
\begin{definition}
	Let $(H,X)$ and $(\hat H,\hat X)$ be two graphical conjunctive queries.
	\begin{enumerate}
		\item
		      \emph{$(H,X)$ maps surjectively to $(\hat H,\hat X)$}, written $(H,X)\ge (\hat H,\hat X)$, if there is a surjective function~$s:X\to\hat X$ that extends to a homomorphism, that is, which satisfies $s\in\homs{H,X}{\hat H}$.
		      Let $\surj{H,X}{\hat H,\hat X}$ be the set of all surjective mappings~$s:X\to\hat X$ that can be extended to a homomorphism from $H$ to $\hat H$.
		\item
		      If $(H,X)\ge(\hat H,\hat X)\ge(H,X)$ holds, then we adopt the notation of Chen and Mengel~\cite{ChenM16} and say that the two queries are \emph{renaming equivalent}.
		      Moreover, $(H,X)$ is a \emph{minimal representative} if it has a lexicographically smallest pair~$(\abs{V(H)},\abs{E(H)})$ among all queries that are renaming equivalent to~$(H,X)$.
		      For each equivalence class, we arbitrarily fix one minimal representative.
		      If $(H,X)$ is the selected minimal representative of its class, we also call it \emph{renaming minimal}.
	\end{enumerate}
\end{definition}
It is clear that $\ge$ defines a partial order and so this notion of equivalence is indeed an equivalence relation. If $X=\hat X=\emptyset$ holds, then the notion specializes to homomorphic equivalence, whereas for $X=V(H)$ and $\hat X=V(\hat H)$, it specializes to isomorphism. It will turn out that renaming equivalence is identical to equivalence of conjunctive queries as introduced in Section~\ref{sec:minors_and_colors}. In what follows we will therefore omit ``renaming'' and only speak of equivalent and minimal queries.
The next lemma generalizes the fact that all homomorphic cores of a graph are isomorphic.
\begin{lemma}\label{lem:min equiv implies iso}
	If two minimal graphical conjunctive queries are equivalent, then they are isomorphic.
\end{lemma}
\begin{proof}
	Let $(H,X)$ and $(\hat H,\hat X)$ be minimal graphical conjunctive queries that are equivalent.
	By equivalence, we get bijective functions~$s:X\to\hat X$ and $\hat s:\hat X\to X$ that can be extended to homomorphisms~$h$ and~$\hat h$, respectively.
	Let~$F$ be the subgraph of~$\hat H$ that is the image of~$h$, that is we have~$V(F)=h(V(H))$ and $E(F)=h(E(H))$.
	We claim that~$F=\hat H$ must hold by minimality.
	Indeed, when~$\hat h$ is restricted to the vertices of~$F$, it must remain a homomorphism that extends~$\hat s$, and so~$(H,X)$ and~$(F,\hat X)$ are equivalent.
	Minimality implies~$\abs{V(H)}=\abs{V(F)}$ and $\abs{E(H)}=\abs{E(F)}$, and so~$h$ must have every vertex and edge of~$\hat H$ in its image.
	Thus~$h$ is in fact an isomorphism between~$H$ and~$\hat H$, which is what we claimed.
\end{proof}

We can easily express the number of all partial homomorphisms as a linear combination of the number of partial surjective homomorphisms.
\begin{lemma}\label{lem:homs surj}
	For all graphical queries~$(H,X)$ and all graphs~$G$, we have the following identity:
	\begin{equation}
		\#\homs{H,X}{G}=\sum_{Z\subseteq V(G)}\#\surj{H,X}{G,Z}
		\,.
	\end{equation}
\end{lemma}
\begin{proof}
	Every element~$a\in\homs{H,X}{G}$ has a unique set~$Z=a(X)\subseteq V(G)$ such that~$a$ is surjective on~$Z$.
	Thus the sets $\surj{H,X}{G,Z}$ are disjoint for distinct~$Z$, and their union is $\homs{H,X}{G}$, so the claimed identity follows.
\end{proof}
In the following lemma, we show that the functions~$\#\homs{H,X}{\star}$ are linearly independent for all minimal conjunctive queries~$(H,X)$.
It was proved implicitly by Chen and Mengel~\cite{ChenM16}.
\begin{lemma}\label{lem:homs linearly independent}
	Let $k\ge 0$ and let $\mathcal M$ be the (finite) set of all minimal graphical conjunctive queries with at most~$k$ vertices, and let~$\mathcal G$ be the finite set of all (unlabeled) simple graphs with at most~$k^k$ vertices.
	\begin{enumerate}[(i)]
		\item\label{item:L}
		      Let $L$ be the $(\mathcal M\times\mathcal M)$-matrix with
		      $L[(H,X), (\hat H,\hat X)] = \#\surj{H,X}{\hat H,\hat X}$.
		      If we linearly sort~$\mathcal M$ consistent with the partial order~``$\leq$'',
		      then $L$ is a lower-triangular matrix whose diagonal entries are positive integers.
		\item\label{item:B}
		      Let $B$ be the $(\mathcal M\times\mathcal G)$-matrix where
		      $B[(\hat H,\hat X), G]$ is the number of sets~$Z\subseteq V(G)$ such that $(\hat H,\hat X)$ and $(G,Z)$ are equivalent.
		      The matrix $B$ has full rank.
		\item\label{item:A}
		      Let $A$ be the $(\mathcal M\times\mathcal G)$-matrix with
		      $A[(H,X), G] = \#\homs{H,X}{G}$.
		      Then $A=LB$ holds and thus~$A$ has full rank.
	\end{enumerate}
\end{lemma}
\begin{proof}
	First we discuss how to sort the elements of~$\mathcal M$.
	Since $\mathcal M$ only contains minimal queries, any two distinct elements of~$\mathcal M$ are not equivalent, and thus
	$(H,X)\not\geq(\hat H,\hat X)$ or $(H,X)\not\leq(\hat H,\hat X)$ holds.
	Thus we can linearly order~$\mathcal M$ in such a way, that~$(H,X)\not\geq(\hat H,\hat X)$ holds whenever~$(H,X)$ occurs before~$(\hat H,\hat X)$ in the order, and this is the order we choose.

	(\ref{item:L}).
	The identity function~$s:X\to X$ is clearly an element of~$\surj{H,X}{H,X}$, so all diagonal entries of~$L$ are positive integers.
	Now let $(H,X)$ and $(\hat H,\hat X)$ be distinct elements of~$\mathcal M$ such that $(H,X)$ occurs before $(\hat H,\hat X)$ in the linear order and so $(H,X)\not\geq(\hat H,\hat X)$ holds.
	By definition, this means that no surjective function~$X\to\hat X$ extends to a homomorphism from~$H$ to~$\hat H$, which implies $L[(H,X), (\hat H,\hat X)] = 0$.
	Thus, $L$ is lower-triangular.

	(\ref{item:B}).
	The proof is similar to the proof of Lemma~\ref{lem:colorful to uncolored}.
	To prove the claim, we show that each~$(\hat H,\hat X)\in\mathcal M$ has a linear combination of the columns~$\sum_{z} \lambda_z \cdot B[\star,G^z]=1$ such that $\sum_{z} \lambda_z \cdot B[(H,X),G^z]\ne 0$ holds if and only if~$(H,X)=(\hat H,\hat X)$.

	For each~$(\hat H,\hat X)$ and each~$z\in\set{1,\dots,k}^{\hat X}$, we construct graphs~$G^{z}$ as follows:
	Start from~$G:=\hat H$ and clone each vertex~$v\in \hat X$ exactly~$z_v-1$ times (i.e. replace it with an independent set of size~$z_v$ where each vertex has the same neighborhood as~$v$).
	Note that $G^z$ is $\hat H$-colored, and let $c\in\homs{G^z}{\hat H}$ be the coloring.
	Now recall that $B[(H,X),G^z]$ counts the sets~$Z'\subseteq V(G^z)$ such that~$(G^z,Z')$ and $(H,X)$ are equivalent.
	Clearly~$\abs{Z'}=\abs{X}$ must hold for this to be the case.

	\noindent We call~$Z'$ \emph{proper} if $c(Z')=\hat X$ holds, and~\emph{improper} otherwise.
	Moreover, we say that~$Z'$ is $H$-equivalent if~$(H,X)$ and $(G^z,Z)$ are equivalent.
	We have:
	\begin{equation}
		B[(H,X),G^z]
		=
		\#\set{\text{proper $H$-equivalent }Z'} + \#\set{\text{improper $H$-equivalent }Z'}
		\,.
	\end{equation}
	If~$Z'$ is proper, then~$(G^z,Z')$ is equivalent to $(\hat H,\hat X)$ by construction of~$G^z$.
	Thus if a proper~$H$-equivalent~$Z'$ exists, then~$(\hat H,\hat X)=(H,X)$ holds and the number of proper $H$-equivalent~$Z'$ in~$G^z$ is equal to~$\prod_{v\in \hat X} z_v$.
	On the other hand, if~$(\hat H,\hat X)\ne(H,X)$, then the number of proper $H$-equivalent~$Z'$ is equal to zero.
	In any case, the number of improper $H$-equivalent~$Z'$ is a polynomial in the~$z_v$ variables which however does not contain the monomial $\prod_{v\in \hat X} z_v$.
	By multivariate Lagrange interpolation, there is a linear combination~$\sum_z \lambda_z B[(H,X), G^z]$ which is equal to the coefficient of the monomial~$\prod_{v\in\hat X} z_v$. This monomial is zero if and only if~$(H,X)\ne(\tilde H,\tilde X)$.

	(\ref{item:A}).
	The fact that~$A=LB$ holds follows directly from Lemma~\ref{lem:homs surj} by collecting terms for equivalent~$(G,Z)$.
	Since~$L$ is invertible and $B$ has full rank, this also implies that~$A$ has full rank.
\end{proof}

We point out, that Lemma~\ref{lem:homs linearly independent} implies that renaming equivalence of two conjunctive queries is an explicit notion for equivalence. Note that the following was also shown by Chen and Mengel~\cite{ChenM16} with a more complicated proof.
\begin{corollary}\label{cor:counting_equivalence_explicit_graphs}
Two conjunctive queries are renaming equivalent if and only if they are equivalent.
\end{corollary}
\begin{proof}
The forward implication is immediate and the reverse follows from the third item of Lemma~\ref{lem:homs linearly independent}. To see this, we observe that the full rank of $A$ certainly implies that its row vectors are pairwise different.
\end{proof}

\subsection{Complexity monotonicity}
\emph{Complexity monotonicity} informally refers to the following concept:
\begin{quote}
	\begin{minipage}{0.8\textwidth}
		\textit{Computing $\#\homs{Q}{\star}$ is precisely as hard as computing the hardest term $\#\homs{H_i,X_i}{G}$ for which $(H_i,X_i)$ is in the support of $Q$.}
	\end{minipage}
\end{quote}
Complexity monotonicity properties of linear combinations of counting problems were explicitly or implicitly used multiple times in recent publications (see, e.g., \cite{ChenM16,hombasis2017,Roth17,Fockeetal,arxivChen,arxivHolger}).

While the fact that computing $\#\homs{Q}{G}$ is at most as hard as computing the hardest term in the support is trivial --- just compute the sum $\sum_{i=1}^k \lambda_i \cdot \#\homs{H_i,X_i}{G}$ naively --- the reverse implication is usually more involved and was proven independently by Chen and Mengel~\cite{ChenM16} and, in a special case, by Curticapean, Dell, and Marx~\cite{hombasis2017}. We state the result in terms of the existence of a parameterized Turing reduction:

\begin{lemma}[Complexity Monotonicity, implicit in \cite{ChenM16}]
	\label{lem:new_pods_monotonicity}
	Let $Q$ be a quantum query.
	There is an oracle algorithm~$\mathbb A$ that is given~$G$ as input and oracle access to the function~$\#\homs{Q}{\star}$, and computes $\#\homs{H,X}{G}$ for all $(H,X) \in \mathsf{supp}(Q)$ in time $t(|Q|)\cdot n$, where $n = |V(G)|$ and $t$ is a computable function. Furthermore, every oracle query $\#\homs{Q}{G'}$ satisfies $|V(G')| \leq t(|Q|) \cdot n$.
\end{lemma}
\begin{proof}
	Let $k$ be the largest number of vertices among the graphs in the support of~$Q$ and let~$\mathcal M$ be the set from Lemma~\ref{lem:homs linearly independent}.
	Let $G\otimes F$ denote the tensor product of two graphs\footnote{The adjacency matrix of $G\otimes F$ is given by the Kronecker product of the adjacency matrices of $G$ and $F$.} and note that $\homs{H,X}{G\otimes F}=\homs{H,X}{G}\cdot\homs{H,X}{F}$ holds.
	Let $Q=\sum_{H\in\mathcal M} \lambda_H H$ and write $x_H=\lambda_H\cdot\homs{H,X}{G}$.
	Moreover, set $b_F=\homs{Q}{G\otimes F}$ and let $A$ be the matrix from Lemma~\ref{lem:homs linearly independent}.
	Then we have $x^TA=b$.
	To compute the vector~$b$, we simply query the oracle, and the queries have the required size bound.
	The matrix~$A$ and in fact its inverse~$A^{-1}$ can be hard-wired into the algorithm.
	Then $x^T = bA^{-1}$ holds, and we can compute the values~$x_H/\lambda_H$ for $H$ in the support of~$Q$ in the time required.
\end{proof}

We are now ready to lift our classification to quantum queries. The theorem follows from the classification for conjunctive queries (Theorem~\ref{thm:main_thm_cq_graphs}) and the complexity monotonicity property.
\begin{theorem}\label{thm:main_quantum}
	Let $\Delta$ be a recursively enumerable class of quantum queries and let $\hat{\Delta}$ be the set of all minimal conjunctive queries that are contained in the support of some query in $\Delta$.
	\begin{enumerate}
		\item If the treewidth of $\hat{\Delta}$ and $\mathsf{contract}(\hat{\Delta})$ is bounded, then $\#\homsprob(\Delta)$ is fixed-parameter tractable.
		\item If the treewidth of $\hat{\Delta}$ is unbounded and the treewidth of $\mathsf{contract}(\hat{\Delta})$ is bounded, then $\#\homsprob(\Delta)$ is $\W$-equivalent.
		\item If the treewidth of $\mathsf{contract}(\hat{\Delta})$ is unbounded and the dominating star size of $\hat{\Delta}$ is bounded, then $\#\homsprob(\Delta)$ is $\#\W$-equivalent.
		\item If the dominating star size of $\hat{\Delta}$ is unbounded, then $\#\homsprob(\Delta)$ is $\#\Wtwo$-hard.
		Moreover, for any fixed quantum query~$\delta$ with $\mathsf{dss}(\delta) \geq 3$, the problem $\#\homs{\delta}{\star}$ cannot be computed in time $O(n^{\mathsf{dss}(\delta) - \varepsilon})$ for any $\varepsilon > 0$ unless SETH fails.
		\item If the linked matching number of $\hat{\Delta}$ is unbounded, then $\#\homsprob(\Delta)$ is $\#\Atwo$-equivalent.
	\end{enumerate}
\end{theorem}
We remark that, in case of graphs, the classification for quantum queries implies both, Theorem~\ref{thm:data complexity} and Theorem~\ref{thm:main2_int}, if we can express existential and universal positive queries with inequalities and non-monotone constraints over the free variables as quantum queries. This is proved in the subsequent Sections~\ref{sec:cqs_with_ineqs}-\ref{sec:most_general_graphs}. Again the general version for arbitrary logical signatures with bounded arity is deferred to Section~\ref{sec:main_result_structs}.

\subsection{Conjunctive queries with inequalities}\label{sec:cqs_with_ineqs}
In what follows, we will generalize Theorem~\ref{thm:main_thm_cq_graphs} to conjunctive queries that may contain inequalities over free variables\footnote{At the end of this subsection, we argue why a similar result which also takes inequalities into account that may contain quantified variables, would require to solve a long standing open problem in parameterized complexity theory.}. In particular we will show that the support of the resulting quantum query can be given explicitly. Answers to conjunctive queries with inequalities are modeled via partially injective homomorphisms.

\paragraph*{Conjunctive queries with inequalities}
A \emph{conjunctive query with inequalities over the free variables} is a triple $(H,I,X)$ where $(H,X)$ is a conjunctive query and $I$ is an irreflexive and symmetric relation $I \subseteq X^2$. We say that $I$ is a set of inequalities. Intuitively, given a graph $G$ and a conjunctive query with inequalities $(H,I,X)$, an assignment $a: X \rightarrow V(G)$ is an answer to $(H,I,X)$ if and only if $a$ is an answer to $(H,X)$ and, additionally, for every inequality $(x,x') \in I$, it holds that $a(x) \neq a(x')$. Formally, we define the set of answers to $(H,I,X)$ in terms of partially injective homomorphisms
\begin{equation*}
 \partinj{H,I,X}{G} := \setc{ a \in \homs{H,X}{G} }{ \forall (x,x') \in I : a(x) \neq a(x')} \,.
\end{equation*}
If there are no quantified variables, this definition coincides with the notion of graphically restricted homomorphisms in~\cite{Roth17}.

We will use the following contraction operation induced by the subsets of $I$. Given a conjunctive query $(H,X)$ and a set $\sigma \subseteq I$ the \emph{contracted query} $(H/\sigma,X/\sigma)$ is obtained by identifying every pair of vertices $x$ and $\hat{x}$ as a single vertex for every inequality $(x,\hat{x})\in \sigma$. Multiple edges are deleted and self-loops are kept. We point out that it is possible that the contraction of all pairs in $\sigma$ might also contract vertices $x$ and $\hat{x}$ that are not contained in $\sigma$. Consider for example $\sigma = \{ \{x_1,x_2\},\{x_2,x_3\} \}$. Then contracting $\{x_1,x_2\}$ and $\{x_2,x_3\}$ will also contract~$x_1$ and~$x_3$.

\begin{theorem}
	\label{thm:quantum_graph_for_ineqs}
	Let $\chi = (H,I,X)$ be a conjunctive query with inequalities over the free variables. Then there exists a quantum query $\qgraphof{\chi}$ such that
\begin{equation*}
 \#\partinj{H,I,X}{\star} = \#\homs{\qgraphof{\chi}}{\star} \,.
\end{equation*}
Furthermore, the mapping $\chi \mapsto \qgraphof{\chi}$ is computable and the support of $\qgraphof{\chi}$ is, up to equivalence, the set of all contracted queries $(H/\sigma,X/\sigma)$ where $\sigma$ is a subset of $I$.
\end{theorem}

In other words, given $\chi = (H,I,X)$, we can contract arbitrary variables in $(H,X)$ that are connected by an inequality in $I$. Then Theorem~\ref{thm:quantum_graph_for_ineqs} guarantees that a minimal equivalent of the resulting query is contained in the support of the quantum query $\qgraphof{\chi}$. The proof requires matroid and lattice theory and is hence, together with some further preliminaries, encapsulated in Subsection~\ref{sec:proof_ineqs}.

We remark that a general theorem in the above form that also includes inequalities over quantified variables remains elusive, as this would require to completely understand the subgraph decision problem, which is one of the most famous open problems in parameterized complexity (see e.g. Chap.~13 in~\cite{flumgrohe}). In terms of conjunctive queries with inequalities, the subgraph decision problem can be formulated by a query without free variables and with all inequalities over the quantified variables. Then the empty assignment is in the set of solutions if and only if there is an injective homomorphism, that is, a subgraph embedding from the quantified variables to the host graph.

\subsubsection{Matroid lattices and the proof of Theorem~\ref{thm:quantum_graph_for_ineqs}}\label{sec:proof_ineqs}

The proof is in the same spirit as in \cite{Roth17}; the key idea is that the coefficients of the quantum queries can be computed using the Möbius function over the lattice of flats of the graphic matroid induced by the inequalities. Hence we first proceed with a detour to matroid theory.

\paragraph*{Matroids} We will follow the definitions of Chapter~1 of the textbook of Oxley \cite{oxley}. A \emph{matroid} $M$ is a pair $(E,\mathcal{I})$ where $E$ is a finite set and $\mathcal{I}\subseteq \mathcal{P}(E)$ such that (1) $\emptyset \in \mathcal{I}$,
(2) if $A \in \mathcal{I}$ and $B \subseteq A$ then $B \in \mathcal{I}$, and (3) if $A,B \in \mathcal{I}$ and $|B|<|A|$ then there exists $a \in A \setminus B$ such that $B\cup \bs{a} \in \mathcal{I}$. We call $E$ the \emph{ground set} and an element $A \in \mathcal{I}$ an \emph{independent set}. A maximal independent set is called a \emph{basis}. The \emph{rank} $\mathsf{rk}(M)$ of $M$ is the size of its bases\footnote{This is well-defined as every maximal independent set has the same size due to (3).}.

Given a subset $X \subseteq E$ we define $\mathcal{I}|X := \bs{A \subseteq X ~|~A \in \mathcal{I}}$. Then $M|X := (X,\mathcal{I}|X)$ is also a matroid and called the restriction of $M$ to $X$. Now the \emph{rank} $\mathsf{rk}(X)$ of $X$ is the rank of $M|X$. Equivalently, the rank of $X$ is the size of the largest independent set $A \subseteq X$.
Furthermore we define the \emph{closure} of $X$ as follows:
\begin{equation}\label{eq:closures}
	\mathsf{cl}(X) := \bs{e \in E ~|~ \mathsf{rk}(X \cup \bs{e}) = \mathsf{rk}(X)} \,.
\end{equation}
Note that by definition $\mathsf{rk}(X)=\mathsf{rk}(\mathsf{cl}(X))$. We say that $X$ is a \emph{flat} if $\mathsf{cl}(X)=X$. We denote $L(M)$ as the set of flats of $M$. It holds that $L(M)$ together with the relation of inclusion is a lattice, called the \emph{lattice of flats} of $M$. The least upper bound of two flats $X$ and $Y$ is $\mathsf{cl}(X \cup Y)$ and the greatest lower bound is $X \cap Y$. It is known that the lattices of flats of matroids are exactly the geometric lattices\footnote{For the purpose of this paper we do not need the definition of geometric lattices but rather the equivalent one in terms of lattices of flats and therefore omit it. We recommend e.g. Chapter~3 of \cite{welsh} and Chapter~1.7 of \cite{oxley} to the interested reader.} and we denote the set of those lattices as $\mathcal{L}$.
Given a graph $H=(V,E)$, the \emph{graphic matroid} $M(H)$ has ground set $E$ and a set of edges is independent if and only if it does not contain a cycle.
If $H$ is connected then a basis of $H$ is a spanning tree of $H$. If $H$ consists of several connected components then a basis of $M(H)$ induces spanning trees for each of those. Every subset $\rho$ of $E$ induces a partition of the vertices of $H$ where the blocks are the vertices of the connected components of $H|_\rho=(V(H),\rho)$ and it holds that
\begin{equation}
	\label{eqn:graphic_rank}
	\mathsf{rk}(\rho) = |V(H)| - c(H|_\rho) \,,
\end{equation}
where $c(H|_\rho)$ is the number of connected components of $H|_\rho$. In particular, the flats of $M(H)$ correspond bijectively to the partitions of vertices of $H$ into connected components as adding an element to $\rho$ such that the rank does not change, will not change the connected components, too. For convenience we will therefore abuse notation and say, given an element $\rho$ of the lattice of flats of $M(H)$, that $\rho$ partitions the vertices of $H$ where the blocks are the vertices of the connected components of $H|_\rho$. The following observation will be useful in the remainder of this section.
\begin{lemma}
	\label{lem:same_blocks_same_ranks}
	Let $\rho,\sigma \in L(M(H))$ for a graph $H$. If the number of blocks of $\rho$ and $\sigma$ are equal then $\mathsf{rk}(\rho)=\mathsf{rk}(\sigma)$.
\end{lemma}
\begin{proof}
	Immediately follows from Equation~(\ref{eqn:graphic_rank}).
\end{proof}

We now have everything we need to proceed with the proof of Theorem~\ref{thm:quantum_graph_for_ineqs}.
\begin{proof}
	Let $G$ be a graph. We will prove that
	\begin{equation}\label{eq:main_eq_ineqs}
		\#\partinj{H,X,I}{G} = \sum_{\rho \in \mathcal{L}(M(X,I))} \mu(\emptyset,\rho) \cdot \#\homs{H/\rho,X/\rho}{G} \,,
	\end{equation}
	where $\mu$ is the Möbius function of the lattice of flats $\mathcal{L}(M(X,I))$ of the graphic matroid $M(X,I)$. To avoid a lengthy introduction to the concept of Möbius inversion,\footnote{We refer the interested reader to~\cite{lovaszbook} where Möbius inversion is introduced and used in a similar setting.} we point out that we only need the following two properties of $\mu$:
	\begin{enumerate}
		\item[1.] Boolean Expansion Formula (see e.g. Proposition~7.1.4 in~\cite{zaslavsky87}): For every $\rho \in \mathcal{L}(M(X,I))$ it holds that
		      \begin{equation}\label{eq:boolean_expansion}
			      \mu(\emptyset,\rho) =\sum_{\substack{\sigma \subseteq I\\ \mathsf{cl}(\sigma)=\rho}} (-1)^{\#\sigma} \,.
		      \end{equation}
		\item[2.] Corollary of Rota's NBC Theorem \cite{rota1964foundations}: For every $\rho \in \mathcal{L}(M(X,I))$ it holds that
		      \begin{equation}\label{eq:rota}
			      \mathsf{sgn}(\mu(\emptyset,\rho)) = (-1)^{\mathsf{rk}(\rho)} \,.
		      \end{equation}
	\end{enumerate}
	For the proof of~\ref{eq:main_eq_ineqs} we will apply the principle of inclusion-exclusion in the following way: First, we compute $\#\homs{H,X}{G}$, then we subtract the number of elements in $\#\homs{H,X}{G}$ that violate at least one inequality, then we add the number of elements that violate at least two inequalities and so forth. It will be convenient to define the following subsets of partial homomorphisms that are subject to a set of vertex identifications corresponding to violated inequalities. To this end, let $\rho \subseteq I$ and define
	\begin{equation}
	\homs{H,X}{G}[\rho] := \{a \in \homs{H,X}{G} ~|~ \forall (x,\hat{x}) \in \rho : a(x)=a(\hat{x}) \}\,.
	\end{equation}
	Next we observe that $\#\homs{H,X}{G}[\rho]$ can be computed by applying the contraction operation according to $\rho$ and that $\#\homs{H,X}{G}[\rho]$ and $\#\homs{H,X}{G}[\sigma]$ are equal whenever $\mathsf{cl}(\rho) = \mathsf{cl}(\sigma)$.
	\begin{fact}\label{fac:ineqs_claim_1}
		For every pair of sets $\rho, \sigma \subseteq I$ with $\mathsf{cl}(\rho) = \mathsf{cl}(\sigma)$, we have:
		\begin{align}\label{eq:block_identity}
			\#\homs{H/\rho,X/\rho}{G}&= \#\homs{H,X}{G}[\rho] \,,\text{ and}\\
		\label{eq:block_equality}
			\#\homs{H/\rho,X/\rho}{G}&= \#\homs{H/\sigma,X/\sigma}{G} \,.
		\end{align}
	\end{fact}
	We proceed with the following, aforementioned application of the principle of inclusion and exclusion. \begin{small} \begingroup \allowdisplaybreaks
	\begin{align*} 
		~ & \#\partinj{H,I,X}{G}                                                                                                                                                        \\
		~ & =\#\homs{H,X}{G} -\#\bs{a \in \homs{H,X}{G} ~|~ \exists \bs{x,\hat{x}} \in I:  a(x) = a(\hat{x})}                                                                           \\
		~ & =\#\homs{H,X}{G} -\#\left(\bigcup_{(x,\hat{x})\in I} \bs{a \in \homs{H,X}{G} ~|~a(x) = a(\hat{x})} \right)                                                                  \\
		~ & =\#\homs{H,X}{G} - \sum_{\emptyset \neq \rho \subseteq I} (-1)^{\#\rho-1} \cdot \#\left(\bigcap_{(x,\hat{x})\in \rho} \bs{a \in \homs{H,X}{G} ~|~a(x) = a(\hat{x})} \right) \\
		~ & =\sum_{\rho \subseteq I} (-1)^{\#\rho} \cdot\#\homs{H,X}{G}[\rho]                                                                                                           \\
		~ & \stackrel{\ref{eq:block_identity}}{=}\sum_{\rho \subseteq I} (-1)^{\#\rho} \cdot\#\homs{H/\rho,X/\rho}{G}                                                                   \\
		~ & \stackrel{\ref{eq:block_equality}}{=}\sum_{\rho \in \mathcal{L}(M(X,I))} \sum_{\substack{\sigma \subseteq I                                                                 \\ \mathsf{cl}(\sigma)=\rho}} (-1)^{\#\sigma} \cdot\#\homs{H/\rho,X/\rho}{G}\\
		~ & \stackrel{\ref{eq:boolean_expansion}}{=}\sum_{\rho \in \mathcal{L}(M(X,I))} \mu(\emptyset,\rho) \cdot\#\homs{H/\rho,X/\rho}{G}
	\end{align*} \endgroup
		\end{small}
	Now let $\llbracket H_1,X_1 \rrbracket,\dots,\llbracket H_k,X_k \rrbracket$ be the equivalence classes of the set \[\{ (H/\rho,X/\rho)~|~\rho \in \mathcal{L}(M(X,I))\}\] with minimal representatives. Then, we can define the desired quantum query to be 
	\begin{equation*}
	\qgraphof{\chi}:= \sum_{i=1}^k \lambda_i \cdot (H_i,X_i)~~~~~\text{ where }~~~
	\lambda_i = \sum_{\substack{\rho \in  \mathcal{L}(M(X,I))\\(H/\rho,X/\rho)~ \sim~ (H_i,X_i)}} \mu(\emptyset,\rho) \,.
	\end{equation*}
	It remains to show that for all $i \in [k]$ we have that $\lambda_i \neq 0$. To this end, we observe that $(H/\rho,X/\rho) \sim (H/\sigma,X/\sigma)$ implies that $\#X/\rho=\#X/\sigma$ and hence that $\rho$ and $\sigma$ have the same number of blocks with respect to the graphic matroid $M(X,I)$. Therefore, by Lemma~\ref{lem:same_blocks_same_ranks}, $\mathsf{rk}(\rho)=\mathsf{rk}(\sigma)$. Now fix $i$ and let $r_i$ be the rank of the flats that contribute to $\lambda_i$. Using Equation~\ref{eq:rota}, we obtain that $\mathsf{sgn}(\mu(\emptyset,\rho)) =(-1)^{r_i}$ for all $\rho$ such that $(H/\rho,X/\rho) \sim (H_i,X_i)$. Consequently, $\mathsf{sgn}(\lambda_i) = (-1)^{r_i}$ and hence, $\lambda_i \neq 0$. Finally, we have that for every \emph{subset} $\sigma$ of $I$, there exists a flat $\rho$ of $M(X,I)$ such that $(H/\rho,X/\rho) = (H/\sigma,X/\sigma)$. In particular it can be observed that $\rho = \mathsf{cl}(\sigma)$. This concludes the proof.
\end{proof}

\subsection{Existential and universal positive formulas with inequalities}
\label{sec:ep_up_ineqs_graphs}
In this subsection we will lift the classification once more, namely to existential and universal positive formulas with inequalities over the free variables. Our goal is hence to find quantum queries~$Q$ that allow us to express the number of solutions to the more general queries as $\#\homs{Q}{\star}$. To this end, we provide a concise introduction to existential and universal positive queries. We refer the reader e.g. to Chapter~4 of~\cite{flumgrohe} for a detailed introduction to the semantics of first-order formulas. An \emph{existential positive formula} (or \emph{query}) is of the form
\begin{equation*}
\psi = x_1 \dots x_k \exists y_1\dots\exists  y_\ell : \psi' \,,
\end{equation*}
and a \emph{universal positive formula} is of the form
\begin{equation*}
\theta = x_1 \dots x_k \forall y_1\dots\forall y_\ell : \theta' \,.
\end{equation*}
Here $X=\{x_1,\dots,x_k\}$ is a set of free variables and $Y=\{y_1,\dots,y_\ell\}$ is a set of quantified variables. Furthermore, $\psi'$ and $\theta'$ are inductively build on atoms $a_i$ and logical connectives~$\vee$ and~$\wedge$ (but without negations). As in case of conjunctive queries, the atoms are of the form $Evv'$ for $v,v' \in X \cup Y$.

Given a graph $G$, an assignment $a: X \rightarrow V(G)$ is a solution to $\psi$ if there is an assignment $h: X \cup Y \rightarrow V(G)$ such that $h|_X=a$ and $h$ satisfies $\psi'$. Similarly, $a$ is a solution to $\theta$, if for all assignments $h: X \cup Y \rightarrow V(G)$ such that $h|_X=a$ it holds that $h$ satisfies $\theta'$. Here, an assignment $h: X \cup Y \rightarrow V(G)$ satisfies an atom $Evv'$ if and only if $\{h(v),h(v')\} \in E(G)$. The semantics of $\vee$ and $\wedge$ are defined inductively in the canonical way. We write $h\models_G \varphi$ if $h$ satisfies a formula $\varphi$ and we write $\psi(G)$ for the set of all solutions to $\psi$ and $\theta(G)$ for the set of all solutions to $\theta$, both with respect to $G$. In what follows we denote $\Sigma^+_1$ as the set of all existential positive queries and $\Pi^+_1$ as the set of all universal positive queries. Now, given a class $\Delta$ of formulas, we define $\#\pmc(\Delta)$ to be the problem of, given a graph $G$ and $\delta \in \Delta$, computing $\#\delta(G)$. It is parameterized by $|\delta|$. Note that $\#\pmc(\Delta)$ coincides with $\#\homsprob(\Delta)$ if $\Delta$ is a set of conjunctive queries and that $\#\pmc(\Delta)$ coincides with $\#\partinjprob(\Delta)$ if $\Delta$ is a set of conjunctive queries with inequalities over the free variables. We point out that we are going to revisit those notions in Section~\ref{sec:a2} where we deal with general parameterized model checking and model counting.

The following is due to Chen and Mengel --- we state their result in terms of quantum queries.
\begin{theorem}[\cite{ChenM16}] \label{thm:chen_mengel_ep}
	Let $\psi \in \Sigma^+_1$ be an existential positive query. Then there exists a quantum query $\qgraphof{\psi}$ such that for every graph $G$ it holds that $\#\psi(G) = \#\homs{\qgraphof{\psi}}{G}$. Furthermore, the mapping $\psi \mapsto \qgraphof{\psi}$ is computable.
\end{theorem}
We observe that this result can easily be extended to universal positive queries.
\begin{corollary}\label{cor:up_to_ep}
	Let $\theta \in \Pi^+_1$ be a universal positive query. Then
	there exists a quantum query $\qgraphof{\theta}$ such that for every graph $G$ it holds that $\#\theta(G) = \#\homs{\qgraphof{\theta}}{\overline{G}}$. Furthermore, the mapping $\theta \mapsto \qgraphof{\theta}$ is computable.
\end{corollary}
\begin{proof}
	Let $\theta = x_1,\dots x_k \forall y_1,\dots y_\ell : \theta'$. Without loss of generality assume that $\theta' = \bigvee_{i=1}^{m_1} \bigwedge_{j=1}^{m_2} Ev_iv_j$. For every graph $G$ with $n$ vertices, it holds that \[\#\theta(G) = n^k - \#\{a: X \rightarrow V(G) ~|~\exists h: X \cup Y \rightarrow V(G): h|_X=a \wedge  h \nvDash_G \theta' \}\,,\] i.e., we can count all assignments $a$ that can be extended to an assignment $h$ that does not satisfy~$\theta'$ and subtract this number from the number $n^k$ of all assignments from $X$ to $V(G)$. Now, using DeMorgan's Law, it holds that $h \nvDash_G \theta'$ if and only if $h \models_{\overline{G}} \overline{\theta'}$, where $\overline{\theta'} = \bigwedge_{i=1}^{m_1} \bigvee_{j=1}^{m_2} Ev_iv_j$. Finally, we let $\overline{\theta} = x_1,\dots x_k \exists y_1,\dots y_\ell : \overline{\theta'}$ and obtain
	\[
		\#\theta(G)  =n^k - \#\{a: X \rightarrow V(G) ~|~\exists h: X \cup Y \rightarrow V(G): h|_X=a \wedge  h \nvDash_G \theta' \} 
		~            =n^k - \#\overline{\theta}(\overline{G})
	\]
	As $\overline{\theta}$ is an existential positive query, we can apply Theorem~\ref{thm:chen_mengel_ep} to compute $\qgraphof{\overline{\theta}}$. Furthermore, it holds that $\#\homs{\mathsf{IS}_k}{\overline{G}} = n^k$, where $\mathsf{IS}_k$ is the graph consisting of vertices $[k]$ without edges. We conclude by setting $\qgraphof{\theta} = (\mathsf{IS}_k,[k]) - \qgraphof{\overline{\theta}}$.
\end{proof}

Now the generalization to existential und universal positive formulas with inequalities over the free variables is straightforward. Formally, we equip the formulas with an additional set $I$ of inequalities as we did for conjunctive queries in Section~\ref{sec:cqs_with_ineqs} and we count only those answers that satisfy all inequality constraints. The definition of $\#\pmc$ is lifted accordingly. Using the inclusion-exclusion principle we can express this number as a linear combination of existential or universal positive formulas, respectively, without inequalities --- the proof is completely analogous to the case of conjunctive queries in Section~\ref{sec:cqs_with_ineqs}. Finally, we apply Corollary~\ref{cor:up_to_ep} or Theorem~\ref{thm:chen_mengel_ep}, depending on whether we are considering existential or universal positive queries, and collect equivalent terms to obtain a quantum query.

\noindent We remark that, in contrast to conjunctive queries with inequalities, the support of the resulting quantum queries and hence the criteria for the classification (Theorem~\ref{thm:main_quantum}) cannot be given explicitly, which is due to the fact that the cancellation behavior in the transformation of Theorem~\ref{thm:chen_mengel_ep} is not yet understood. More precisely, there is no explicit criterion for a conjunctive query to be contained in the support of the quantum query in case an existential positive formula is expressed as quantum query.

\subsection{Non-monotone constraints over free variables}\label{sec:most_general_graphs}
Last but not least we will lift the classification theorem to existential and universal positive queries with inequalities over the free variables that additionally may contain non-monotone constraints of the form $\neg Ex\hat{x}$ over free variables. The idea is quite simple: Just perform inclusion-exclusion over the non-monotone constraints. Unfortunately, this requires us to circumvent the following tedious technicality: We have to guarantee that a transformation of existential or universal positive queries with non-monotone constraints over free variables to a linear combination of conjunctive queries does not create queries that contain non-monotone contraints over quantified variables. This latter issue will be dealt with by taking a closer into the proof of Theorem~\ref{thm:chen_mengel_ep} by Chen and Mengel~\cite{ChenM16}.

We start by considering conjunctive queries with non-monotone constraints over the free variables. For technical reasons, given a formula $\varphi$, a set $J$ of atoms containing only variables that are new or free in $\varphi$ and a set $\mathcal{V}=\{v_1,\dots,v_k\}$ disjoint from all variables in $\varphi$ and~$J$, we define
\[[\varphi \wedge J]^\mathcal{V} := v_1,\dots,v_k:\varphi \wedge \bigwedge_{a \in J} a   \,. \]
We write $[\varphi \wedge J]$ if $\mathcal{V}$ is empty and $[\varphi]^\mathcal{V}$ if $J$ is empty.
Now let $G$ be a graph and let $\varphi$ be a conjunctive query with non-monotone constraints
\begin{equation}\label{eq:sdef}
\neg S = \{\neg s_1 , \neg s_2 , \dots , \neg s_\ell \} 
\end{equation} 
where each $s_i$ is an atom $Ex_i\hat{x}_i$ for some free variables $x_i$ and $x'_i$. 

\noindent As $\varphi$ is a conjunctive query, there is a subquery $\delta$ without non-monotone constraints such that $\varphi = [\delta \wedge \neg S]$.
 Now observe that
\begin{equation}\label{eq:non_monotone_technicality}
\varphi(G) = [\delta\wedge \neg S](G)=  \left\lbrace a \in [\delta ]^{\mathcal{V}(\neg S \setminus \delta)}(G) ~\middle|~ \bigwedge_{i=1}^\ell \{a(x_i),a(\hat{x}_i)\} \notin E(G) \right\rbrace \,,
\end{equation} 
where $\mathcal{V}(\neg S \setminus \delta)$ is the set of variables occuring only in $\neg S$ and not $\delta$.

\noindent In other words,~\eqref{eq:non_monotone_technicality} states that the assignments satisfying $\varphi$ are precisely those assignments that satisfy $\delta$ and all non-monotone constraints in $\neg S$. As, however, it might be possible that there are free variables in $\varphi$ that only occur in the non-monotone part $\neg S$, we have to extend $\delta$ by those variables.

Again, we will use the principle of inclusion and exclusion to first get rid of the non-monotone constraints and then build up on the prior transformations to quantum queries.

\begin{lemma}\label{thm:non_monotone_constraints_cqs}
Let $\varphi=[\delta \wedge \neg S]$ be a conjunctive query with non-monotone constraints $\neg S$ as given by~\eqref{eq:sdef}. Then we have that
\[\#\varphi(\star) = \sum_{J \subseteq S} (-1)^{\#J} \cdot \#[\delta \wedge J]^{\mathcal{V}(S,\delta,J)}(\star)  \,,\]
where $S := \{s_1 , s_2 , \dots , s_\ell \}$ and $\mathcal{V}(S,\delta,J)$ is the set of all variables occurring in $S$ but neither in $\delta$ nor in $J$.
\end{lemma}
\begin{proof}
Let $G$ be a graph. Using inclusion-exclusion, we obtain that
	\begin{align*}
		\#\varphi(G) &= \#[\delta \wedge \neg S](G)\\
		~ & = \#\left\lbrace a \in [\delta ]^{\mathcal{V}(\neg S \setminus \delta)}(G) ~\middle|~ \bigwedge_{i=1}^\ell \{a(x_i),a(\hat{x}_i)\} \notin E(G) \right\rbrace \\
		              & = \#[\delta ]^{\mathcal{V}(\neg S \setminus \delta)}(G) - \#\left\lbrace a \in [\delta ]^{\mathcal{V}(\neg S \setminus \delta)}(G) ~\middle|~ \bigvee_{(Ex\hat{x})\in S} \{a(x),a(\hat{x})\} \in E(G) \right\rbrace                                                         \\
		              & = \sum_{J \subseteq S} (-1)^{\#J} \cdot \#\left\lbrace a \in [\delta ]^{\mathcal{V}(\neg S \setminus \delta)}(G) ~\middle|~ \bigwedge_{(Ex\hat{x})\in J} \{a(x),a(\hat{x})\} \in E(G) \right\rbrace \\
		              & =\sum_{J \subseteq S} (-1)^{\#J} \cdot \#[\delta \wedge J]^{\mathcal{V}(S,\delta,J)}(G)\,.
	\end{align*}
\end{proof}

Our next goal is to generalize to existential and universal positive formulas with inequalities and non-monotone constraints over the free variables. To this end, we wish to invoke the transformation given by Theorem~\ref{thm:chen_mengel_ep}. However, the statement of the latter theorem does formally not apply to formulas with non-monotone constraints. 

\noindent To circumvent this issue, we will just add a the relation symbol $\overline{E}$ to the signature of graphs; we will argue in Chapter~\ref{sec:structures} that all results for the signature of graphs readily extend to arbitrary signatures of bounded arity.
Now let $\psi$ be an existential or universal positive formula over the signature of graphs with non-monotone constraints of the form $\neg Ex\hat{x}$. The formula $\psi_\uparrow$ is obtained from $\psi$ by substituting every atom $\neg Ex\hat{x}$ by $\overline{E}x\hat{x}$, where $\overline{E}$ is a new relation symbol of arity $2$. Consequently, the signature of $\psi_\uparrow$ is $\tau=(E,\overline{E})$. Similarly, given a graph $G$, that is, a structure over the signature $(E)$, we let $G_\uparrow$ be the following structure over signature $\tau$: The vertices $V(G_\uparrow)$ of $G_\uparrow$ are precisely the vertices of $G$ and a pair $(x,\hat{x})$ is in~$E(G_\uparrow)$ if and only if $\{x,\hat{x}\}$ is an edge of $G$. Furthermore, a pair $(x,\hat{x})$ is in~$\overline{E}(G_\uparrow)$ if and only if $\{x,\hat{x}\}$ is not an edge of $G$.

The operation $\downarrow$ is defined analogously: Given an existential or universal positive formula $\varphi$ over the signature $\tau$, we obtain the formula $\varphi_\downarrow$ from~$\varphi$ by substituting every atom $\overline{E}x\hat{x}$ by a non-monotone constraint $\neg Ex\hat{x}$. Consequently, the signature of $\varphi_\downarrow$ is $(E)$. Similarly, given a structure $G$ over signature $\tau$, we obtain a graph $G_\downarrow$ without self-loops from $G$ by taking the same set of vertices and adding an edge $\{x,\hat{x}\}$ to $E(G_\downarrow)$ if and only if $x\neq \hat{x}$ and $(x,\hat{x})\in E(G)$. The following is immediate.

\begin{fact}\label{fac:help_monotonicity}
We have that
\begin{enumerate}
\item $\psi_I(G) = \psi_{\uparrow I}(G_\uparrow)$ for every graph~$G$ without self-loops, existential or universal positive formula $\psi$ over the signature of graphs, and inequalities $I$ over the free variables of $\psi$,
\item $\varphi_I(G) = \varphi_{\downarrow I}(G_\downarrow)$ for every structure~$G$ and existential or universal positive formula $\varphi$ over the signature $(E,\overline{E})$, and for every set of inequalities $I$ over the free variables of $\varphi$,
\item and $G_{\uparrow\downarrow} = G$ for every graph $G$ without self-loops.
\end{enumerate}
\end{fact}

\begin{theorem}\label{thm:main_quantum_query}
Let $\psi$ be an existential or universal positive formula with non-monotone constraints over the free variables and let $I$ be a set of inequalities over the free variables of~$\psi$. There exists a quantum query $Q[\psi,I]$ satisfying that
\[\#\psi_I(\star) = \#\homs{Q[\psi,I]}{\star}\,.\]
Furthermore, the mapping $(\psi,I)\mapsto Q[\psi,I]$ is computable.
\end{theorem}

\begin{proof}
We have that for every graph $G$
\begin{align}
\#\psi_I(G) & = \#\psi_{\uparrow I}(G_\uparrow)\label{eq:nm_1}\\
~&= \#\homs{Q[\psi_{\uparrow},I]}{G_\uparrow}\label{eq:nm_2}\\
~&= \sum_{\varphi \in \supp(Q[\psi_{\uparrow},I])} \lambda_{\varphi}\cdot  \#\varphi(G_\uparrow)\label{eq:nm_3}\\
~&= \sum_{\varphi \in \supp(Q[\psi_{\uparrow},I])} \lambda_{\varphi}\cdot  \#\varphi_\downarrow(G)\label{eq:nm_4}\,,
\end{align}
where~\eqref{eq:nm_1} holds by Fact~\ref{fac:help_monotonicity} and~\eqref{eq:nm_2} holds by the generalized version of Theorem~\ref{thm:chen_mengel_ep} that works for arbitrary signatures, including inequalities over the free variables (see Section~\ref{sec:structures}). Furthermore,~\eqref{eq:nm_3} holds by definition of a quantum query and \eqref{eq:nm_4} is again due to Fact~\ref{fac:help_monotonicity}. Now consider the conjunctive queries $\varphi_\downarrow$: Those formulas might contain non-monotone constraints and we wish to get rid of them by invoking Lemma~\ref{thm:non_monotone_constraints_cqs}. However, this requires that the non-monotone constraints are only over free variables of $\varphi_\downarrow$. Equivalently, $\varphi$ must satisfy that each atom~$\overline{E}x\hat{x}$ in only over free variables. To this end, we observe that the quantum query~$Q[\psi_{\uparrow},I]$ in~\eqref{eq:nm_2} is obtained in two steps:

In the first step $\#\psi_{\uparrow I}(G_\uparrow)$ is transformed into a linear combination of quotient formulas $\#\psi_{\uparrow}/J$ for $J\subseteq I$. As, by assumption, all non-monotone constraints of $\psi$ are over free variables, it hence holds that all atoms $\overline{E}x\hat{x}$ in  $\#\psi_{\uparrow}/J$ are over free variables as well --- recall that the quotient only contracts free variables. 

Note that all quotient formulas $\#\psi_{\uparrow}/J$ are universal or existential positive. Now, in the second step, depending on whether $\psi$ is existential or universal positive, the formulas $\#\psi_{\uparrow}/J$ are transformed to a linear combination of conjunctive queries by either Theorem~\ref{thm:chen_mengel_ep} or Corollary~\ref{cor:up_to_ep}. The latter does not change the free variables and also relies on Theorem~\ref{thm:chen_mengel_ep}. Consequently, we have to guarantee that the construction of the quantum queries~$Q[\#\psi_{\uparrow}/J]$ yields as constituents only conjunctive queries satisfying that every atom $\overline{E}x\hat{x}$ is only over free variables. Now taking a look into the proof of Chen and Mengel~\cite[Section~4 and~5.3]{ChenM16} reveals that they perform inclusion and exclusion over conjunctions of subformulas of $\psi_{\uparrow}/J$. Consequently, no atom $\overline{E}x\hat{x}$ such that either $x$ or $\hat{x}$ is quantified, can be constructed.
This allows us to continue from~\eqref{eq:nm_4} by invoking Lemma~\ref{thm:non_monotone_constraints_cqs}:
\begin{align*}
\#\psi_I(G) &= \sum_{\varphi \in \supp(Q[\psi_{\uparrow},I])} \lambda_{\varphi}\cdot  \#\varphi_\downarrow(G)\\
~&=\sum_{\varphi \in \supp(Q[\psi_{\uparrow},I])} \lambda_{\varphi}\cdot \sum_{J \subseteq S} (-1)^{\#J} \cdot \#[\delta \wedge J]^{\mathcal{V}(S,\delta,J)}(G)\,,
\end{align*}
where $\varphi\downarrow = [\delta \wedge \neg S]$ and $\mathcal{V}(S,\delta,J)$ are as in Lemma~\ref{thm:non_monotone_constraints_cqs}. Finally, we collect for equivalent conjunctive queries and obtain the quantum query $Q[\psi,I]$.
\end{proof}

We are now able to proof Theorem~\ref{thm:main2_int} in case of the signature of graphs.
\begin{proof}[Proof of Theorem~\ref{thm:main2_int}]
Given a family $\Phi$ of existential or universal positive formulas with non-monotone constraints and inequalities over the free variables, we let $\Delta$ be the set of the corresponding quantum queries as given by Theorem~\ref{thm:main_quantum_query}. Then the problems $\#\pmc(\Phi)$ and $\#\homsprob(\Delta)$ are equivalent. The claim follows hence by Theorem~\ref{thm:main_quantum}.
\end{proof}\pagebreak

\section{Generalization to hypergraphs}\label{sec:structures}
In this section we are going to generalize all results that have been proved for graphs to logical structures. It will be very convenient to speak of hypergraphs instead of structures to adopt the notions of subgraphs and edges. However, let us make clear that what is called hypergraph in the following subsections is usually referred to as logical structure in the literature (see e.g. Chapter~4 in~\cite{flumgrohe}).\footnote{Readers that are used to hypergraphs with only one edge-relation should view our hypergraphs as edge-colored ones.}

\subsection{Further preliminaries}\label{sec:more_prelims}
\paragraph*{Hypergraphs} A \emph{signature} $\tau$ is a finite set of \emph{relation symbols} $E_1,\dots,E_\ell$ with arities $\arity_1,\dots,\arity_\ell$. We set $\arity(\tau) = \max\{ \arity_i~|~i\in [\ell]\}$ to be the arity of $\tau$. A \emph{hypergraph} $\mathcal{H}$ with signature $\tau$ consists of a finite set of vertices $V(\mathcal{H})$ and sets of (hyper-)edges $E_i(\mathcal{H}) \subseteq V(\mathcal{H})^{\arity_i}$ for every $i \in [\ell]$. The \emph{complementary hypergraph} $\overline{\mcH}$ of $\mcH$ has vertices $V(\mcH)$ and for every $i \in [\ell]$ and every $\vec{a} \in V(\mcH)^{\arity_i}$ it holds that $\vec{a} \in E_i(\overline{\mcH})$ if and only $\vec{a} \notin E_i(\mcH)$. Given hypergraphs $\mcH$ and $\mcF$ over the same signature $\tau$, we say that $\mcF$ is a \emph{subgraph} of $\mcH$ if $V(\mcF) \subseteq V(\mcH)$ and $E_i(\mcF) \subseteq E_i(\mcH)$ for every $i \in [\ell]$.

Given two hypergraphs $\mcH$ and $\mcF$ with signature $\tau$, a \emph{homomorphism} from $\mcH$ to $\mcF$ is a function $h: V(\mcH) \rightarrow V(\mcF)$ such that the following holds
\begin{equation*}
 \forall i \in [\ell]: \forall \vec{a} \in E_i(\mcH) : h(\vec{a}) \in E_i(\mcF) \,,
\end{equation*}
where $h(\vec{a})=(h(\vec{a}_1),\dots,h(\vec{a}_{\arity_i}))$. We denote $\homs{\mcH}{\mcF}$ as the set of all homomorphisms from $\mcH$ to $\mcF$. Now the notions of isomorphisms, endomorphisms and automorphisms, as well as color-prescribed and colorful homomorphisms are defined similarly to the case of graphs (see Section~\ref{sec:prelims} and Section~\ref{sec:minors_and_colors}). In particular, given two hypergraphs $\mathcal{H}, \mathcal{G}$ and a set of vertices $X \subseteq V(\mathcal{H})$, the set $\homs{\mathcal{H},X}{\mathcal{G}}$ is defined to be the set of all assignments $a:X \rightarrow V(\mathcal{G})$ that can be extended to a homomorphism $h \in \homs{\mathcal{H}}{\mathcal{G}}$. $\#\cphomsprob$ and $\#\cfhomsprob$ are defined likewise for color-prescribed and colorful homomorphisms.

\paragraph*{Example} Let $\tau = (E)$ such that $\arity(E) = 2$. Then the set of hypergraphs with signature $\tau$ is precisely the set of directed graphs. If we consider the subset of hypergraphs $\mcH$ such that additionally $E(\mcH)$ is symmetric and irreflexive, then this set is precisely the set of undirected graphs without self-loops. In this case the notion of homomorphisms and graph homomorphisms coincide.

\begin{definition}[Gaifman graph]\label{def:gaifman_graph}
	Given a hypergraph $\mcH$ of signature $\tau=(E_1,\dots,E_\ell)$, the \emph{Gaifman graph} $\mathbb{G}(\mcH)$ of $\mcH$ has vertices $V(\mcH)$ and contains an edge $\{u,v\}$ if and only if $u\neq v$ and $u$ and $v$ contained in a common edge of $\mcA$, i.e., there exists $i\in[\ell]$ and $\vec{a} \in E_i(\mcH)$ such that $u = \vec{a}_j$ and $v = \vec{a}_k$ for some $j,k \in [\arity_i]$.
\end{definition}

\paragraph*{First-order formulas and model checking for hypergraphs}
Let $\tau =(E_1,\dots,E_\ell)$ be a fixed signature and let $\mathcal{V}$ be a countably infinite set of variables. For every $i\in[\ell]$ and $\vec{z} \in \mathcal{V}^{\arity_i}$, ``$\vec{z} \in E_i$'' is called an \emph{atom}. Now \emph{first-order formulas} (over~$\tau$) are defined inductively over atoms, boolean connectives ($\wedge,\vee,\neg$) and quantifiers ($\exists,\forall$). We consider the following subsets of first-order formulas: As in case of graphs, a \emph{conjunctive query} is a first-order formula of the form
\begin{equation*}
 \varphi = x_1 \dots x_k \exists y_1 \dots \exists y_\ell : a_1 \wedge \dots \wedge a_m \,,
\end{equation*}
where each $a_i$ is an atom.  An \emph{existential positive formula} (or \emph{query}) is of the form
\begin{equation*}
\varphi = x_1 \dots x_k \exists y_1 \dots \exists y_\ell : \psi \,,
\end{equation*}
where $\psi$ is either a disjunctive or a conjunctive normal form of atoms without negations\footnote{As we will parameterize by the length of the formula, computing the DNF out of a CNF and vice versa only takes time depending on the parameter and hence does not influence the (parameterized) complexity results.}. Note that every conjunctive query is also an existential positive query. A \emph{universal positive formula} is of the form
\begin{equation*}
\varphi = x_1 \dots x_k \forall y_1\dots\forall y_\ell : \psi \,,
\end{equation*}
where $\psi$ is either a disjunctive or a conjunctive normal form of atoms without negations.

Given a hypergraph $\mcG$ and an existential positive query $\varphi$ with free variables $X$ and quantified variables $Y$ as above. The set $\varphi(\mcG)$ of solutions (or answers) to $\varphi$ in $\mcG$ is the set of all assignments $a: X \rightarrow V(\mcG)$ such that there is an assignment $h: X \cup Y \rightarrow V(\mcG)$ such that $h|_X = a$ and $h \models \psi$, where $\models$ is defined as follows: Given an atom  ``$\vec{z} \in E_i$'', we define $h \models \vec{z} \in E_i$ to be true if and only if $h(\vec{z}) \in E_i(\mcG)$. Given a DNF or CNF $\psi$ of atoms, $h \models \psi$ is defined inductively over $\vee$ and $\wedge$. Given a universal positive query $\varphi$, the set $\varphi(\mcG)$ is the set of all assignments $a: X \rightarrow V(\mcG)$ such that for all assignments $h: X \cup Y \rightarrow V(\mcG)$ that satisfy $h|_X = a$ it holds that $h \models \psi$.

Every conjunctive query $\varphi$ with free variables $X$ and quantified variables $Y$ is associated with a pair $(\mcH,X)$ where $\mcH$ is a hypergraph over the same signature as $\varphi$. Here, $V(\mcH) =X \cup Y$ and for every $i \in [\ell]$ we add a vector $\vec{z} \in (X \cup Y)^{\arity_i}$ to $E_i(\mcH)$ if and only if ``$\vec{z} \in E_i$'' is an atom of $\varphi$. Observe that for all hypergraphs $\mcG$ it holds that $\varphi(\mcG) = \homs{\mcH,X}{\mcB}$. We adapt the notion for graphs and call $(\mcH,X)$ a graphical conjunctive query. The definition of a quantum query transfers in the canonical way.

Now let $\Delta$ be a set of conjunctive queries, existential positive queries, or universal positive queries. We say that $\Phi$ has \emph{bounded arity} if there is a constant $C \in \mathbb{N}$ such that every signature of some formula in $\Delta$ has arity at most $C$. We assume all classes of queries in this work to have bounded arity.

\subsection{Reduction from the Gaifman graph}

In what follows we prove that counting color-preserving partial homomorphisms from a hypergraph is at least as hard as counting color-preserving partial homomorphisms from its Gaifman graph.

\begin{lemma}\label{lem:red_gaifman_graph}
	Let $(\mcH,X)$ be a graphical conjunctive query of signature $\tau =(E_1,\dots,E_\ell)$. Then there exists a deterministic algorithm $\mathbb{A}$ with oracle access to $\#\cphoms{\mcH,X}{\star}$ that computes $\#\cphoms{\mathbb{G}(\mcH),X}{\star}$. Furthermore $\mathbb{A}$ runs in time $O(f(\mcH)\cdot n^{\arity(\tau)})$ for some computable function~$f$.
\end{lemma}
\begin{proof}
	We will first provide the intuition behind the proof by considering the following restriction on $\mcH$. We assume that $\ell=1$ and let $\arity = \arity_i$. Furthermore, we assume that $\mcH$ does not have any edge that contains a multiple occurrence of the same vertex. Now given a $\mathbb{G}(\mcH)$-colored graph $G$ for which we want to compute $\#\cphoms{\mathbb{G}(\mcH),X}{G}$, we can construct a $\mcH$-colored hypergraph $\mathcal{G}'$ from $G$ as follows. For every $\vec{u}=(u_1,\dots,u_\arity)\in E_1(\mcH)$ we search all cliques in $G$ of size $\arity$ that are colored with $u_1,\dots,u_\arity$. Every clique $\vec{c}=(c_1,\dots,c_\arity)$ such that~$c_i$ has color $u_i$ is then added to $G$ as a hyperedge in $E_1$. If this is done for all $\vec{u} \in E_1(\mcH)$ we delete all former edges of $G$. It is easy to see that the resulting hypergraph $\mathcal{G}'$ is $\mcH$-colored, except for the case that there was a $\vec{u} \in E_1(\mcH)$ for which there was no corresponding clique in $G$. In this case, however, there is no color-preserving homomorphism from $\mathbb{G}(\mcH)$ to $G$ at all and we can just output $0$. Otherwise we claim that
	\begin{equation}\label{eq:gaifman_to_hyper}
		\cphoms{\mathbb{G}(\mcH),X}{G} = \cphoms{\mcH,X}{\mathcal{G}'} \,.
	\end{equation}
	For the first direction, let $a \in \cphoms{\mathbb{G}(\mcH),X}{G}$. Then there exists a homomorphism $h \in \cphoms{\mathbb{G}(\mcH)}{G}$ such that $h|_X = a$. We claim that $h$ is contained in $\cphoms{\mcH}{\mathcal{G}'}$ as well. To this end, let $\vec{u}=(u_1,\dots,u_\arity) \in E_1(\mcH)$. By the definition of the Gaifman graph it holds that $\vec{u}$ is a clique in $\mathbb{G}(\mcH)$ (recall that we assumed the absence of multiple occurrences). As $h \in \cphoms{\mathbb{G}(\mcH)}{G}$ it
	hence holds that $h(\vec{u})$ is a clique of size $\arity$ in $G$ with colors $u_1,\dots,u_a$. By the construction of $\mathcal{G}'$ we have that $h(\vec{u}) \in E_1(\mcG')$. Consequently $h \in \cphoms{\mcH}{\mathcal{G}'}$ and $a \in \cphoms{\mcH,X}{\mathcal{G}'}$.

	For the backward direction, let $a \in \cphoms{\mcH,X}{\mathcal{G}'}$. Then there exists a homomorphism $h \in \cphoms{\mcH}{\mathcal{G}'}$ such that $h|_X = a$. We claim that $h$ is contained in $\cphoms{\mathbb{G}(\mcH)}{G}$ as well. To see this, let $e=\{v,w\} \in E(\mathbb{G}(\mcH))$. By the definition of the Gaifman graph, there exists an edge $\vec{u}=\{u_1,\dots,u_\arity\}$ in $E_1(\mcH)$ such that $v=u_i$ and $w=u_j$ for some $1\leq i < j \leq \arity$. As $h \in \cphoms{\mcH}{\mathcal{G}'}$ it holds that $h(\vec{u})\in E_1(\mathcal{G}')$. By the construction of $\mathcal{G}'$ we have that $h(\vec{u})$ is a clique in $G$ colored with $u_1,\dots,u_\arity$. In particular it holds that $\{h(u_i),h(u_j)\} \in E(G)$ and that $h(u_i)$ has color $u_i$ and $h(u_j)$ has color $u_j$. As $v = u_i$ and $w=u_j$ we conclude that $h \in \cphoms{\mathbb{G}(\mcH)}{G}$ and hence $a \in \cphoms{\mathbb{G}(\mcH),X}{G}$.

	This completes the reduction for the restricted case. We remark that the claimed running time bound follows from the fact that $\mcG'$ can be constructed in time $f(\mcH)\cdot n^{\arity(\tau)}$ as we only need to search for cliques of size $\leq \arity(\tau)$.

	Let us now explain how to get rid of the restrictions. First, consider an edge of $\mcH$ that contains multiple occurrences of a vertex, say $\vec{u}=(u_1,u_2,u_1,u_1,u_3,u_2) \in E_1(\mcH)$. We observe that there are exactly $3$ different vertices: $u_1,u_2$ and $u_3$. Hence, when constructing $\mathcal{G}'$, we search for cliques of size $3$ that are colored with $u_1,u_2$ and $u_3$. Then, for every clique $(a,b,c)$ in $G$ colored with $(u_1,u_2,u_3)$ we add $(a,b,a,a,c,b)$ to $E_1(\mcG')$.

	Finally, we can assume that $\tau$ contains more than one relation symbol by employing the construction for every relation. It is easy to see that Equation~\ref{eq:gaifman_to_hyper} remains true for the unrestricted case if the construction is modified as explained above.
\end{proof}

\subsection{Equivalence of Conjunctive Queries}\label{sec:counting_minimality}
In this subsection we prove that every endomorphism of a minimal graphical conjunctive that bijectively maps the free variables to itself is already an automorphism. Recall that two conjunctive queries $(\mathcal{H},X)$ and $(\hat{\mathcal{H}},\hat{X})$ are called equivalent, we write $(\mathcal{H},X) \sim (\hat{\mathcal{H}},\hat{X})$, if $\#\homs{\mathcal{H},X}{\star}$ and $\#\homs{\hat{\mathcal{H}},\hat{X}}{\star}$ are the same functions and a query $(\mathcal{H},X)$ is called minimal if it is a vertex-minimal element in its equivalence class.

Chen and Mengel provided an explicit criterion for equivalence.\footnote{Note that in case of graphs, Lemma~\ref{lem:count_equiv_chen_mengel} is identical with Corollary~\ref{cor:counting_equivalence_explicit_graphs}. While Chen and Mengel gave a more involved proof of the lemma, we point out that the generalization to hypergraphs can be proven just as easy as we demonstrated it for graphs in Section~\ref{sec:lin_combs}.}
\begin{lemma}[\cite{ChenM16}]\label{lem:count_equiv_chen_mengel}
	Two conjunctive queries $(\mcH,X)$ and $(\hat{\mcH}, \hat{X})$ are equivalent if and only if there exist surjective functions $s:X \rightarrow \hat{X}$ and $\hat{s}: \hat{X} \rightarrow X$ that can be extended to homomorphisms $h \in \homs{\mcH}{\hat{\mcH}}$ and $\hat{h} \in \homs{\hat{\mcH}}{\mcH}$, respectively.
\end{lemma}

We do not want to distinguish between two conjunctive queries that are equal up to consistently renaming both, the quantified as well as the free variables. Hence we say that two conjunctive queries $(\mcH,X)$ and $(\hat{\mcH},\hat{X})$ are \emph{isomorphic} if and only if there is an isomorphism from $\mcH$ to $\hat{\mcH}$ that bijectively maps $X$ to $\hat{X}$. We write $(\mcH,X) \cong (\hat{\mcH},\hat{X})$.

We now prove Lemma~\ref{lem:counting_minimality_graphs} in the more general context of hypergraphs. We start by introducing the necessary preliminaries with respect to cores of hypergraphs and homomorphic equivalence:
Two hypergraphs $\mcH$ and $\hat{\mcH}$ are \emph{homomorphically equivalent} if there exist homomorphisms from $\mcH$ to $\hat{\mcH}$ and from $\hat{\mcH}$ to $\mcH$. $\mcH$ is called a \emph{core} if it is not homomorphically equivalent to a proper subgraph of $\mcH$. As for every proper subgraph of $\mcH$, the identity function is a homomorphism, we obtain
\begin{observation}\label{obs:core_endo_auto}
	A hypergraph $\mcH$ is a core if and only if there exists no homomorphism from $\mcH$ to a proper subgraph. Hence, every endomorphism of a core is an automorphism.
\end{observation}
We say that a subgraph $\hat{\mcH}$ of $\mcH$ is a core of $\mcH$ if $\hat{\mcH}$ and $\mcH$ are homomorphically equivalent and $\hat{\mcH}$ is a core.
\begin{lemma}[See e.g. Lemma~13.9 in \cite{flumgrohe}] \label{lem:hom_cores_iso}
	Let $\mcH$ and $\mcG$ be homomorphically equivalent and let $\hat{\mcH}$ and $\hat{\mcG}$ be cores of $\mcH$ and $\mcG$, respectively. Then $\hat{\mcH}$ and $\hat{\mcG}$ are isomorphic. In particular, all cores of a hypergraph are isomorphic.
\end{lemma}
\begin{corollary}\label{cor:min_hom_is_unique}
	All minimal elements in a single equivalence class with respect to homomorphic equivalence are isomorphic.
\end{corollary}
Now Lemma~\ref{lem:hom_cores_iso} and Corollary~\ref{cor:min_hom_is_unique} allow us to speak of \emph{the} core of a hypergraph --- we write $\mathsf{core}(\mcH)$--- and \emph{the} minimal representative of a homomorphic equivalence class. Our goal is to achieve a similar result for equivalence of conjunctive queries. To this end, we use the \emph{augmented core} of a graphical conjunctive query which refines the notion of a core.
\begin{definition}
	Given a conjunctive query $(\mcH,X)$, we obtain an augmented hypergraph $\aug(\mcH,X)$ from $\mcH$ by adding a new relation $\raug(\mcH) = \{(x,x') \in X^2 ~|~x \neq x'\}$. Note that this also adds a new relation symbol $\raug$ to the signature.
	We let $\mcH_c$ be the hypergraph obtained from $\mathsf{core}(\aug(\mcA,X))$ by removing $\raug(\mcH)$ and we define the \emph{augmented core} of $(\mcH,X)$ to be $\acore(\mcH,X):=(\mcH_c,X)$.
\end{definition}
We note that $\acore$ is well-defined, i.e., that every element $x\in X$ is contained in the core of $\mathsf{aug}(\mcH,X)$. To see this, we observe that $\raug(\mcH)$ induces a clique without loops on $X$ and hence, any subgraph $\hat{\mcH}$ of $\mcH$ such that there is a homomorphism from $\mcH$ to $\hat{\mcH}$ must contain every element $x\in X$.
\begin{lemma}\label{lem:main_augmented_structs}
	Let $(\mcH,X)$ and $(\hat{\mcH},\hat{H})$ be two graphical conjunctive queries. Then it holds that $(\mcH,X)$ and $(\hat{\mcH},\hat{H})$ are equivalent if and only if $\aug(\mcH,X)$ and $\aug(\hat{\mcH},\hat{H})$ are homomorphically equivalent. Moreover, $\aug(\mcH,X)$ is minimal with respect to homomorphic equivalence if and only if $(\mcH,X)$ is minimal.
\end{lemma}
\begin{proof}
	Every homomorphism from $\aug(\mcH,X)$ to $\aug(\hat{\mcH},\hat{H})$ must surjectively map $X$ to $\hat{X}$ as $\raug$ induces a clique (without self-loops) on $X$ and $\hat{X}$. On the other hand, every homomorphism from $\mcH$ to $\hat{\mcH}$ that surjectively maps $X$ to $\hat{X}$ is a homomorphism from $\aug(\mcH,X)$ to $\aug(\hat{\mcH},\hat{H})$.

	If $\aug(\mcH,X)$ is not minimal with respect to homomorphic equivalence then it is not a core. Hence $\aug(\mcH,X)$ is homomorphically equivalent to a proper subgraph $\mcF$ of $\aug(\mcH,X)$. Hence $\mcF$ must contain all vertices in $X$ and $\raug(\mcF) = \raug(\mcH)$, because otherwise there would be no homomorphism from $\aug(\mcH,X)$ to $\mcF$. It follows that $\mcF=\aug(\hat{\mcH},X)$ for some hypergraph $\hat{\mcH}$. By the first part, it follows that $(\mcH,X)$ and $(\hat{\mcH},X)$ are equivalent. As $\hat{\mcH}$ is a proper subgraph of $\mcH$ is follows that $(\mcH,X)$ is not minimal.

	On the other hand, if $(\mcH,X)$ is not minimal, then there exists a conjunctive query $(\hat{\mcH},\hat{H})$ such that $\hat{\mcH}$ is a proper subgraph of $\mcH$ and there are homomorphisms $h$ from $\mcH$ to $\hat{\mcH}$ and $\hat{h}$ from $\hat{\mcH}$ to $\mcH$ that surjectively map $X$ to $\hat{X}$ and $\hat{X}$ to $X$ respectively. Hence $h$ and $\hat{h}$ are also homomorphisms from  $\aug(\mcH,X)$ to $\aug(\hat{\mcH},\hat{H})$ and from $\aug(\hat{\mcH},\hat{H})$ to $\aug(\mcH,X)$, respectively. Therefore $\aug(\mcH,X)$ is not minimal with respect to homomorphic equivalence.
\end{proof}
This, together with Lemma~\ref{lem:hom_cores_iso} and Corollary~\ref{cor:min_hom_is_unique}, immediately implies the following:
\begin{corollary}\label{cor:main_counting_equiv} For all graphical conjunctive queries $(\mcH,X)$ and $(\hat{\mcH},\hat{H})$ is holds that
	\begin{enumerate}
		\item $(\mcH,X) \sim \acore(\mcH,X)$.
		\item If $(\mcH,X)$ is minimal then $(\mcH,X) \cong \acore(\mcH,X)$.
		\item If $(\mcH,X)$ and $(\hat{\mcH},\hat{H})$ are minimal and $(\mcH,X) \sim (\hat{\mcH},\hat{H})$ then $(\mcH,X) \cong (\hat{\mcH},\hat{H})$.
		\item If $(\mcH,X)$ is minimal and $h$ is an endomorphism of $\mcH$ that surjectively maps $X$ to $X$, then $h$ is an automorphism.
	\end{enumerate}
\end{corollary}
Note that Lemma~\ref{lem:counting_minimality_graphs} follows from 4. restricted to the signature of graphs.

\subsection{The generalized classification theorem}
Before generalizing our main theorem, we need to point out that all structural parameters we have seen for conjunctive queries over the signature of graphs are generalized to hypergraphs via the Gaifman graph. We furthermore observe that the reduction from color-prescribed to uncolored homomorphisms in Section~\ref{sec:phom_to_hom} as well as the transformation of existential and universal positive queries, possibly including inequalities over the free variables, can be done completely analogously as in the case of graphs. However, we did not give a formal proof of $\#\W$-easiness in case of bounded dominating star size yet, which is hence provided first.

\subsubsection{\#W[1]-easiness for hypergraphs}\label{sec:algo_hyper}
\begin{theorem}[Theorem~\ref{thm:algo_graphs} for hypergraphs]\label{thm:algo_structs}
	Let $\Delta$ be a class of conjunctive queries with bounded dominating star size. Then $\#\homsprob(\Delta)$ is $\#\W$-easy.
\end{theorem}
\begin{proof}
	We use the identity $\#\W = \#\cc{A}[1]$ (see Theorem~14.17 in~\cite{flumgrohe}). By definition, $\#\mathsf{A}[1]$ is the class of all parameterized counting problems that are parameterized Turing reducible to $\#\homsprob(\Pi_0)$, where $\Pi_0$ is the set of all conjunctive queries without quantified variables \cite{flumgrohe}.

	Given an instance of $\#\homsprob(\Delta)$, i.e., a conjunctive query $(\mcH,X) \in \Delta$ and a hypergraph $\mcG$, we will construct a query $(\hat{\mcH},X) \in \Pi_0$ and a hypergraph $\hat{\mcG}$ such that $\#\homs{\mcH,X}{\mcG} =\#\homs{\hat{\mcH},X}{\hat{\mcG}}$ in $\mathsf{FPT}$ time, using an oracle for $\#\W$ with the additional restriction that the parameter of every query to the oracle only depends on $(\mcH,X)$.

	Let $\mathbb{G}=\mathbb{G}(\mcH)$ be the Gaifman graph of $\mcH$, let $Y$ be the set of quantified variables in $(\mcH,X)$ and let $Y_1,\dots,Y_\ell$ be the connected components of $\mathbb{G}[Y]$. Furthermore, for each $i \in [\ell]$ we set $c_i$ to be the number of vertices in $X$ that are adjacent to a vertex in $Y_i$ in $\mathbb{G}$. Note that every $c_i$ is bounded by some overall constant as $\Delta$ has bounded dominating star size. Next we add new relation symbols
	$\hat{E}_1,\dots,\hat{E}_\ell$ with arities $c_1,\dots,c_\ell$ to the signature. We proceed for each $i \in [\ell]$ as follows: Let $x^i_1,\dots,x^i_{c_i} \in X$ be the elements in $\mcH$ that are adjacent to an element in $Y_i$ in $\mathbb{G}$. For every tuple $\vec{v}=(v_1,\dots,v_{c_i}) \in V(\mcG)^{c_i}$, we check whether there is a homomorphism $h$ from $\mcH[Y_i \cup \{ x^i_1,\dots,x^i_{c_i}\}]$ to $\mcG$ such that $h(x^i_j)=v_j$ for all $j \in [c_i]$. Intuitively, this check is positive if $\vec{v}$ is a possible candidate for the image of the free variables $x^i_1,\dots,x^i_{c_i} \in X$ with respect to ``neighborhood'' $Y_i$. Now we add $\vec{v}$ to $\hat{E}_i(\mcG)$ if and only if the check was positive and remark that we only need to perform $\ell \cdot |\mcG|^{O(1)}$ checks as all $c_i$ are upper bounded by a constant. The resulting hypergraph is $\hat{\mcG}$. We observe that every check can be done by querying an oracle to $\W$ as it can be formulated as an instance of the problem of deciding the existence of a solution to a conjunctive query, which is known to be contained in $\W$ (see e.g. Theorem~7.22 in \cite{flumgrohe}). As we have access to an oracle for $\#\W$, we can certainly simulate an oracle for $\W$ --- if we know the number of solutions we also know whether one exists.

	It remains to show how to construct $(\hat{\mcH},X)$: We fix an ordering of the free variables $X$ of $\varphi$. Next, starting from $\mcH$, for every $i \in [\ell]$, let $x^i_1,\dots x^i_{c_i}$ be the ordered neighbors of vertices in $Y_i$ in $\mathbb{G}$. We then add an atom $(x^i_1,\dots x^i_{c_i}) \in \hat{E}_i$ to $\mcH$.
	In the end, we delete all quantified variables from $\mcH$ along with all edges that contain a quantified variable and denote the resulting hypergraph as $\hat{\mcH}$. Now it can be easily verified that $\#\homs{\mcH,X}{\mcG}$ equals $\#\homs{\hat{\mcH},X}{\hat{\mcG}}$.

	We summarize the reduction:
	\begin{enumerate}
		\item Given $(\mcH,X)$ and $\mcG$, construct $\hat{\mcG}$ by performing at most $\ell \cdot |\mcG|^{\mathsf{dss}(\mcH,X)}$ checks, using the oracle.
		\item Construct $\mcH$.
		\item Query the oracle to obtain $\#\homs{\hat{\mcH},X}{\hat{\mcG}}$ and output the result.
	\end{enumerate}
	This can be done in $\mathsf{FPT}$ time as the dominating star size of $\Delta$ is bounded. For the same reason, the arity of $\mcH$ is bounded by a constant. Furthermore, there is a computable function $t$ such that the size of every oracle query is bounded by $t(|(\mcH,X)|) \cdot \mathsf{poly}(|\mcG|)$ and the parameter of every oracle query is bounded by $t(|(\mcH,X)|)$.
\end{proof}

\subsubsection{The main results}\label{sec:main_result_structs}
\begin{theorem}[Pentachotomy for counting answers to conjunctive queries]\label{thm:main_hyper}
	Let $\Delta$ be a recursively enumerable class of minimal conjunctive queries over arbitrary signatures with bounded arity.
	\begin{enumerate}
		\item If the treewidth of $\Delta$ and $\mathsf{contract}(\Delta)$ is bounded, then $\#\homsprob(\Delta)$ is polynomial-time computable.
		\item If the treewidth of $\Delta$ is unbounded and the treewidth of $\mathsf{contract}(\Delta)$ is bounded, then $\#\homsprob(\Delta)$ is $\W$-equivalent.
		\item If the treewidth of $\mathsf{contract}(\Delta)$ is unbounded and the dominating star size of $\Delta$ is bounded, then $\#\homsprob(\Delta)$ is $\#\W$-equivalent.
		\item If the dominating star size of $\Delta$ is unbounded, then $\#\homsprob(\Delta)$ is $\#\Wtwo$-hard. In particular, given a formula $\delta$ with $\mathsf{dss}(\delta) \geq 3$, computing $\#\homs{\delta}{\star}$ cannot be done in time $O(n^{\mathsf{dss}(\delta) - \varepsilon})$ for any $\varepsilon > 0$ unless SETH fails.
		\item If additionally the linked matching number of $\Delta$ is unbounded, then $\#\homsprob(\Delta)$ is $\#\Atwo$-equivalent.
	\end{enumerate}
	Furthermore, the classification remains true if $\Delta$ is a family of quantum queries, with the exception that the queries in the first case might only be fixed parameter tractable.
\end{theorem}
\begin{proof}
	We have that 1., 2., and $\#\W$-hardness in 3. follow from~\cite{ChenM15}. $\#\W$-easiness in 3. follows from Theorem~\ref{thm:algo_structs}. 4. and hardness in 5. hold as we can apply the corresponding results for the Gaifman graphs (Theorem~\ref{thm:wtwo_hardness_graphs} and Theorem~\ref{thm:atwo_hardness_graphs}) and then reduce to the primal hypergraphs via Lemma~\ref{lem:red_gaifman_graph}. After that we can reduce color-prescribed homomorphisms to uncolored homomorphisms
	completely analogously as we did for graphs (Section~\ref{sec:phom_to_hom}). $\#\Atwo$-easiness follows from the fact that the general problem of counting answers to conjunctive queries is easy for $\#\Atwo$ (see Section~\ref{sec:a2}). Finally, the extension to quantum queries holds by the complexity monotonicity property, which follows by a straight-forward adaption of the proof of Lemma~\ref{lem:new_pods_monotonicity} for hypergraphs. Alternatively, the original but more involved version of Chen and Mengel~\cite{ChenM16} can be applied.
\end{proof}

Now transforming universal and existential positive formulas, possibly with inequalities and non-monotone constraints over the free variables, to a quantum query can also be done completely analogously as for graphs, proving Theorem~\ref{thm:data complexity} and Theorem~\ref{thm:main2_int} in the general case. We point out that, in case of conjunctive queries with inequalities over the free variables, the criterion for a hypergraph to be contained in the support of the linear combination can be stated in terms of vertex contractions along matroid flats similar to the case of graphs (see Theorem~\ref{thm:quantum_graph_for_ineqs}).

\begin{corollary}
	The classification of Theorem~\ref{thm:main_hyper} applies to all recursively enumerable families $\Delta$ of universal and existential positive formulas that might contain inequalities and non-monotone constraints over the free variables. In particular, if $\Delta$ is a set of conjunctive queries with inequalities over the free variables, the criteria for the classification can be stated explicitly as in Theorem~\ref{thm:quantum_graph_for_ineqs}.
\end{corollary}

\section{\#A[2]-Normalization and hardness of grates} \label{sec:a2}
In this section, we prove that the problem of counting answers to graphical conjunctive queries is $\#\Atwo$-equivalent.
This result can be seen as a counting analogue of the Normalization Theorem of the $\cc{A}$-Hierarchy due to Flum and Grohe~\cite[Chapter~8]{flumgrohe}. While most of their reductions directly translate to the counting version, one major step can be drastically simplified using the framework of quantum queries, namely the reduction from existential positive formulas to conjunctive queries.

After establishing the normalization result, we are going to show the $\#\Atwo$-hardness of $\#\cphomsprob(\mathsf{Grates})$ for some specific class $\mathsf{Grates}$ of all \emph{grates}. This latter result is the starting point for our reduction in Section~\ref{sec:hard_a_2}.

\subsection{Hardness of counting answers to graphical conjunctive queries}
The class $\#\Atwo$ is formally defined as every problem that has a parameterized parsimonious reduction to the parameterized model counting problem for universal first-order formulas.
From this definition, which we make formal shortly, it is clear that counting answers to graphical conjunctive queries is $\#\Atwo$-easy.
The main difficulty when proving hardness is that the defining problem of~$\#\Atwo$ talks about first-order formulas that may contain relations of unbounded arity, and that may contain (in)equalities and negations (see Chapter~14 in~\cite{flumgrohe}).
More precisely, let~$\Sigma_1$ be the set of all first-order formulas (possibly with (in)equalities and negations) that are of the form
\begin{equation*}
\exists y_1\dots\exists y_\ell: \varphi \,.
\end{equation*}
Here $\varphi$ is quantifier-free and contains a set $X$ of free variables. Similarly, $\Pi_1$ is defined using $\forall$-quantifiers. $\Sigma_1^+$ and $\Pi_1^+$ are the subsets of $\Sigma_1$ and $\Pi_1$ that do not contain negations (when transferred to prenex normal form). Furthermore, given a class $\Phi$ of first-order formulas, we write $\Phi[r]$ for the subset of $\Phi$ that contains only formulas over signatures of arity at most $r$.

Recall that until now, every set $\Phi$ of first-order formulas was assumed to have a constant bound on the arity of every signature of every formula in $\Phi$ and furthermore, that no formula in $\Phi$ contains equalities. In this subsection (and \emph{only} in this subsection), we omit these restrictions. We will emphasize this by using the original notation of Flum and Grohe:\footnote{We refer the reader to Chapt.~4 in~\cite{flumgrohe} for a complete and formal introduction of first-order model checking. We follow their definitions closely, with the exception that we refer to logical structures as (edge-colored) hypergraphs.}

\begin{definition}
	Given a class $\Phi$ of first-order formulas that may contain equalities and that may have unbounded arity, the problem $\#\pmc(\Phi)$ asks, given a formula $\varphi \in \Phi$ with free variables $X$ and a hypergraph $\mathcal{G}$ over the same signature as $\varphi$, to compute the cardinality of the set $\varphi(\mathcal{G})$ of assignments $a:X \rightarrow V(\mathcal{G})$ that make the formula true. The problem is parameterized by $|\varphi|$. In particular, we write $\#\pmc(\Phi[\mathsf{GRAPHS}])$ for the problem of, given a (simple and loopless) graph $G$ and a formula $\varphi \in \Phi$ over the signature of loopless graphs, i.e. $Ezz$ is not allowed as an atom for any variable $z$, computing $\#\varphi(G)$.
\end{definition}
\begin{remark}\label{rem:graph_sig}
	If $\Phi$ is the set of all conjunctive queries, then the problem $\#\pmc(\Phi[\mathsf{GRAPHS}])$ is precisely the problem $\#\homsprob(\mathsf{G})$ where $\mathsf{G}$ is the set of all conjunctive queries over the signature of graphs without self-loops.
\end{remark}
\begin{definition}[{\cite[Chapter~14]{flumgrohe}}]
	The parameterized complexity class $\#\Atwo$ is defined via
	\begin{equation*}
	\#\Atwo = [\#\pmc(\Pi_1)]^{\cc{FPT}} \,,
	\end{equation*}
	where $[\cdot]^{\cc{FPT}}$ is the closure with respect to parameterized parsimonious reductions.
\end{definition}

We are ready to prove a normalization lemma for $\#\Atwo$. To this end, let ${\mathsf{CQ}_=}$ be the set of all conjunctive queries that may contain equalities and $\mathsf{CQ}$ be the subset without equalities.

\begin{lemma}[Normalization for $\#\Atwo$]\label{lem:Atwo_norm}
	It holds that
\begin{equation*}
	\#\pmc(\Pi_1) \leq^{\cc{FPT}}_T \#\pmc(\mathsf{CQ}[\mathsf{GRAPHS}]) \,.
\end{equation*}
\end{lemma}
\begin{proof}
	We construct a sequence of reductions:
	\begin{claim}\label{clm:better_a_two}
		It holds that $\#\pmc(\Pi_1) \leq^{\cc{FPT}}_T \#\pmc(\Sigma_1)$.
	\end{claim}
	\begin{claimproof}
		Let
		\begin{equation*}
		\varphi = x_1 \dots x_k \forall y_1 \dots\forall y_\ell: \psi
		\end{equation*}
		be a formula in $\Pi_1$ such that $\psi$ is quantifier free --- note that $x_1,\dots,x_k$ denote the free variables of~$\psi$ --- and let $\mathcal{G}$ be a hypergraph over the same signature as $\varphi$. Furthermore, let $n = \#V(\mathcal{G})$.

		We define $\varphi' \in \Sigma_1$ as follows:
		\begin{equation*}
		\varphi' = x_1 \dots x_k \exists y_1\dots \exists y_\ell: \neg \psi \,.
		\end{equation*}
		Now it can easily be verified that $\#\varphi(\mathcal{G}) = n^k - \#\varphi'(\mathcal{G})$. This induces the reduction.
	\end{claimproof}

	\begin{claim}
		It holds that $\#\pmc(\Sigma_1) \leq^{\cc{FPT}}_T \#\pmc(\Sigma_1^+[\mathsf{GRAPHS}])$.
	\end{claim}
	\begin{claimproof}
		Flum and Grohe proved the following sequence of reductions for every odd $t \geq 1$ in the decision realm (see Chapter~8 in~\cite{flumgrohe} and \cite{flumgrohe_modelchecking}):
		\begin{equation}
		\pmc(\Sigma_t) \leq^{\cc{FPT}}_T \pmc(\Sigma_t^+)\leq^{\cc{FPT}}_T \pmc(\Sigma_t^+[2]) \leq^{\cc{FPT}}_T \pmc(\Sigma_t^+[\mathsf{GRAPHS}])\,.
		\end{equation}
		A close look reveals that all of the above constructions work as well in the counting world, in particular for $t = 1$.
	\end{claimproof}

	\begin{claim}\label{clm:easy_simplification}
		It holds that $\#\pmc(\Sigma_1^+[\mathsf{GRAPHS}]) \leq^{\cc{FPT}}_T \#\pmc(\mathsf{CQ}_=[\mathsf{GRAPHS}])$.
	\end{claim}
	\begin{claimproof}
		Every formula $\varphi \in \Sigma_1^+[\mathsf{GRAPHS}]$ is an existential-positive formula. Hence we can express $\varphi$ as a linear combination of conjunctive queries over the same signature as in~\cite{ChenM16}. We then use the oracle for $\#\pmc(\mathsf{CQ}_=[\mathsf{GRAPHS}])$ to compute every term in the linear combination.
	\end{claimproof}

	\begin{claim}
		It holds that $\#\pmc(\mathsf{CQ}_=[\mathsf{GRAPHS}]) \leq^{\cc{FPT}}_T \#\pmc(\mathsf{CQ}[\mathsf{GRAPHS}])$.
	\end{claim}
	\begin{claimproof}
		An equality $z_1=z_2$ in a conjunctive query can easily be simulated by substituting every occurrence of $z_2$ by $z_1$ and removing the equality afterwards. If the substitution leads to self-loops we can just output zero, as the input graphs are loopless.
	\end{claimproof}
	This concludes the proof of the normalization lemma.
\end{proof}
The proof of Claim~\ref{clm:easy_simplification} is considerably easier than its analogue in the decision world as we were able to make use of the framework of quantum queries. The proof of Claim~\ref{clm:better_a_two} actually shows that $\#\pmc(\Pi_1)$ and $\#\pmc(\Sigma_1)$ are interreducible from which we conclude that every problem we considered in this paper is $\#\Atwo$-easy.

\subsection{Hardness of counting vertex sets matching to a clique}\label{sec:gamma_hard}
Let $\Gamma$ be the class of the following conjunctive queries
\begin{equation}
	\gamma_k := x_1 \dots x_k \exists y_1\dots\exists y_k : \bigwedge_{i=1}^k Ex_iy_i \wedge \bigwedge_{1 \leq i < j \leq k} Ey_iy_j
\end{equation}

Recall that $\mathsf{G}$ is the set of all conjunctive queries over the signature of graphs without self-loops. Using Remark~\ref{rem:graph_sig} we can state Lemma~\ref{lem:Atwo_norm} in terms of $\#\homsprob$ as follows.

\begin{corollary}\label{cor:graphs_hard}
	The problem $\#\homsprob(\mathsf{G})$ is $\#\Atwo$-equivalent.
\end{corollary}
In Section~\ref{sec:minors_and_colors} we proved that $\#\cphomsprob(\Delta)$ reduces to $\#\homsprob(\Delta)$ whenever $\Delta$ contains only minimal queries. The next lemma states that the reverse direction holds unconditionally.
\begin{lemma} \label{lem:hom_to_phom}
	Let $(H,X)$ be a conjunctive query. Then there exists a algorithm~$\mathbb{A}$ with oracle access to $\#\cphoms{H,X}{\star}$ that computes $\#\homs{H,X}{\star}$. Furthermore $\mathbb{A}$ runs in time $O(f(|H,X|)\cdot n^c)$ for some computable function $f$ and some constant $c$ independent of $(H,X)$.
\end{lemma}
\begin{proof}
	Given a graph $G$ for which we want to compute $\#\homs{H,X}{G}$, we construct a $H$-colored graph $G'$ as follows:

	We first copy the vertex set $V(G)$ exactly $\#V(H)$ times and color the copies according to the vertices of $H$. If we write~$v(G')$ for the set of vertices in~$G'$ whose color is~$v\in V(G)$. Hence the vertices of $G'$ are partitioned into $\mathcal{P}=\{v(G')~|~v \in V(H)\}$.
	Now let $\{u,v\}$ be an edge of $H$. Then, for every $a \in u(H')$ and $b \in v(H')$ we add the edge $\{a,b\}$ to $G'$ if and only if (the initial vertices) $a$ and $b$ have been adjacent in $G$.
Note that this reduction yields indeed a $H$-colored graph, except for the case that $G$ contains no edge. In this case, however, we can compute $\#\homs{H,X}{G}$ by ``brute-force'' in linear time in $|V(G)|$. Otherwise, it can easily be verified that $\#\homs{H,X}{G}$ equals $\#\cphoms{H,X}{G'}$.
\end{proof}

The next lemma is required as a starting point for the reduction in Section~\ref{sec:hard_a_2}. The proof is a straight-forward application of the minor reduction of Section~\ref{sec:minors_and_colors}.

\begin{lemma}\label{lem:gamma_hard}
	$\#\cphomsprob(\Gamma)$ is $\#\Atwo$-hard.
\end{lemma}
\begin{proof}
	Given $\gamma_{2k}$, we can contract $y_i$ to $x_i$ for $i=1,\dots, k$ and then delete $x_{k+1},\dots,x_{2k}$. The resulting minor is the conjunctive query with $k$ free variables and $k$ quantified variables containing every edge between two vertices. It can easily be seen that every query with $\leq k$ free and $\leq k$ quantified variables is a minor of this query. Hence, by Lemma~\ref{lem:minor_closed} and Corollary~\ref{cor:graphs_hard}, we have that $\#\cphomsprob(\Gamma)$ is $\#\Atwo$-hard.
\end{proof}

\subsection{Hardness of counting answers to grates}\label{sec:grates_hard}

In what follows we provide the formal proof of Lemma~\ref{lem:grates_are_hard}.
\begin{lemma}[Lemma~\ref{lem:grates_are_hard} restated]\label{lem:grate_reduction}
  Let $\omega_k$ be the \kkite{k}.
  There is a Turing reduction from $\#\cphoms{\gamma_k}{\star}$ to $\#\cphoms{\omega_k}{\star}$ which runs in time $O(k^2\cdot n^2)$, where $n$ is the number of vertices of the input graph.
\end{lemma}
\begin{proof}
	We follow the lines of the $\#\W$-hardness proof of the Grid-Tiling problem.

	Let $(H_\gamma,X)$ be the query graph of $\gamma_k$. Without loss of generality we assume that the $k$ free variables of $\gamma_k$ are labeled $x^0_{k-1},x^1_{k-2},\dots,x^{k-1}_0$, that the $k$ quantified variables of $\gamma_k$ are labeled $y^0_{k-1},y^1_{k-2},\dots,y^{k-1}_0$ and that the atoms of $\gamma_k$ are of the form $Ey^i_jx^i_j$. Then it is well-defined to write $(H_\omega,X)$ for the query graph of $\omega_k$, i.e., the free variables coincide. Furthermore we have that the quantified variables of $\gamma_k$ are a subset of the quantified variables of $\omega_k$. In particular it holds that $V(H_\gamma) \subseteq V(H_\omega)$. While we give a formal construction of the reduction in the remainder of the proof, we encourage the reader to first consider the example in Figure~\ref{fig:kite_reduction} to get an intuition.
	\begin{figure}[ptb]
		\centering
		\begin{tikzpicture}[scale=0.5,-,thick]
			\node[circle,inner sep=1pt,fill,label={[label distance = -1mm]315:{\small $1$}}] (1) at (0,8) {};
			\node[circle,inner sep=1pt] (11) at (-0.75,8) {\small $\hdots$};
			\node[circle,inner sep=1pt,fill,label={[label distance = -1mm]315:{\small $2$}}] (2) at (0,4) {};
			\node[circle,inner sep=1pt] (22) at (-0.75,4) {\small $\hdots$};
			\node[circle,inner sep=1pt,fill,label={[label distance = -1mm]315:{\small $3$}}] (3) at (0,0) {};
			\node[circle,inner sep=1pt] (33) at (-0.75,0) {\small $\hdots$};
			\draw(1)--(2);\draw(3)--(2);\draw[bend right=40](1) to (3);

			\node[circle,inner sep=1pt,fill,label={[label distance = -1mm]315:{\small $\alpha$}}] (4) at (4,8) {};
			\node[circle,inner sep=1pt] (44) at (4.75,8) {\small $\hdots$};
			\node[circle,inner sep=1pt,fill,label={[label distance = -1mm]315:{\small $\beta$}}] (5) at (4,4) {};
			\node[circle,inner sep=1pt] (55) at (4.75,4) {\small $\hdots$};
			\node[circle,inner sep=1pt,fill,label={[label distance = -1mm]315:{\small $\gamma$}}] (6) at (4,0) {};
			\node[circle,inner sep=1pt] (66) at (4.75,0) {\small $\hdots$};
			\draw(1)--(4);\draw(5)--(2);\draw(6) to (3);

			\draw (0,0) circle (35pt);
			\draw (0,4) circle (35pt);
			\draw (0,8) circle (35pt);

			\draw (4,0) circle (35pt);
			\draw (4,4) circle (35pt);
			\draw (4,8) circle (35pt);

			\node[circle,inner sep=1pt,fill,label={[label distance = -2mm]315:{\tiny $(1,1)$}}] (7) at (18,8) {};
			\node[circle,inner sep=1pt] (77) at (17.25,8) {\small $\hdots$};
			\node[circle,inner sep=1pt,fill,label={[label distance = -2mm]315:{\tiny $(2,2)$}}] (8) at (18,4) {};
			\node[circle,inner sep=1pt] (88) at (17.25,4) {\small $\hdots$};
			\node[circle,inner sep=1pt,fill,label={[label distance = -2mm]315:{\tiny $(3,3)$}}] (9) at (18,0) {};
			\node[circle,inner sep=1pt] (99) at (17.25,0) {\small $\hdots$};

			\node[circle,inner sep=1pt,fill,label={[label distance = -2mm]135:{\tiny $(1,2)$}}] (13) at (14,6) {};
			\node[circle,inner sep=1pt] (1313) at (14.75,6) {\small $\hdots$};
			\node[circle,inner sep=1pt,fill,label={[label distance = -2mm]225:{\tiny $(2,3)$}}] (14) at (14,2) {};
			\node[circle,inner sep=1pt] (1414) at (14.75,2) {\small $\hdots$};
			\node[circle,inner sep=1pt,fill,label={[label distance = -2mm]135:{\tiny $(1,3)$}}] (15) at (10,4) {};
			\node[circle,inner sep=1pt] (1515) at (10.75,4) {\small $\hdots$};

			\node[circle,inner sep=1pt,fill,label={[label distance = -1mm]315:{\small $\alpha$}}] (10) at (22,8) {};
			\node[circle,inner sep=1pt] (1010) at (22.75,8) {\small $\hdots$};
			\node[circle,inner sep=1pt,fill,label={[label distance = -1mm]315:{\small $\beta$}}] (11) at (22,4) {};
			\node[circle,inner sep=1pt] (1111) at (22.75,4) {\small $\hdots$};
			\node[circle,inner sep=1pt,fill,label={[label distance = -1mm]315:{\small $\gamma$}}] (12) at (22,0) {};
			\node[circle,inner sep=1pt] (1212) at (22.75,0) {\small $\hdots$};

			\draw(7)--(13);\draw(7)--(10);\draw(15)--(13);\draw(15)--(14);\draw(8)--(11);
			\draw(8)--(13);\draw(8)--(14);\draw(9)--(14);\draw(9)--(12);

			\draw (10,4) circle (35pt);
			\draw (14,2) circle (35pt);
			\draw (14,6) circle (35pt);

			\draw (22,0) circle (35pt);
			\draw (22,4) circle (35pt);
			\draw (22,8) circle (35pt);

			\draw (18,0) circle (35pt);
			\draw (18,4) circle (35pt);
			\draw (18,8) circle (35pt);

			\node (100) at (1.4,-1.4) { $y^2_0(G)$};
			\node (100) at (1.4,2.6) { $y^1_1(G)$};
			\node (100) at (1.4,6.6) { $y^0_2(G)$};
			\node (100) at (5.4,-1.4) { $x^2_0(G)$};
			\node (100) at (5.4,2.6) { $x^1_1(G)$};
			\node (100) at (5.4,6.6) { $x^0_2(G)$};

			\node (100) at (19.4,-1.4) { $y^2_0(G')$};
			\node (100) at (19.4,2.6) { $y^1_1(G')$};
			\node (100) at (19.4,6.6) { $y^0_2(G')$};
			\node (100) at (23.4,-1.4) { $x^2_0(G')$};
			\node (100) at (23.4,2.6) { $x^1_1(G')$};
			\node (100) at (23.4,6.6) { $x^0_2(G')$};

			\node (100) at (14,7.75) { $y^0_1(G')$};
			\node (100) at (14,0.25) { $y^1_0(G')$};
			\node (100) at (10,5.75) { $y^0_0(G')$};

		\end{tikzpicture}
		\caption{\label{fig:kite_reduction} Illustration of the construction of $G'$ for $k=3$. The graph $G$ (\emph{left}) is $\gamma_3$-colored and the mapping $a= \{x^0_2 \mapsto \alpha, x^1_1 \mapsto \beta, x^2_0 \mapsto \gamma\}$ is contained in $\cphoms{H_\gamma, X}{G}$. The graph~$G'$ (\emph{right}) is $\omega_3$-colored and $a$ is contained in $\cphoms{H_\omega, X}{G}$ as well.}
	\end{figure}
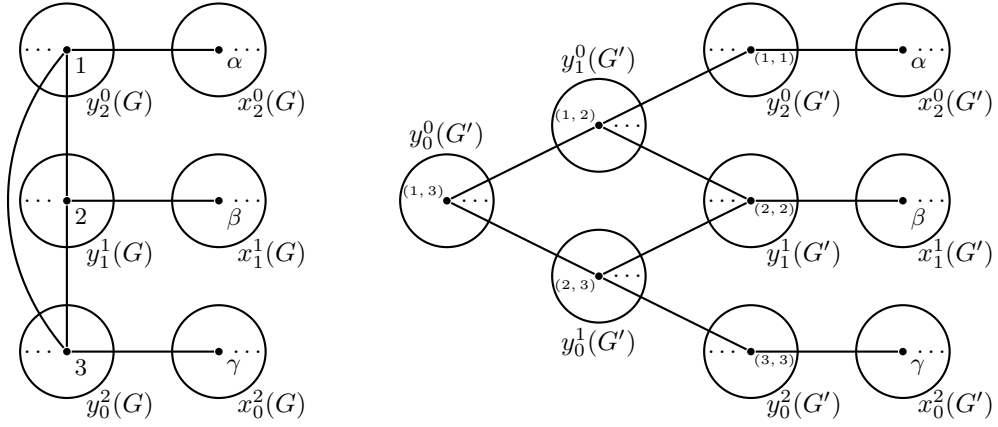

	Now let $G$ be a $\gamma_k$-colored graph for which we want to compute $\#\cphoms{H_\gamma, X}{G}$. Moreover we let $E^Y(G)$ denote the subset of edges of $G$ that connect two vertices colored with quantified variables of $\gamma_k$. We construct a $\omega_k$-colored graph $G'$ as follows:
	\begin{enumerate}
		\item For every pair $i,j$ with $i+j = k-1$ we add the vertices in $x^i_j(G)$ to $G'$ and preserve the color.
		\item For every pair $i,j$ with $i+j = k-1$ and for every vertex $u \in y^i_j(G)$ we add the vertex $(u,u)$ to $G'$ and preserve its color.
		\item For every edge $\{u,v\} \in E(G)$ such that $u$ is colored with $y^i_j$ and $v$ is colored with $x^{i'}_{j'}$ for some $i,j,i',j'$, we add the edge $\{(u,u),v\}$ to $G'$.
		\item For every pair $i,j$ with $i+j < k-1$ we add vertices $\{(u,u') ~|~\{u,u'\} \in E^Y(G) \}$ and color them with $y^i_j$. Note that this yields two vertices $(u,u')$ and $(u',u)$ for every edge $\{u,u'\}$ between vertices colored with quantified variables  in $G$.
		\item For every pair $i,j$ with $i+j < k$ and $j < k-1$, we add an edge between $(u,u')\in y^i_j(G')$ and $(v,v') \in y^i_{j+1}(G')$ if and only if $u=v$.
		\item Similarly, for every pair $i,j$ with $i+j < k$ and $i < k-1$, we add an edge between $(u,u')\in y^i_j(G')$ and $(v,v') \in y^{i+1}_{j}(G')$ if and only if $u'=v'$.
	\end{enumerate}
	
	\noindent Now $G'$ is $\omega_k$-colored and we claim that
	\begin{equation*}
	\cphoms{H_\gamma, X}{G} = \cphoms{H_\omega, X}{G'}\,.
	\end{equation*}
	For the first direction let $a : X \rightarrow V(G)$ in $\cphoms{H_\gamma, X}{G}$. Write $v(G)$ for the set of vertices in~$G$ that have color~$v$. As $a$ can be extended to a homomorphism $h$ that satisfies $h(v)\in v(G)$ for all $v \in V(H_M)$, it holds that $a(x^i_j) \in x^i_j(G)$ for every pair $i,j$ with $i+j = k-1$. Now let $u^i_j=h(y^i_j)$ for every pair $i,j$ with $i+j = k-1$. We construct a homomorphism $h':V(H_\omega) \rightarrow V(G')$ as follows: (1) $h'$ coincides with $h$ (or $a$, respectively) on $X$, (2) $h'(y^i_j) = (u^i_j,u^i_j)$ for every pair $i,j$ with $i+j = k-1$ and (3) $h'(y^i_j) = (h(y^i_{j+1}),h(y^{i+1}_j))$ for every pair $i,j$ with $i+j < k-1$. Note that Step~(3) is well-defined as the image of vertices in $H_\gamma$ corresponding to quantified variables is a clique (of size $k$) in $G$, because otherwise $h$ would be no homomorphism. By construction of $G'$ it holds that $h'$ is a homomorphism satisfying $h'(v) \in v(G')$ for every $v \in V(H_\omega)$. Hence $a \in \cphoms{H_\omega, X}{G'}$.

	For the other direction let $a: X \rightarrow V(G')$ in $\cphoms{H_\omega, X}{G'}$ and let $h': V(H_\omega) \rightarrow V(G')$ be the homomorphism extending $a$ that satisfies $h'(v) \in v(G')$ for every $v \in V(H_\omega)$. Now let $(u^i_j,u^i_j)=h'(y^i_j)$ for every pair $i,j$ with $i+j = k-1$. We construct a homomorphism $h:V(H_\gamma) \rightarrow V(G)$ as follows: (1) $h$ coincides with $h'$ (or $a$, respectively) on $X$ and (2) $h(y^i_j) = u^i_j$ for every pair $i,j$ with $i+j = k-1$. Now it can easily be seen that $h$ is indeed a homomorphism: There must be an edge between every pair of vertices $h(y^i_j)$ and $h(y^{i'}_{j'})$ for $i+j=i'+j'=k-1$ as otherwise there would be no path through the half-grid in $G'$ connecting $(u^i_j,u^i_j)$ and $(u^{i'}_{j'},u^{i'}_{j'})$. Hence $a \in \cphoms{H_\gamma, X}{G}$.

	We conclude that $\mathbb{A}$ just constructs $G'$ from $G$, which takes time $O(k^2\cdot n^2)$, and then queries the oracle for $G'$.
\end{proof}

\section{Conclusions}
We established a comprehensive classification of the complexity of counting answers to conjunctive queries and linear combinations thereof. Depending on the structural parameters of the class of allowed queries, the problem is either fixed-parameter tractable, $\W$-equivalent, $\#\W$-equivalent, $\#\Wtwo$-hard or $\#\Atwo$-equivalent. This classification, however, leaves out a gap between the latter two cases. 
More precisely, the following question remains open.
\begin{center}
\textit{Does a class of conjunctive queries $\Delta$ exist for which $\#\homsprob(\Delta)$ is $\#\Wtwo$-hard\newline but neither equivalent for $\#\Wtwo$ nor for $\#\Atwo$?}
\end{center}
We conjecture a positive answer; the interested reader is encouraged to make themself familiar with the parameterized complexity class $\mathsf{W}_{\mathsf{func}}[2]$ (see e.g.\ \cite[Chapter~8.8]{flumgrohe}). This class has a canonical counting version which we call $\#\mathsf{W}_{\mathsf{func}}[2]$ and which interpolates between $\#\Wtwo$ and $\#\Atwo$. In particular, we conjecture that there exists a class of conjunctive queries $\Delta$ for which $\#\homsprob(\Delta)$ is $\#\mathsf{W}_{\mathsf{func}}[2]$-equivalent. Consequently, a negative answer to the previous question would imply that either $\#\mathsf{W}_{\mathsf{func}}[2] = \#\Wtwo$ or $\#\mathsf{W}_{\mathsf{func}}[2] = \#\Atwo$, which seems to be very unlikely (see e.g.\ the discussion of $\mathsf{W}_{\mathsf{func}}[2]$ in~\cite[Chapter~8.8]{flumgrohe}).

A further question that remains open, and which should be considered a stronger version of the previous question, reads as follows:
\begin{center}
\textit{Does a class of conjunctive queries $\Delta$ exist such that $\Delta$ has bounded linked matching number and the problem $\#\homsprob(\Delta)$ is $\#\Atwo$-equivalent?}
\end{center}
In other words, the above question asks whether the absence of a bound on the linked matching number is not only sufficient, but also necessary for $\#\Atwo$-equivalence. In contrast to the previous question, we conjecture a negative answer. Let us provide some intuition for the latter conjecture: It seems that a constant bound on the linked matching number of a class of conjunctive queries $\Delta$ yields a separator decomposition of the quantified variables of queries in $\Delta$ in components that have either small treewidth or a small matching number to the free variables. We conjecture that such a decomposition implies the existence of what is called a $\kappa$-restricted nondeterministic Turing machine $\mathbb{M}$ such that the number of accepting paths of $\mathbb{M}$ on input $(H,X)\in \Delta$ and a hypergraph $\mcG$ is precisely $\#\homs{H,X}{\mcG}$ (see e.g.\ ~\cite[Definition 14.15]{flumgrohe}). If additionally $\#\homsprob(\Delta)$ is $\#\Atwo$-equivalent, this would imply that the set of $\#\Atwo$-equivalent problems is a subset of the set of $\#\mathsf{W[P]}$-equivalent problems; consult e.g.\ ~\cite[Chapter~3 and~14.2]{flumgrohe} for a treatment of the class $\#\mathsf{W[P]}$. However, the latter inclusion seems to be unlikely and we refer the interested reader to~\cite[Chapter~8]{flumgrohe} for a detailed treatment of the corresponding question whether $\Atwo \subseteq \mathsf{W[P]}$ in the decision world. We conclude with the remark that even a proof of $\Atwo \subseteq \#\mathsf{W[P]}$ would be a major breakthrough as it constitutes the first step of a parameterized analogue of Toda's theorem~\cite{Toda91}, which is one of the fundamental open problems in (structural) parameterized counting complexity.

\bibliographystyle{plainurl}
\bibliography{databasemotifs}

\end{document}